\documentclass[10pt]{raqis}
\usepackage[T1]{fontenc}
\usepackage[applemac]{inputenc}
\usepackage{amsmath,amssymb,amsfonts,amscd}
\usepackage{graphicx,epsfig}
\usepackage{theorem,cite}
\usepackage{multicol}
\usepackage[enableskew]{youngtab}
\usepackage{euscript,enumerate,cancel}

\input{tcilatex}

\usepackage[english]{minitoc}
\theoremstyle{plain}
\theorembodyfont{\rmfamily}
\numberwithin{theorem}{section}
\numberwithin{definition}{section}
\numberwithin{equation}{section}
\numberwithin{figure}{section}
\numberwithin{remark}{section}
\numberwithin{example}{section}

\input{tcilatex}

\begin{document}
\setcounter{MaxMatrixCols}{10}
\def\mtctitle{}
\normalsize
\baselineskip=16pt
\pagestyle{myheadings}
\pagenumbering{arabic}
\setcounter{page}{1}
\setcounter{tocdepth}{2}
\setcounter{minitocdepth}{1}
\setcounter{secnumdepth}{3}

\chapter{\Large Cylindric versions of specialised Macdonald functions and\\ a deformed Verlinde
algebra}
\markboth{A deformed Verlinde algebra}{Christian Korff}

\begin{center}

{\large Christian Korff}

\vspace*{5mm}

\textit{School of Mathematics and Statistics, University of Glasgow, \\
15 University Gardens, Glasgow G12 8QW, Scotland, UK\\
E-mail: christian.korff@glasgow.ac.uk}
\end{center}

\vfill

\begin{abstract}
We define cylindric generalisations of skew Macdonald functions $P_{\lambda
/\mu }(q,t)$ when either $q=0$ or $t=0$. Fixing two integers $n>2$ and $k>0$
we shift the skew diagram $\lambda/\mu$, viewed as a subset of the two-dimensional integer
lattice, by the period vector $(n,-k)$. Imposing a periodicity condition one
defines cylindric skew tableaux and associated weight functions. The
resulting weighted sums over these cylindric tableaux are symmetric
functions. They appear in the coproduct of a commutative Frobenius algebra which is
a particular quotient of the spherical Hall algebra. We realise this Frobenius algebra as
a commutative subalgebra in the endomorphisms over a $U_{q}\widehat{\mathfrak{sl}}(n)$ %
Kirillov-Reshetikhin module. Acting with special elements of this subalgebra, which
are noncommutative analogues of Macdonald polynomials, on a highest weight vector, one obtains Lusztig's canonical basis. In the limit $q=t=0$, this Frobenius algebra is isomorphic to the $\widehat{\mathfrak{sl}}(n)$ Verlinde
algebra at level $k$, i.e. the structure constants become the $\widehat{\mathfrak{sl}}(n)_{k}$ Wess-Zumino-Novikov-Witten fusion coefficients.
Further motivation comes from exactly solvable lattice models in statistical
mechanics: the cylindric Macdonald functions discussed here arise as
partition functions of so-called vertex models obtained from solutions to
the Yang-Baxter equation. We show this by stating explicit bijections between cylindric tableaux and lattice configurations of non-intersecting paths. Using the algebraic Bethe ansatz the idempotents of the Frobenius algebra are computed.
\end{abstract}

\vfill

\vfill

\clearpage
\parindent=0pt

\section{Introduction}

The fusion or Verlinde ring $\mathcal{V}_{k}$, $k\in \mathbb{Z}_{\geq 0}$ of
a Kac Moody algebra $\mathfrak{\hat{g}}$ is a particular finite-dimensional
quotient of the Grothendieck ring $\limfunc{Rep}\mathfrak{g}%
=\tbigoplus_{\lambda \in \mathcal{P}^{+}}\mathbb{Z[\pi }_{\lambda }\mathbb{]}
$ (with product $\otimes $), where $\mathfrak{g}$ is
the corresponding (non-affine) semi-simple Lie algebra, $\mathcal{P}^{+}$ is
the set of dominant integral weights and $[\pi _{\lambda }]$ stands for the
isomorphism class of the irreducible representation $\pi _{\lambda }$ with
highest weight $\lambda $. Given a non-negative
integer $k$ define $\mathfrak{I}_{k}$ to be the ideal generated by elements of the form $[\pi _{\lambda
}]-(-1)^{\ell (w)}[\pi _{w\circ \lambda }]$ where $w$ is an element in the
affine Weyl group $\tilde{W}$ and $w\circ \lambda = w(\lambda+\rho)-\rho$ denotes the non-affine
part of the weight obtained under the shifted level-$k$ action of $\tilde{W}$ %
with $\rho$ being the Weyl vector. For instance, in the case of $\mathfrak{g}=\mathfrak{sl}(n)$ the action
 of the simple Weyl reflections is detailed in
equation (\ref{affaction}) in the text.
The Verlinde algebra is then defined as $\mathcal{V}_{k}:=\limfunc{Rep}\mathfrak{g}/\mathfrak{I}_{k}$; this
is in essence the celebrated Kac-Walton formula \cite{Kac} \cite{Walton}. The structure
constants of $\mathcal{V}_{k}$ are known to coincide with the fusion
coefficients in Wess-Zumino-Novikov-Witten (WZNW) conformal field theory, dimension
of moduli spaces of generalised $\theta $-functions and multiplicities of
tilting modules of quantum groups at roots of unity.

\subsection{Review of previous results}

In this article we will only consider the simplest case when $\mathfrak{g}=\mathfrak{sl}(n)$ or $\mathfrak{gl}(n)$ with Weyl group $W=\mathfrak{S}_{n}$, the symmetric group. There is a
ring isomorphism $\chi :\limfunc{Rep}\mathfrak{gl}(n)\rightarrow \mathbb{Z}%
[x_{1},\ldots ,x_{n}]^{\mathfrak{S}_{n}}$ which maps each isomorphism class $%
[\pi _{\lambda }]$ onto its Weyl character which can be identified with the
Schur function $s_{\lambda }$ and we have $s_{\mu }s_{\nu }=\sum_{\lambda
\in \mathcal{P}_{n}^{+}}c_{\mu \nu }^{\lambda }s_{\lambda }$ with $c_{\mu
\nu }^{\lambda }$ being the famous Littlewood-Richardson coefficients. In
what follows it will be important to note that the ring of symmetric
functions $\mathbb{Z}[x_{1},\ldots ,x_{n}]^{\mathfrak{S}_{n}}$ can be turned
into an infinite-dimensional bialgebra \cite{Gei} \cite{Z} (which for convenience we
define over $\mathbb{C}$) with coproduct $\Delta s_{\lambda }=\sum_{\mu \in
\mathcal{P}^{+}}s_{\lambda /\mu }\otimes s_{\mu }$ where $s_{\lambda /\mu
}=\sum_{\nu \in \mathcal{P}_{n}^{+}}c_{\mu \nu }^{\lambda }s_{\nu }$ is the
skew Schur function.\smallskip

\noindent Following \cite{Gepner}, \cite{GoodmanWenzl} define $\mathcal{I}%
_{k}=\langle s_{(1^{n})}-1,s_{(k+1)},s_{(k+2)},\ldots ,s_{(k+n-1)}\rangle $.

\begin{theorem}[Gepner, Goodman-Wenzl]
\label{GGW}The map $[\pi _{\lambda }]\mapsto \lbrack s_{\lambda
}]:=s_{\lambda }+\mathcal{I}_{k}$ defines a ring isomorphism $\chi _{k}:%
\mathcal{V}_{k}\rightarrow \mathbb{Z}[x_{1},\ldots ,x_{n}]^{\mathfrak{S}%
_{n}}/\mathcal{I}_{k}$.
\end{theorem}

Define a non-degenerate bilinear form $\eta ([\pi _{\lambda }],[\pi _{\mu
}])=\delta _{\lambda \mu ^{\ast }}$ on the Verlinde algebra $\mathcal{V}_{k}\otimes _{\mathbb{Z}}\mathbb{%
C}$, where $\lambda ,\mu \in \mathcal{P}_{n,k}^{+}$ are dominant integral weights at level $k$ and $\mu ^{\ast }$
denotes the contragredient weight of $\mu$. Then $(\mathcal{V}_{k}\otimes _{\mathbb{Z}}\mathbb{%
C},\eta)$ is a (finite-dimensional) commutative Frobenius algebra.
This fact is not often mentioned in the literature, but it will motivate our definition of a deformed Verlinde
algebra below.

In \cite{KS} it was shown that there exists an alternative, combinatorial
description of the Verlinde algebra which employs a local affine version of
the plactic algebra \cite{LasSchutz} in the Robinson-Schensted-Knuth
correspondence. Each dominant integral weight $\lambda \in \mathcal{P}%
_{n,k}^{+}\}$ at level $k$ corresponds to a unique composition $m(\lambda )=(m_{1},\ldots
,m_{n})\in \mathbb{Z}_{\geq 0}^{n}$ where $m_{i}=m_{i}(\lambda )$ is the
multiplicity of the part $i$ in the conjugate partition $\lambda ^{\prime }$
for $i=1,\ldots ,n-1$ and $k=\sum_{i=1}^{n}m_{i}$. Interpret each $m(\lambda
)$ as a particle configuration on the $\widehat{\mathfrak{sl}}(n)$ Dynkin
diagram where $m_{i}(\lambda )$ particles are sitting at the $i$th node; see
Figure \ref{fig:hoppingalg} for a simple example. Define maps $\beta
_{i}^{\ast},\beta_i:\mathcal{P}_{n,k}^{+}\rightarrow \mathcal{P}_{n,k\pm 1}^{+}$
which increase and decrease the number of particles at node $i$ by one, respectively. The
\emph{affine plactic algebra} is then generated by the maps $a_{i}=\beta
_{i+1}^{\ast}\beta _{i}$ which move one particle from node $i$ to node $i+1$
with $i\in \mathbb{Z}_{n}$. The directed coloured graph obtained from
setting $\lambda \overset{i}{\longrightarrow }\mu $ if $a_{i}m(\lambda
)=m(\mu )$ matches the Kirillov-Reshetikhin crystal graph $B^{1,k}$ of the quantum enveloping algebra $%
U_{q}\widehat{\mathfrak{sl}}(n)$. Define the \emph{affine plactic
Schur polynomial} \cite{KS, Korff} as $\boldsymbol{s}_{\lambda }:=\det (\boldsymbol{h%
}_{\lambda _{i}-i+j})_{1\leq i,j\leq n}$ with $\boldsymbol{h}_{r}=\sum_{\mu \vdash r}(\beta
_{1}^{\ast})^{\mu _{n}}a_{1}^{\mu _{1}}\cdots a_{n-1}^{\mu _{n-1}}\beta
_{n}^{\mu _{n}}$. The polynomial $\boldsymbol{s}_{\lambda }$ is well-defined, since one can show that
$\boldsymbol{h}_{r}\boldsymbol{h}_{r'}=\boldsymbol{h}_{r'}\boldsymbol{h}_{r}$ for all $r,r'\in\mathbb{Z}_{\geq 0}$.

\begin{theorem}[Korff-Stroppel]
\label{KSThm0}Consider the free abelian group $\mathbb{Z}\mathcal{P}%
_{n,k}^{+}$ with respect to addition. Introduce the product $\lambda
\circledast \mu :=\boldsymbol{s}_{\lambda }\mu $, then $(\mathbb{Z}\mathcal{P%
}_{n,k}^{+},\circledast )$ is canonically isomorphic to the Verlinde ring $\mathcal{V}%
_{k}$.
\end{theorem}

Note that $\boldsymbol{s}_{\lambda }$ specialises to the finite, non-affine
plactic Schur polynomial of Fomin and Greene \cite{FG} when setting formally
$a_{n}=0$; see also the construction of noncommutative Schur's $P,Q$-functions
using a shifted plactic monoid in \cite{Serrano}. The combinatorics of these constructions
is less involved than the one of the affine polynomials. Other approaches to
noncommutative symmetric functions can be found in \cite{Gelfandetal} and,
in particular, noncommutative Hall-Littlewood functions have been discussed
in \cite{Hivert}, \cite{Novellietal}.
\begin{figure}[tbp]
\begin{equation*}
\includegraphics[scale=0.4]{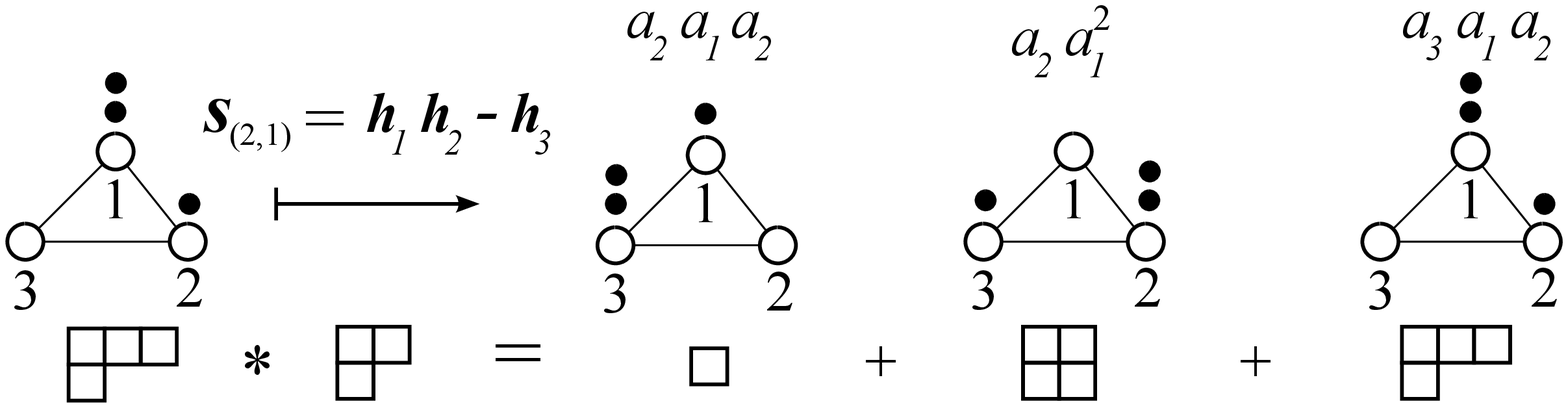}
\end{equation*}%
\caption{Graphical depiction of the combinatorial fusion product for $n=k=3$%
. The start configuration corresponds to the weight $\protect\lambda =2%
\protect\omega _{1}+\protect\omega _{2}$. Acting with the affine plactic
Schur polynomial $\boldsymbol{s}_{(2,1)}$ yields three nonzero terms; each
is obtained as a sequence of 3 hopping moves written as monomials in the $%
a_{i}$'s (top). The bottom line shows the fusion product where columns of
height 3 have been dropped from the Young diagrams which represent the
weights.}
\label{fig:hoppingalg}
\end{figure}
\smallskip

\noindent Define the two-sided ideal $\mathcal{I}_{k}^{\prime }=\langle
s_{(n)}-1,s_{(n+1)},s_{(n+2)},\ldots ,s_{(n+k-1)}\rangle $.

\begin{theorem}[Korff-Stroppel]
\label{KSThm1} The map $[\pi _{\lambda }]\mapsto \lbrack s_{\lambda ^{\prime
}}]:=s_{\lambda ^{\prime }}+\mathcal{I}_{k}^{\prime }$, where $\lambda
^{\prime }$ is the conjugate partition, defines a
ring isomorphism $\chi _{k}^{\prime }:\mathcal{V}_{k}\rightarrow \mathbb{Z}%
[x_{1},\ldots ,x_{k}]^{\mathfrak{S}_{k}}/\mathcal{I}_{k}^{\prime }$.
\end{theorem}

This result \cite[Theorem 1.3]{KS} is intimately linked to a quantum
integrable system, the so-called phase model, which has been considered in
\cite{Bogoliubovetal}. Interpret the complex linear span $\mathbb{C}\mathcal{%
P}_{n,k}^{+}$ as the state space of a discrete quantum mechanical system and
$H_{r}^{\pm }=\boldsymbol{s}_{(r)}\pm \boldsymbol{s}_{(n-r)}$ as its set of
commuting quantum Hamiltonians. Computing the common eigenbasis of the
latter via the so-called \emph{Bethe ansatz} or the \emph{quantum inverse
scattering method} leads to the coordinate ring of a finite 0-dimensional
affine variety, the solutions to the so-called Bethe ansatz equations. This
coordinate ring is the quotient ring $\mathbb{Z}[x_{1},\ldots ,x_{k}]^{%
\mathfrak{S}_{k}}/\mathcal{I}_{k}^{\prime }$ and the so-called Bethe states,
the eigenstates of the quantum Hamiltonians, are the idempotents of the
Verlinde algebra \cite{KS} \cite{Korffproceedings}.

\subsection{Deformed fusion matrices and canonical bases}

The main result of this article is that the combinatorial description of the
Verlinde algebra in terms of affine plactic Schur polynomials and
Kirillov-Reshetikhin crystal graphs can be lifted to the quantum affine algebra $%
U_{q}^{\prime }\widehat{\mathfrak{gl}}(n)$ using Lusztig's canonical basis \cite{Lusztig} \cite%
{LusztigBook}, which is the same as Kashiwara's lower crystal basis \cite%
{Kashiwara}. This induces in a natural way a product on the corresponding Kirillov-Reshetikhin
module which then becomes a commutative Frobenius algebra.

The central algebraic object
which we are going to employ is the $n$-fold tensor product $\mathcal{H}%
_{q}^{\otimes n}$ of the $q$-oscillator or Heisenberg algebra $\mathcal{H}%
_{q}$ whose generators will be realised as maps $\beta _{i}^{\ast},\beta_i:\mathbb{C}%
(q)\mathcal{P}_{n,k}^{+}\rightarrow \mathbb{C}(q)\mathcal{P}_{n,k\pm 1}^{+}$ %
which generalise the maps mentioned previously in the contaxt of the combinatorial Verlinde algebra.
We will show that there exists an algebra homomorphism $U_{q}^{\prime }%
\widehat{\mathfrak{gl}}(n)\rightarrow \mathcal{H}_{q}^{\otimes n}$ which
allows one to pull the $n$-fold tensor product of any $\mathcal{H}_{q}$%
-module back\ to the quantum affine algebra. The linear span of the particle configurations on
the $\widehat{\mathfrak{sl}}(n)$ Dynkin diagram discussed earlier corresponds
to the infinite-dimensional highest weight module known as Fock space $%
\mathcal{F}^{\otimes n}\cong \tbigoplus_{k\geq 0}\mathbb{C}(q)\mathcal{P}%
_{n,k}^{+}$. Denote by $S^{k}(V)$ the $k^{\text{th}}$ divided power in the
quantum symmetric tensor algebra of the vector representation $V$ of $U_{q}%
\mathfrak{gl}(n)$ and let $\omega _{1}$ be the first fundamental affine
weight of $\widehat{\mathfrak{sl}}(n)$.

\begin{proposition}
There exists a $U_{q}^{\prime }\widehat{\mathfrak{sl}}(n)$-module
isomorphism $\mathcal{F}^{\otimes n}\cong \tbigoplus_{k\geq 0}W^{1,k}$,
where $W^{1,k}$ is the Kirillov-Reshetikhin module $W(k\omega _{1})$. When
restricting to the finite algebra $U_{q}\mathfrak{gl}(n)$ one obtains the
module isomorphism $\mathcal{F}^{\otimes n}\cong S(V):=\tbigoplus_{k\geq
0}S^{k}(V)$.
\end{proposition}

Similar to the case of the Verlinde algebra, we consider the (noncommutative) subalgebra $\subset
\mathcal{H}_{q}^{\otimes n}$ which is generated by the alphabet $%
\{a_{i}=\beta _{i+1}^{\ast}\beta _{i}:i\in \mathbb{Z}_{n}\}$. The latter
corresponds to the images of the products $\{K_{1}F_{1},\ldots ,K_{n}F_{n}\}$
of quantum group Chevalley generators under the above homomorphism $U_{q}%
\widehat{\mathfrak{gl}}(n)\rightarrow \mathcal{H}_{q}^{\otimes n}$. In
particular they obey
\begin{eqnarray}
a_{i}a_{j} &=&a_{j}a_{i},\qquad |i-j|>1,  \notag \\
a_{i+1}a_{i}^{2}+q^{2}a_{i}^{2}a_{i+1} &=&(1+q^{2})a_{i}a_{i+1}a_{i},  \notag
\\
a_{i+1}^{2}a_{i}+q^{2}a_{i}a_{i+1}^{2}
&=&(1+q^{2})a_{i+1}a_{i}a_{i+1}\;,\qquad i,j\in \mathbb{Z}_{n}\;.
\label{quantumKnuth}
\end{eqnarray}%
These identities are simply the quantum Serre relations of $U_{q}\widehat{%
\mathfrak{gl}}(n)$ rewritten in the generators $K_{i}F_{i}$. In the
crystal limit $q=0$ (\ref{quantumKnuth}) are the Knuth relations of the (local) affine
plactic algebra considered in \cite{KS} and we have now the following
generalisation of the affine plactic Schur polynomials.

Let $|k^{n}\rangle $ denote the $U_{q}\widehat{%
\mathfrak{sl}}(n)$ highest weight vector in $S^{k}(V)$ and
denote by $B_{n,k}=\{|\lambda \rangle :\lambda \in \mathcal{P}%
_{n,k}^{+}\}\subset S^{k}(V)$ the canonical basis. The precise definition
of the basis vectors $|\lambda \rangle$ will be given in the text.

\begin{theorem}[deformed fusion matrices]
There exists a set of \emph{commuting }elements $\boldsymbol{B}_{n}:=\{%
\boldsymbol{Q}_{\lambda }^{\prime }:\lambda _{1}\geq \cdots \geq \lambda
_{n},\;\lambda _{i}\in \mathbb{Z}_{\geq 0}\}\subset \mathcal{H}%
_{q}^{\otimes n}$, polynomial in the $a_{i}$'s, such that%
\begin{equation}
|\mu \rangle \circledast |\nu \rangle :=\boldsymbol{Q}_{\mu }^{\prime }|\nu
\rangle \;,\qquad \mu ,\nu \in \mathcal{P}_{n,k}^{+}  \label{tfusion0}
\end{equation}%
defines a commutative Frobenius algebra $\mathfrak{F}_{n,k}=(\mathbb{C}(q)%
\mathcal{P}_{n,k}^{+},\circledast )$. The unit is given by the highest weight vector $|k^n\rangle$, that is $\boldsymbol{Q}%
_{\lambda }^{\prime }|k^{n}\rangle =|\lambda \rangle $.
\end{theorem}

Thus, the set $\boldsymbol{B}_{n}\subset \mathcal{H}_{q}^{\otimes n}$
generates the canonical basis in each $S^{k}(V)$ when acting on the
respective highest weight vector. The polynomials $\boldsymbol{Q}_{\lambda
}^{\prime }$ in the generators $a_{i}$ exhibit a particularly nice structure
which allows one to identify them as noncommutative analogues of Macdonald
functions, where one parameter is set to zero. We refer to $\{\boldsymbol{Q}%
_{\lambda }^{\prime }\}$ as deformed fusion matrices since setting formally $%
q=0$ in (\ref{quantumKnuth}) one recovers the combinatorial ring in Theorem %
\ref{KSThm0}.

To put these findings further into perspective we recall that Lusztig's
geometric construction of the canonical basis $\mathfrak{B}%
\subset U_{q}\mathfrak{n}^{_{-}}$ focusses on polynomials in the $F_{i}$'s,
that is, one considers $U_{q}\mathfrak{n}^{_{-}}$ instead of $U_{q}\mathfrak{%
b}^{-}$. Certain special elements $X,Y$ in the \emph{dual} canonical basis $%
\mathfrak{B}^{\ast }$\ are known to quasi-commute, $XY=qYX$; this was
conjectured in \cite{BZ} and proven for semi-simple quantum algebras in \cite%
{Reineke} using Ringel's Hall algebra approach. For $U_{q}\widehat{\mathfrak{%
gl}}(n)$ with $q$ a root of unity the canonical basis is known to be linked
to the Ringel-Hall algebra of the cyclic quiver \cite{Schiffmann}.

\subsection{The deformed Verlinde algebra: Demazure characters}

Denote by $Q_{\lambda }(q,t)=b_{\lambda }(q,t)P_{\lambda }(q,t)$ the
celebrated Macdonald functions, where $b_{\lambda }(q,t)$ is some
normalisation factor; details will be provided in the text. Consider the
limit $P_{\lambda }:=P_{\lambda }(0,t)$, $Q_{\lambda }:=Q_{\lambda }(0,t)\in
\mathbb{C}(t)[x_{1},\ldots ,x_{k}]^{\mathfrak{S}_{k}}$ which are the
celebrated Hall-Littlewood functions. Then one has the product expansion
\begin{equation}
P_{\mu }P_{\nu }=\sum_{\lambda \in \mathcal{P}_{k}^{+}}f_{\mu \nu }^{\lambda
}(t)P_{\lambda },  \label{Hall}
\end{equation}%
where the $f_{\mu \nu }^{\lambda }(t)$ are the structure constants of Hall's
algebra or the spherical Hecke algebra; see \cite{Macdonald} and references therein. \smallskip

\noindent Define the two-sided ideal
\begin{equation*}
\mathcal{I}_{k}^{\prime }=\langle Q_{(n)}+t^{k}-1,Q_{(n+1)}+t^{k}\bar{Q}%
_{1},\ldots ,Q_{(n+k-1)}+t^{k}\bar{Q}_{(k-1)}\rangle ,
\end{equation*}%
where $\bar{Q}_{\lambda }=Q_{\lambda }(0,t^{-1})$ and let $\Bbbk =\mathbb{C}%
\{\!\{t\}\!\}$ be the field of formal Puiseux series. The extension of the
base field to $\Bbbk $ is required to construct the idempotents of $%
\mathfrak{F}_{n,k}$. The following statement is the analogue of Theorem \ref%
{KSThm1}.

\begin{theorem}[deformed Verlinde algebra]
For $t=q^{2}$ the map $|\lambda \rangle \mapsto \lbrack P_{\lambda ^{\prime
}}]$ is an algebra isomorphism $\mathfrak{F}_{n,k}\otimes \Bbbk \cong \Bbbk
\lbrack x_{1},\ldots ,x_{k}]^{\mathfrak{S}_{k}}/\mathcal{I}_{k}^{\prime }$.
\end{theorem}

The analogue of Theorem \ref{GGW} for $\mathfrak{F}_{n,k}$ is currently
missing. However, it is more natural to consider the deformed fusion product
(\ref{tfusion0}) as a modification of the product
\begin{equation}
Q_{\mu }^{\prime }Q_{\nu }^{\prime }=\sum_{\lambda \in \mathcal{P}%
_{n}^{+}}f_{\mu ^{\prime }\nu ^{\prime }}^{\lambda ^{\prime }}(t)Q_{\lambda
^{\prime }}^{\prime },  \label{Demazureprod}
\end{equation}%
where $Q_{\lambda }^{\prime }:=Q_{\lambda }(t,0)\in \mathbb{C}%
(t)[x_{1},\ldots ,x_{n}]^{\mathfrak{S}_{n}}$ is now the complementary limit
of Macdonald functions and $\lambda ^{\prime }$ denotes the conjugate partition.
In the projective limit of $\infty$-many variables there exists a
bialgebra automorphism $\omega _{t}:P_{\lambda }(0,t)\mapsto Q_{\lambda
^{\prime }}(t,0)$ which is simply the known duality relation of Macdonald
functions \cite[VI.5, Equation (5.1), p327]{Macdonald} when one of the
parameters is set to zero. The dual Macdonald function $P_{\lambda }^{\prime
}$ can be identified with characters of so-called Demazure modules
related to $\widehat{\mathfrak{sl}}(n)$ \cite{Sanderson} (see also \cite%
{Ion} for other Kac-Moody algebras) and \cite{KMOTU}. Define a modified product%
\begin{equation}
Q_{\mu }^{\prime }\ast Q_{\nu }^{\prime }:=\sum_{\lambda \in \mathcal{P}%
_{n,k}^{+}}N_{\mu ^{\prime }\nu ^{\prime }}^{\lambda ^{\prime }}(t)Q_{\nu
}^{\prime },  \label{Demazurefusion}
\end{equation}%
where $N_{\mu ^{\prime }\nu ^{\prime }}^{\lambda ^{\prime }}(t)$ is defined through $[P_{\mu ^{\prime }}P_{\nu ^{\prime
}}]=\sum_{\lambda \in \mathcal{P}_{n,k}^{+}}N_{\mu ^{\prime }\nu ^{\prime
}}^{\lambda ^{\prime }}(t)[P_{\lambda ^{\prime }}]$ in the quotient $\Bbbk
\lbrack x_{1},\ldots ,x_{k}]^{\mathfrak{S}_{k}}/\mathcal{I}_{k}^{\prime }$.
The coproduct of the resulting Frobenius algebra leads to a cylindric
generalisation $P_{\lambda /d/\mu }^{\prime }=\sum_{\lambda \in \mathcal{P}%
_{n,k}^{+}}N_{\mu ^{\prime }\nu ^{\prime }}^{\lambda ^{\prime }}(t)P_{\nu
}^{\prime }$ of the skew Macdonald function $P_{\lambda /\mu }^{\prime
}:=P_{\lambda /\mu }(t,0)=\sum_{\nu}f^{\lambda'}_{\mu'\nu'}(t)P'_\nu$. We will define $P_{\lambda /d/\mu }^{\prime }$
explicitly as weighted sum over cylindric Young tableaux (also called
reversed cylindric plane partitions). The latter were first considered in
\cite{GesselKrattenthaler} and are maps $T:\lambda /d/\mu \rightarrow
\mathbb{N}$, where $\lambda /d/\mu $ denotes a cylindric skew diagram which
can be seen as set of points in $\mathbb{Z}^{2}$ obtained by a periodic
continuation of the ordinary skew diagram $(\lambda _{1}+d,\ldots ,\lambda
_{n}+d)/\mu $ with respect to the period vector $\Omega =(n,-k)$; see Figure %
\ref{fig:cyltabintro} for an example.
\begin{figure}[tbp]
\begin{equation*}
\includegraphics[scale=0.4]{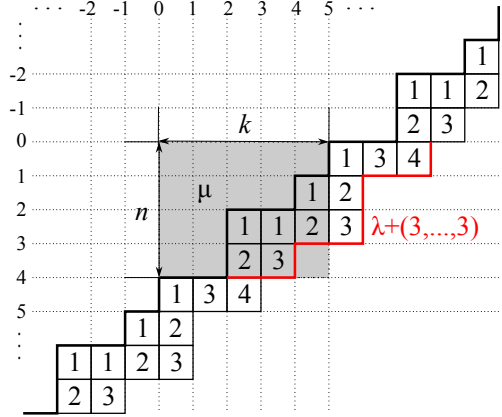}
\end{equation*}%
\caption{Example of a cylindric tableau for $n=4$ and $k=5$. Set $\protect%
\lambda =(5,3,3,1)$, $\protect\mu =(5,4,2,2)$ and $d=3$.}
\label{fig:cyltabintro}
\end{figure}
We believe the cylindric skew Macdonald functions $P_{\lambda /d/\mu
}^{\prime }$ to be of interest because of the following observations.\smallskip

\noindent Firstly, at $t=0$ they yield cylindric Schur functions $s_{\lambda /d/\mu }=\sum_{\lambda \in \mathcal{P}%
_{n,k}^{+}}N_{\mu\nu}^{\lambda}(0)s_{\nu
}$ where $N_{\mu\nu}^{\lambda}(0)$ are the WZNW fusion coefficients.
These cylindric Schur functions arise from the coproduct of the Verlinde algebra seen as Frobenius
algebra. Similar cylindric Schur functions appear in the context of the
small quantum cohomology ring of the Grassmannian \cite{Postnikov} (see also
\cite{McNamara}). Furthermore, cylindric Young diagrams and tableaux occur
in the representation theory of double affine Hecke algebras \cite{SV}.

\noindent Secondly, we have the following observation:

\begin{conjecture}
The polynomials $K_{\nu ^{\prime },\lambda ^{\prime }/d/\mu ^{\prime }}(t)$
defined through the expansion%
\begin{equation}
P_{\lambda /d/\mu }^{\prime }(x_{1},\ldots ,x_{n-1};t)=\sum_{\nu \in \mathcal{P%
}_{n,k}^{+}}K_{\nu ^{\prime },\lambda ^{\prime }/d/\mu ^{\prime }}(t)s_{\nu
}(x_{1},\ldots ,x_{n-1}),  \label{Conj0}
\end{equation}%
where $s_{\nu }$ is the Schur function, have always non-negative
coefficients.
\end{conjecture}

This conjecture is currently based on numerical computations, but as we will explain
in the text, for an appropriate choice of $\mu $ these polynomials
specialise to the celebrated Kostka-Foulkes polynomials $K_{\nu ^{\prime
},\lambda ^{\prime }}(t)$ and the known expansion $P_{\lambda }^{\prime
}(t)=\sum_{\mu }K_{\mu ^{\prime }\lambda ^{\prime }}(t)s_{\mu }$ has
representation theoretic interpretations: setting $t=1$ the Kostka numbers $K_{\nu ^{\prime
},\lambda ^{\prime }}(1)$ are multiplicities of finite-dimensional $%
\mathfrak{sl}(n)$-modules in the Demazure module corresponding to $%
P_{\lambda }^{\prime }$ \cite[Remark after Theorem 8]{Sanderson}. %
In certain cases this result can be generalised to arbitrary $t$: the function $P_{\lambda }^{\prime }(t)$ can be identified
as the graded character of Demazure modules of the current algebra $\mathfrak{sl}(n)\otimes\mathbb{C}[t]$ and
the coefficient of $t^r$ in the Kostka polynomial $K_{\lambda'\mu'}(t)$ provides the dimension of subspaces of degree $r$
in Feigin-Loktev fusion products. This has been conjectured in \cite{FeLo} and proved in \cite[Corollary 1.5.2]{CL}; see also \cite[Theorem 5.2]{KMOTU}, \cite{HKOTT} and \cite{SchillingTingley} for a computation of the graded characters using crystal bases. Results on Demazure modules and fusion products related to Lie algebras other than type $A$ can be found in
\cite{FoLi}. It would be desirable to find similar representation theoretic interpretations of (%
\ref{Conj0}) proving the above conjecture.

In this context we also mention that it has been shown in \cite{ArKe} that graded multiplicities in
Feigin-Loktev fusion products of Kirillov-Reshetikhin modules are related to the generalized Kostka-Foulkes
polynomials introduced in \cite{SchWa} and \cite{ANKS}. Related deformations of fusion
coefficients can be found in \cite{FLOT} and \cite{SS1,SS2}. However, the deformed fusion
coefficients in \emph{loc. cit.} specialise to the known fusion coefficients
at $q=1$ instead of $q=0$ and appear to be different from the structure
constants of the Frobenius algebra discussed in this article. Also, there
have been $q$-deformed versions of the Virasoro algebra suggested \cite%
{Shiraishietal}; at the moment the connection between these constructions
and the one in this article is unclear.

\subsection{Yang-Baxter algebras and quantum integrable systems}
One of the novel aspects of the combinatorial description of the Verlinde algebra in
\cite{KS,Korff,Korffproceedings} is its identification with the commutative algebra generated
by the Hamiltonians or integrals of motion of a quantum integrable system. Here
we show that these findings generalize from the simple combinatorial phase model considered in the case of the Verlinde algebra to a genuinely
strongly-correlated quantum many body system, the so-called q-boson model \cite{Bogoliubovetal}. The quantum Hamiltonians generate a
commutative Frobenius algebra and the latter are in one-to-one correspondence with
two-dimensional topological field theories; we will address this aspect in the conclusions and connect it with recent developments in four-dimensional supersymmetric $N=2$ gauge theories \cite{NS}.

Central starting point for the algebraic formulation of a quantum integrable model are solutions to
the quantum Yang-Baxter equation,
\begin{equation}
R_{12}(x,y)L_{13}(x)L_{23}(y)=L_{23}(y)L_{13}(x)R_{12}(x,y)\;.  \label{YBE0}
\end{equation}%
The latter is an identity in $\limfunc{End}(V(x)\otimes V(y)\otimes \mathcal{%
H}_{q})$ where $V$ is some complex vector space, $V(x):=\mathbb{C}%
[\![x]\!]\otimes V$ and it is understood that $R(x,y):V(x)\otimes
V(y)\rightarrow V(x)\otimes V(y)$ and $L(x):V(x)\otimes \mathcal{H}%
_{q}\rightarrow V(x)\otimes \mathcal{H}_{q}$ only act nontrivially in those
factors of $V(x)\otimes V(y)\otimes \mathcal{H}_{q}$ indicated by the lower
indices when numbering the spaces from left to right with $1,2,3$. Here $x,y$
are some formal variables which - when evaluated in the complex numbers -
are called \emph{spectral parameters}.

One can interpret the relation (\ref{YBE0}) as the definition of a subalgebra
in $\mathcal{H}_{q}$ which is called the \emph{Yang-Baxter algebra}. The
latter comes naturally equipped with a coproduct $\Delta
L(x)=L_{02}(x)L_{01}(x)$ and one is led to consider the monodromy matrix $%
T(x)=L_{0,n}(x)L_{0,n-1}(x)\cdots L_{0,1}(x)\in \limfunc{End}(V(x)\otimes
\mathcal{H}_{q}^{\otimes n})$, where the first index 0 now refers to the
factor $V(x)$ and the second numbers the copies of $\mathcal{H}_{q}$ in $%
\mathcal{H}_{q}^{\otimes n}$. Taking the formal partial trace over $V$ one
obtains the current operator $\mathcal{O}(x)=\sum_{r\geq 0}x^{r}\mathcal{O}%
_{r}=\limfunc{Tr}_{V}T(x)$ with $\mathcal{O}_{r}\in \mathcal{H}_{q}^{\otimes
n}$ and it follows from the Yang-Baxter equation that $\mathcal{O}(x)%
\mathcal{O}(y)=\mathcal{O}(y)\mathcal{O}(x),\;\forall x,y$. These are
Baxter's `commuting transfer matrices' and their matrix elements $\langle
\lambda |\mathcal{O}(x_{1})\cdots \mathcal{O}(x_{\ell })|\mu \rangle $ in
the Fock space $\mathcal{F}^{\otimes n}$ can be interpreted as partition
functions of an exactly-solvable lattice model in statistical mechanics with
periodic boundary conditions. Alternatively, one can interpret the
coefficients $\{\mathcal{O}_{r}\}$ as the commuting Hamiltonians of a
quantum integrable model. Both points of view are important.

\begin{description}
\item[\textbf{Statistical mechanics}.] We will discuss two solutions $(R,L)$
and $(R^{\prime },L^{\prime })$ to the Yang-Baxter equation (\ref{YBE0})
setting $V=\mathbb{C}^{2}$ and $V=\mathcal{F}$, the infinite-dimensional
Fock space of the $q$-Heisenberg algebra $\mathcal{H}_{q}$. The solution for
$V=\mathbb{C}^{2}$ has been obtained previously \cite{Bogoliubovetal}, the
other one is new. We will show that the resulting current operators $%
\mathcal{O}=\boldsymbol{E}$ and $\mathcal{O}^{\prime }=\boldsymbol{G}%
^{\prime }$ can be interpreted as noncommutative analogues of generating
functions for elementary symmetric, $\mathcal{O}_{r}=\boldsymbol{e}_{r}$, and
multivariate Rogers-Szeg\"{o} polynomials, $\mathcal{O}_{r}^{\prime }=%
\boldsymbol{Q}_{(r)}^{\prime }$, in the alphabet $\{a_{1},\ldots ,a_{n}\}$,
respectively. They satisfy an analogue of Baxter's famous $TQ$-relation for
the six and eight-vertex model \cite{Baxter}. The corresponding partition
functions $\langle \lambda |\mathcal{O}(x_{1})\cdots \mathcal{O}(x_{\ell
})|\mu \rangle $ are symmetric functions, since the current operators or
transfer matrices commute. Stating explicit bijections between lattice
configurations of the associated statistical mechanics models and cylindric
Young tableaux we show that they yield cylindric skew Macdonald functions $%
Q_{\lambda ^{\prime }/d/\mu ^{\prime }}$ and $P_{\lambda /d/\mu }^{\prime
}=\sum_{\lambda \in \mathcal{P}_{n,k}^{+}}N_{\mu ^{\prime }\nu ^{\prime
}}^{\lambda ^{\prime }}(t)P_{\nu }^{\prime }$ mentioned earlier. The
ordinary skew Macdonald functions $Q_{\lambda ^{\prime }/\mu ^{\prime
}}=Q_{\lambda ^{\prime }/\mu ^{\prime }}(0,t)$ and $P_{\lambda /\mu
}^{\prime }=P_{\lambda /\mu }(t,0)$ are obtained as special cases for $d=0$
which corresponds to open boundary conditions on the lattice. Another instance
where $P_{\lambda}(t,0)$ has occurred in the context of integrable systems
is in the case of the $q$-deformed Toda chain where they appear as eigenfunctions
of the quantum Hamiltonians \cite{Gerasimovetal}. Here the eigenfunctions of the
Hamiltonians are instead Hall-Littlewood functions; compare with \cite{Tsilevich}.

\item[\textbf{Quantum integrable models}.] The other perspective motivated
by physics is to interpret the coefficients $\{\mathcal{O}_{r}\}$ as quantum
Hamiltonians and the canonical basis vectors in $S^{k}(V)$ as quantum
particle configurations on a ring-shaped lattice, analogous to the
discussion of the Verlinde algebra above. The two solutions $(R,L)$ and $%
(R^{\prime },L^{\prime })$ will yield the Hamiltonians of the so-called $q$%
-boson model and the discrete Laplacians introduced in \cite{vanDiejen} for
a discrete version of the quantum nonlinear Schr\"{o}dinger model.
Mathematically, the $\{\mathcal{O}_{r}\}$ will generate the Frobenius
algebra $\mathfrak{F}_{n,k}$; they can be thought of as the basic building
blocks of the deformed fusion matrices in (\ref{tfusion0}). Computing the
eigenbasis $\{\mathfrak{e}_{\lambda }\}$ of the quantum Hamiltonians $\{%
\mathcal{O}_{r}\}$ in the $k$-particle space, via the so-called \emph{%
algebraic Bethe ansatz}, we will obtain the idempotents of the Frobenius
algebra, $\mathfrak{e}_{\lambda }\ast \mathfrak{e}_{\mu }=\delta _{\lambda
\mu }\mathfrak{e}_{\lambda }$ with $1=\sum_{\lambda }\mathfrak{e}_{\lambda }$%
. This is known as \emph{Peirce decomposition} in the literature \cite%
{Peirce} and requires a novel algebro-geometric proof of completeness of the Bethe ansatz for this model, which we state in Section 7.
\end{description}

\subsection{Outline of the article}

The following list summarises the discussion and results contained in each
section of this article.

\begin{description}
\item[\em Section 2.] We introduce the necessary conventions and algebraic
notions (the extended affine symmetric group, the weight lattice, the affine
Hecke algebra, the quantum affine algebra $U_{q}^{\prime }\widehat{\mathfrak{gl%
}}(n)$, Macdonald functions) which we need to keep this article
self-contained.

\item[\em Section 3.] We discuss the central algebraic object, the $q$-boson
algebra and its Fock space representation. We state the algebra homomorphism
with the quantum affine algebra $U_{q}^{\prime }\widehat{\mathfrak{gl}}(n)$
and prove the module isomorphism with the Kirillov-Reshetikhin module. In
particular, we identify the basis in the Fock space with Lusztig's canonical
basis in the quantum algebra module. Then we introduce several solutions to
the Yang-Baxter equation in terms of the $q$-boson algebra and use them to
introduce analogues of Macdonald polynomials in a noncommutative alphabet.

\item[\em Section 4.] Employing one of the solutions to the quantum Yang-Baxter
equation we define a statistical vertex model and show that its partition
functions on a square lattice with fixed boundary conditions yield ordinary
skew Hall-Littlewood functions.

\item[\em Section 5.] We generalise the discussion of the previous section to
periodic boundary conditions on the lattice and show that the associated
partition functions can be interpreted as cylindric Hall-Littlewood
functions. This section contains in particular the definition of cylindric
loops and cylindric skew tableaux adapted to the present discussion of the
Verlinde algebra. We will state explicit bijections between periodic lattice
configurations and cylindric tableaux. We relate the expansion coefficients
of the cylindric Hall-Littlewood functions in terms of monomial symmetric,
Schur and ordinary Hall-Littlewood functions to matrix elements of the
noncommutative Macdonald functions of Section 3. In particular, we express
the inverse Kostka-Foulkes matrix as a noncommutative analogue of a Schur
polynomial in the $q$-boson algebra.

\item[\em Section 6.] We use the second solution of the quantum Yang-Baxter
equation in Section 3 to define another statistical vertex model whose
partition functions lead to cylindric generalisations of the skew Macdonald
functions $P_{\lambda /\mu }(t,0)$. Similar as in the Hall-Littlewood case
of the previous section, we relate also their expansion coefficients in
various bases in the ring of symmetric functions to matrix elements of
noncommutative Macdonald polynomials in the $q$-boson algebra. We show that
the celebrated Kostka-Foulkes polynomials coincide with the matrix element
of such a polynomial which is dual to the Schur polynomial.

\item[\em Section 7.] Using the algebraic Bethe ansatz we compute the eigenbasis
of the noncommutative Macdonald polynomials in the Fock space. This leads to
a set of polynomial equations which define a discrete algebraic variety. We
discuss the related coordinate ring and, using invariance under the extended
affine symmetric group, identify it as a quotient of the spherical Hecke algebra,
which is closely related to the Ringel-Hall algebra of the Jordan quiver. Furthermore it is
equipped with the structure of a commutative Frobenius algebra and we show
that its product and coproduct are related to cylindric Hall-Littlewood and
Macdonald functions discussed in Sections 5 and 6. Choosing a distinguished
basis its structure constants are polynomials in an indeterminate $t=q^{2}$
whose constant terms equal the WZNW fusion coefficients, the structure
constants of the Verlinde algebra.

\item[\em Section 8.] We summarise our findings and set them into relation with
a recent observation in the context of four-dimensional $N=2$ supersymmetric
gauge theories which suggests a correspondence between two-dimensional
topological quantum field theories and integrable quantum many-body systems.
\end{description}

\section{Preliminaries}

\subsection{q-numbers}

Let $q$ be an indeterminate then we define the following standard $q$%
-numbers,%
\begin{equation*}
\lbrack m]_{q}:=\frac{q^{m}-q^{-m}}{q-q^{-1}},\qquad \lbrack
m]_{q}!:=[m]_{q}[m-1]_{q}\cdots \lbrack 2]_{q}[1]_{q}\;.
\end{equation*}%
In addition, we require the $q$-Pochhammer symbol%
\begin{equation}
(x;q)_{\infty }:=\prod_{r=0}^{\infty }(1-xq^{r}),\quad
(x;q)_{r}:=\prod_{s=0}^{r-1}(1-xq^{s})=\frac{(x;q)_{\infty }}{%
(xq^{r};q)_{\infty }}\;.
\end{equation}%
For any composition $\lambda =(\lambda _{1},\ldots ,\lambda _{r})\in \mathbb{%
Z}_{\geq 0}^{r}$ we will use the shorthand notations%
\begin{equation}
(x;q)_{\lambda }:=\prod_{i>0}(x;q)_{\lambda _{i}},\quad
(q)_{r}:=(q;q)_{r},\quad (q)_{\lambda }:=(q)_{\lambda _{1}}\cdots
(q)_{\lambda _{r}}\;.
\end{equation}%
We will also need the $q$-binomial coefficients
\begin{equation}
\QATOPD[ ] {m+n}{n}_{q}:=\left\{
\begin{array}{cc}
\frac{(q)_{m+n}}{(q)_{m}(q)_{n}}=\frac{(1-q^{m+n})\cdots (1-q^{n+1})}{%
(1-q)\cdots (1-q^{n})}, & m,n\geq 0 \\
0, & \text{else}%
\end{array}%
\right.
\end{equation}%
N.B. $q$ plays here the role of a dummy variable and we will apply the same
definitions for other indeterminates or powers of $q$ in particular we will
often use $t=q^{2}$ instead of $q$.

\subsection{The extended affine symmetric group}

We recall the definition of the extended affine symmetric group $\mathfrak{%
\hat{S}}_{r}$. The latter is generated by the elements $\{\sigma _{0},\sigma _{1},\ldots ,\sigma
_{r-1},\tau ^{\pm 1}\}$ subject to the relations%
\begin{eqnarray}
\sigma _{i}^{2} &=&1,\quad \tau \sigma _{i}\tau ^{-1}=\sigma _{i+1}  \notag
\\
\sigma _{i}\sigma _{i+1}\sigma _{i} &=&\sigma _{i+1}\sigma _{i}\sigma
_{i+1},\quad \sigma _{i}\sigma _{j}=\sigma _{j}\sigma
_{i},\;|i-j|>1,\;i,j\in \mathbb{Z}_{r}\ .  \label{affSk}
\end{eqnarray}%
Each $w\in \mathfrak{\hat{S}}_{r}$ can be written as $w=\tau ^{m}\sigma $
for some $m\in \mathbb{Z}$ and $\sigma \in \mathfrak{S}_{r}$ with $\mathfrak{%
S}_{r}\subset \mathfrak{\hat{S}}_{r}$ being the symmetric group on $r$%
-letters generated by $\{\sigma _{1},\ldots ,\sigma _{r-1}\}$. The Bruhat
order can be extended from $\mathfrak{S}_{r}$ to $\mathfrak{\hat{S}}_{r}$ by
setting $w<w^{\prime }$ if $w=\tau ^{m}\sigma ,w^{\prime }=\tau ^{m^{\prime
}}\sigma ^{\prime }$ with $m=m^{\prime }$ and $\sigma <\sigma ^{\prime }$.
Similarly, one can use the decomposition $w=\tau ^{m}\sigma $ to define the
length function as $\ell (w)=\ell (\sigma )$. We shall denote the longest
element in $\mathfrak{S}_{r}$ by $w_{r}$.

Note that an alternative set of generators for $\mathfrak{\hat{S}}_{r}$ is $%
\{\sigma _{1},\ldots ,\sigma _{r-1}\}\cup \{y_{1},\ldots ,y_{r}\}$ where $%
y_{i}y_{j}=y_{j}y_{i}$ for all $1\leq i,j\leq r$ and $\sigma _{i}y_{i}\sigma
_{i}=y_{i+1},~\sigma _{i}y_{j}=y_{j}\sigma _{i}$ for $j\neq i,i+1$. Both
definitions are related via $\sigma _{0}=\sigma _{r-1}\cdots \sigma
_{2}\sigma _{1}\sigma _{2}\cdots \sigma _{r-1}y_{1}^{-1}y_{r}$ and $\tau
=\sigma _{1}\sigma _{2}\cdots \sigma _{r-1}y_{r}$.

\subsection{Action on the weight lattice}

Let $\mathcal{P}_{r}=\tbigoplus_{i=1}^{r}\mathbb{Z}\epsilon _{i}$ be the $%
\mathfrak{gl}(r)$ weight lattice with standard basis $\epsilon _{1},\ldots
,\epsilon _{r}$ and inner product $(\epsilon _{i},\epsilon _{j})=\delta
_{i,j}$. We denote the simple roots by $\alpha _{i}=\epsilon _{i}-\epsilon
_{i+1}$, $i=1,\ldots ,r-1$ and the affine root by $\alpha _{r}=\epsilon
_{r}-\epsilon _{1}$. Let $\mathcal{P}_{r}^{\pm }$ denote the set of
(integral) \emph{dominant} and \emph{anti-dominant} weights. Recall the
following right level $s$ action of $\mathfrak{\hat{S}}_{r}$\ on $\mathcal{P}%
_{r}$ for $s\geq 1$:%
\begin{eqnarray}
\lambda \sigma _{i} &=&(\lambda _{1},\ldots ,\lambda _{i+1},\lambda
_{i},\ldots ,\lambda _{r}),\quad i=1,\ldots ,r-1,  \notag \\
\lambda \sigma _{0} &=&(\lambda _{r}+s,\lambda _{2},\ldots ,\lambda
_{r-1},\lambda _{1}-s),  \notag \\
\lambda \tau &=&(\lambda _{r}+s,\lambda _{1},\lambda _{2},\ldots ,\lambda
_{r-1}),  \notag \\
\lambda y_{i} &=&(\lambda _{1},\ldots ,\lambda _{i}+s,\ldots ,\lambda
_{r})\;.  \label{affaction}
\end{eqnarray}%
The subsets
\begin{equation}
\mathcal{A}_{r,s}^{+}:=\left\{ \left. (\lambda _{1},\lambda _{2},\ldots
,\lambda _{r})\in \mathcal{P}_{r}^{+}~\right\vert ~s\geq \lambda _{1}\geq
\lambda _{2}\geq \cdots \geq \lambda _{r}\geq 1\right\}  \label{alcove}
\end{equation}%
and $\mathcal{A}_{r,s}^{-}=w_{r}\mathcal{A}_{r,s}^{+}$ are both fundamental
domains with respect to the level $s$ action of $\mathfrak{\hat{S}}_{r}$\ on
$\mathcal{P}_{r}$. For each $\lambda \in \mathcal{A}_{r,s}^{-}$ denote by $%
\mathfrak{S}_{\lambda }\subset \mathfrak{S}_{r}$ the stabilizer subgroup of $%
\lambda $ and by $\mathfrak{S}^{\lambda }$ the set of minimal length
representatives of the cosets $\mathfrak{S}_{\lambda }\backslash \mathfrak{%
\hat{S}}_{r}$. We shall use the symbol $w_{\lambda }$ for the longest
element in $\mathfrak{S}_{\lambda }$.

Throughout this article we will make use of the following bijections.

\begin{description}
\item[\em Reduction.] For practical reasons we will also need the set
\begin{equation*}
\mathcal{\tilde{A}}_{r,s}^{+}:=\left\{ \left. (\lambda _{1},\lambda
_{2},\ldots ,\lambda _{r})\in \mathcal{P}_{r}^{+}~\right\vert ~s>\lambda
_{1}\geq \lambda _{2}\geq \cdots \geq \lambda _{r}\geq 0\right\}
\end{equation*}%
which again is a fundamental domain. We note that there exists a bijection $%
\symbol{126}:\mathcal{A}_{r,s}^{+}\rightarrow \mathcal{\tilde{A}}_{r,s}^{+}$
by sending $\lambda $ to $\tilde{\lambda}$, the partition obtained from $%
\lambda $ by deleting all parts of size $s$. We shall make frequently use of
this map.

\item[\em $\ast $-involution.] In addition, we will require the following $\ast $%
-involution on $\mathcal{A}_{r,s}^{+}$: given $\lambda \in \mathcal{A}%
_{r,s}^{+}$ define $\lambda ^{\ast }$ to be the unique element which is the
inverse image of $(s-\lambda _{r},\ldots ,s-\lambda _{2},s-\lambda _{1})\in
\mathcal{\tilde{A}}_{r,s}^{+}$ under the above bijection $\symbol{126}:%
\mathcal{A}_{r,s}^{+}\rightarrow \mathcal{\tilde{A}}_{r,s}^{+}$. Note that
when identifying partitions with weights, $\lambda^\ast$ is simply the
contragredient weight of $\lambda$.

\item[\em Rotation.] The $\widehat{\mathfrak{%
gl}}(n)$ Dynkin diagram automorphism induces a bijection $\limfunc{rot}:\mathcal{A}%
_{k,n}^{+}\rightarrow \mathcal{A}_{k,n}^{+}$ given by $%
\lambda \mapsto \limfunc{rot}(\lambda ):=\mu $ with $m_{i}(\mu
):=m_{i+1}(\lambda )$, $i\in \mathbb{Z}_{n}$.
\end{description}

\subsection{The affine Hecke algebra}

The affine Hecke algebra $\hat{H}_{r}$ is the $\mathbb{C}[q,q^{-1}]$\
algebra generated by $\{T_{0},T_{1},\ldots ,T_{r-1}\}$ and an invertible
element $T_{\tau }\equiv \tau $ subject to the relations%
\begin{gather*}
(T_{i}-q^{-1})(T_{i}+q)=0,\quad \quad \tau T_{i}\tau ^{-1}=T_{i+1} \\
T_{i}T_{i+1}T_{i}=T_{i+1}T_{i}T_{i+1},\quad
T_{i}T_{j}=T_{j}T_{i},\quad |i-j|>1,\quad i,j\in \mathbb{Z}_{r}
\end{gather*}%
A basis $\{T_{w}\}_{w\in \mathfrak{\hat{S}}_{r}}$ is constructed as follows:
for any $w,w^{\prime }\in \mathfrak{\hat{S}}_{r}$ set $T_{ww^{\prime
}}:=T_{w}T_{w^{\prime }}$ if $\ell (w)+\ell (w^{\prime })=\ell (ww^{\prime
}) $ with $T_{\sigma _{i}}\equiv T_{i}$. Alternatively, $\hat{H}_{r}$ is in
the \emph{Bernstein presentation} generated by $\{T_{1},\ldots ,T_{r-1}\}$ and
a set of commuting, invertible elements $\{Y_{1},\ldots ,Y_{r}\}$ obeying%
\begin{equation*}
T_{i}Y_{i}T_{i}=Y_{i+1},\qquad T_{i}Y_{j}=Y_{j}T_{i}\qquad\text{for}\quad j\neq
i,i+1\;.
\end{equation*}%
To relate this with the previous presentation use the formulae $%
T_{0}=T_{r-1}^{-1}\cdots T_{2}^{-1}T_{1}^{-1}T_{2}^{-1}\cdots
T_{r-1}^{-1}Y_{1}^{-1}Y_{r}$\ and $\tau =T_{1}^{-1}T_{2}^{-1}\cdots
T_{r-1}^{-1}Y_{r}$. Conversely,\ $%
Y_{i}=T_{i-1}T_{i-2}\cdots T_{1}T_{0}^{-1}T_{r-1}^{-1}T_{r-2}^{-1}\cdots
T_{i}$ and $Y^{\lambda }=T_{y^{\lambda }}^{-1}$ for $\lambda \in \mathcal{P}%
_{r}^{+}$. There exists a canonical bar involution $\overline{\;\;}:\hat{H}%
_{r}\rightarrow \hat{H}_{r}$ defined by $\bar{q}=q^{-1}$ and $\bar{T}%
_{w}=(T_{w^{-1}})^{-1}$; in particular $\bar{T}_{i}=T_{i}-(q-q^{-1})$.

Define $\boldsymbol{1}_r:=(1/c_r)\sum_{w\in \mathfrak{S}_r}q^{-\ell(w)}T_w$ with $c_r=\sum_w q^{-\ell(w)}$ and denote by
$\mathcal{Z}(\hat{H}_r)\cong\mathbb{C}[q,q^{-1}][Y_1,\ldots,Y_r]^{\mathfrak{S}_r}$ the centre of the affine Hecke algebra. Then the map
$\Phi:\mathcal{Z}(\hat{H}_r)\rightarrow\boldsymbol{1}_r\hat{H}_r\boldsymbol{1}_r$ given by
$\frac{c_\lambda}{c_r}P_\lambda(Y_1,\ldots,Y_r;0,t=q^2)\mapsto \boldsymbol{1}_rY^\lambda
\boldsymbol{1}_r
$, where $P_\lambda(0,t)$ are the Hall-Littlewood polynomials (see below),
is known as the {\em Satake isomorphism} and $\boldsymbol{1}_r\hat{H}_r\boldsymbol{1}_r$ as the {\em spherical Hecke algebra}. Recall that $\{\boldsymbol{1}_rY^\lambda
\boldsymbol{1}_r:\lambda_1\geq\ldots\geq\lambda_r\geq 0\}$ is a basis of the spherical Hecke algebra;
see e.g. \cite{NelsenRam} for details and references.

\subsection{The quantum enveloping algebra of $\widehat{\mathfrak{gl}}(n)$}

Let $\hat{U}_{n}=U_{q}^{\prime }\widehat{\mathfrak{gl}}(n)$ be the unital
associative $\mathbb{C}(q)$-algebra generated by $%
\{E_{i},F_{i},K_{i}^{\pm 1}\}_{i=1,\ldots ,n}$ and subject to the relations:

\begin{itemize}
\item[(R1)] The $K_{i}^{\pm 1}$'s commute with one another and $%
K_{i}K_{i}^{-1}=K_{i}^{-1}K_{i}=1$.

\item[(R2)] $K_{i}E_{j}=q^{\delta _{ij}-\delta
_{ij+1}}E_{j}K_{i},\;K_{i}F_{j}=q^{-\delta _{ij}+\delta _{ij+1}}F_{j}K_{i}$
and
\begin{equation*}
E_{i}F_{j}-F_{j}E_{i}=\delta _{ij}\dfrac{K_{i,i+1}-K_{i+1,i}}{q-q^{-1}}
\end{equation*}

\item[(R3)] For $X=E,F$ we have $X_{i}X_{j}=X_{j}X_{i}$ for $|i-j|>1$ and
else%
\begin{eqnarray}
X_{i}^{2}X_{i+1}-(q+q^{-1})X_{i}X_{i+1}X_{i}+X_{i+1}X_{i}^{2} &=&0,  \notag
\\
X_{i}X_{i+1}^{2}-(q+q^{-1})X_{i+1}X_{i}X_{i+1}+X_{i+1}^{2}X_{i} &=&0\;.
\label{qSerre}
\end{eqnarray}
\end{itemize}

\noindent Here $K_{i,j}:=K_{i}K_{j}^{-1}$ and all indices are understood
modulo $n$.

We will denote by $\mathbf{\hat{U}}_{n}=U_{q}^{\prime }\widehat{\mathfrak{sl}%
}(n)$ the subalgebra generated by $E_{i},F_{i}$ and $K_{i,i+1},K_{i+1,i}$
and by $U_{n},\mathbf{U}_{n}$ the finite-dimensional subalgebras obtained
when restricting the index to $i=1,\ldots ,n-1$. We choose to work with the
coproduct defined via $\Delta (K_{i})=K_{i}\otimes K_{i}$ and%
\begin{equation}
\Delta (E_{i})=E_{i}\otimes K_{i+1,i}+1\otimes E_{i},\quad \Delta
(F_{i})=F_{i}\otimes 1+K_{i,i+1}\otimes F_{i}\;.  \label{coprod}
\end{equation}%
This will allow us to make contact with the discussion in \cite[Section 3,
page 7]{Brundan}. The corresponding co-unit and antipode are respectively
given by
\begin{equation}
\varepsilon (E_{i})=\varepsilon (F_{i})=0,\qquad \varepsilon (K_{i}^{\pm
1})=1  \label{counit}
\end{equation}%
and%
\begin{equation}
S(E_{i})=-E_{i}K_{i,i+1},\quad S(F_{i})=-K_{i+1,i}F_{i},\quad S(K_{i}^{\pm
1})=K_{i}^{\mp 1}\;.  \label{antipode}
\end{equation}%
Furthermore, for discussing the canonical basis below we will require the
\emph{bar involution}; this is the unique antilinear automorphism defined
via $\bar{E}_{i}=E_{i},$ $\bar{F}_{i}=F_{i}$ and $\bar{K}_{i}=K_{i}^{-1}$. A
$U_{n}$-module $V$ is said to possess a compatible bar involution $%
V\rightarrow V$ if $\overline{uv}=\bar{u}\bar{v}$ for all $u\in U_{n}$ and $%
v\in V$.

\subsection{Macdonald functions}

Let $q,t$ be indeterminates and consider the following extension of the ring
of symmetric functions $\Lambda (q,t)=\Lambda \otimes _{\mathbb{Z}}\mathbb{C}%
(q,t)$ where $\Lambda =\lim\limits_{\longleftarrow }\Lambda _{n}$ is the
projective limit of the projective system of symmetric polynomials in $n$
variables, $\Lambda _{n}=\mathbb{Z}[x_{1},\ldots ,x_{n}]^{\mathfrak{S}_{n}},$
$n\geq 1$, with the canonical map $\Lambda _{n+1}\twoheadrightarrow \Lambda
_{n}$ which sends $x_{n+1}$ to zero. We now review the definition of a
special basis in $\Lambda (q,t)$, known as \emph{symmetric Macdonald
functions}; we shall follow the conventions used in \cite{Macdonald}.

For a given partition $\lambda $ define
\begin{equation}
b_{\lambda }(q,t)=\prod_{s\in \lambda }b_{\lambda }(s;q,t),\quad \quad
b_{\lambda }(s;q,t):=\frac{1-q^{a_{\lambda }(s)}t^{l_{\lambda }(s)+1}}{%
1-q^{a_{\lambda }(s)+1}t^{l_{\lambda }(s)}},  \label{Macb}
\end{equation}%
where the product runs over all squares $s=(i,j)$ in the Young diagram of $%
\lambda $ and $a_{\lambda }(s)=\lambda _{i}-j$ is the arm-length (number of
squares to the east) and $l_{\lambda }(s)=\lambda _{j}^{\prime }-i$ the
leg-length (number of squares to the south). Given a skew diagram $\lambda
/\mu $, define two functions $\varphi _{\lambda /\mu },\psi _{\lambda /\mu }$
which are zero unless $\lambda -\mu $ is a horizontal $r$-strip in which case%
\begin{equation}
\varphi _{\lambda /\mu }(q,t)=\prod_{s\in C_{\lambda /\mu }}\frac{b_{\lambda
}(s;q,t)}{b_{\mu }(s;q,t)}\quad \text{and}\quad \psi _{\lambda /\mu
}(q,t)=\prod_{s\in R_{\lambda /\mu }-C_{\lambda /\mu }}\frac{b_{\mu }(s;q,t)%
}{b_{\lambda }(s;q,t)}  \label{Macphipsi}
\end{equation}%
with $C_{\lambda /\mu }$ (respectively $R_{\lambda /\mu }$) being the union
of all columns (respectively rows) which intersect $\lambda /\mu $.

Let $\lambda $, $\mu $ be partitions with $\mu \subset \lambda $ and denote
by $\lambda /\mu $ the associated skew diagram. Given a (semi-standard%
\footnote{%
In this article we will only consider semi-standard tableaux and,
henceforth, simply call them \textquotedblleft tableaux\textquotedblright .}%
) tableau $T$ of shape $\lambda /\mu $ decompose it into a sequence of
partitions $\mu =\lambda ^{(0)}\subset \lambda ^{(1)}\subset \ldots \subset
\lambda ^{(r)}=\lambda $ such that $\lambda ^{(i+1)}/\lambda ^{(i)}$ is a
horizontal strip and set $\varphi _{T}:=\prod_{i\geq 0}\varphi _{\lambda
^{(i+1)}/\lambda ^{(i)}}$, $\psi _{T}:=\prod_{i\geq 0}\psi _{\lambda
^{(i+1)}/\lambda ^{(i)}}$. N.B. we have the identity $b_{\mu }\varphi
_{T}=b_{\lambda }\psi _{T}$ for any tableau $T$.

\begin{example}
Let $\lambda =(3,2,2,1),$ $\mu =\emptyset $, the empty partition, and
consider all tableaux of weight $(2,2,2,2)$,%
\begin{equation*}
\begin{tabular}{ccc}
{\small \young(112,23,34,4) } & {\small \young(113,22,34,4) } & {\small %
\young(114,22,33,4)}%
\end{tabular}%
\;.
\end{equation*}%
Then one finds the weights%
\begin{equation*}
\begin{tabular}{cccc}
$\psi _{T}=$ & $\dfrac{(1+q)(1-t)}{1-qt},$ & $\dfrac{(1-q^{2}t)(1-t^{2})}{%
(1-qt)(1-qt^{2})},$ & $\dfrac{(1-q^{2}t^{2})(1-t^{3})}{(1-qt^{2})(1-qt^{3})}%
, $%
\end{tabular}%
\end{equation*}%
and $\varphi _{T}=b_{\lambda }\psi _{T}$ with%
\begin{equation*}
b_{\lambda }(q,t)=\frac{%
(1-t)^{3}(1-t^{2})(1-qt^{2})(1-qt^{3})^{2}(1-q^{2}t^{4})}{%
(1-t)^{3}(1-t^{2})(1-qt^{2})(1-qt^{3})^{2}(1-q^{2}t^{4})}\;.
\end{equation*}
\end{example}

\noindent\emph{Skew Macdonald functions} $Q_{\lambda /\mu }(x;q,t)=b_{\lambda
}(q,t)P_{\lambda }(x;q,t)/b_{\mu }(q,t)$ can be defined as the following
weighted sums over semi-standard tableaux $T$,%
\begin{equation}
Q_{\lambda /\mu }(x;q,t)=\sum_{|T|=\lambda/\mu }\varphi _{T}(q,t)x^{T}\qquad
\text{and}\qquad P_{\lambda /\mu }(x;q,t)=\sum_{|T|=\lambda/\mu }\psi
_{T}(q,t)x^{T}\,.  \label{MacdonaldP&Q}
\end{equation}%
Specialising to $\mu =\emptyset $, the empty partition, one obtains
ordinary, non-skew, Macdonald functions which are simply denoted by $%
Q_{\lambda }$, $P_{\lambda }$.

\begin{theorem}[Macdonald]
The family $\{Q_{\lambda }(q,t):\lambda $ \emph{partition~}$\}$ forms a
basis of the ring of symmetric functions $\Lambda (q,t)$.
\end{theorem}

Fix a bilinear form $\Lambda (q,t)\times \Lambda (q,t)\rightarrow \mathbb{C}%
(q,t)$ (antilinear in the first factor) by setting%
\begin{equation}
(P_{\lambda },Q_{\mu })\mapsto \langle P_{\lambda },Q_{\mu }\rangle
_{q,t}:=\delta _{\lambda \mu }\;.  \label{Macprod}
\end{equation}%
Macdonald functions interpolate between various other bases in the ring of
symmetric functions and we shall make repeated use of the following special
cases.

\subsubsection{Special cases of Macdonald functions}

\label{specialMac}

\begin{description}
\item[\emph{Elementary symmetric functions.}] Suppose $\lambda =(1^{r})$ is
a vertical strip, then $P_{(1^{r})}=e_{r}$ with%
\begin{equation}
E(u)=\prod_{i>0}(1+ux_{i})=\sum_{r\geq 0}e_{r}u^{r},\qquad
e_{r}=\sum_{i_{1}<...<i_{r}}x_{i_{1}}\cdots x_{i_{r}}\;.  \label{E}
\end{equation}%
The set $\{e_{\lambda }\}$ where $\lambda $ ranges over the partitions and $%
e_{\lambda }:=e_{\lambda _{1}}e_{\lambda _{2}}\cdots $ is a $\mathbb{Z}$%
-basis of the ring of symmetric functions $\Lambda $.\smallskip

\item[\emph{Basic Macdonald functions.}] Suppose $\lambda =(r)$ is a
horizontal $r$-strip, then $Q_{(r)}=g_{r}$ with%
\begin{equation}
G(u)=\prod_{i\geq 1}\frac{(t~x_{i}u;q)_{\infty }}{(x_{i}u;q)_{\infty }}%
=\sum_{r\geq 0}g_{r}(x;q,t)u^{r},\quad g_{r}(q,t)=\sum_{|\mu |=r}\frac{%
(t;q)_{\mu }}{(q;q)_{\mu }}~m_{\mu },  \label{G}
\end{equation}%
where $m_{\mu }(x)=P_{\lambda }(x;q,1)$\emph{\ }are the \emph{monomial
symmetric functions}. The set $\{g_{\lambda }(q,t)\}$ where $\lambda $
ranges over the partitions and $g_{\lambda }(q,t):=g_{\lambda
_{1}}(q,t)g_{\lambda _{2}}(q,t)\cdots $ is a basis in $\Lambda (q,t)$. Below
we will consider the special limits $g_{\lambda }:=g_{\lambda }(0,t)$ and $%
g_{\lambda }^{\prime }:=g_{\lambda }(q,0)$. In particular, we have that%
\begin{equation}
g_{r}^{\prime }(q):=g_{r}(q,0)=\sum_{\mu \vdash r}\frac{m_{\mu }}{(q)_{\mu }}%
,\qquad (q)_{\mu }:=(q)_{\mu _{1}}(q)_{\mu _{2}}\cdots  \label{g'}
\end{equation}%
are multivariate Rogers-Szeg\"{o} polynomials \cite{Rogers} $h_{r}^{\prime
}=(q)_{r}g_{r}^{\prime }$ with generating function \cite[Chapter 3, Example
17]{Andrews}%
\begin{equation}
G^{\prime }(u;q)=\prod_{i>0}\frac{1}{(ux_{i};q)_{\infty }}=\sum_{r\geq
0}h_{r}^{\prime }\frac{u^{r}}{(q)_{r}}\;.  \label{G'}
\end{equation}%
\smallskip

\item[\emph{Hall-Littlewood functions.}] Specialising $q=0$ the Macdonald
functions become \emph{Hall-Littlewood (HL) functions} sometimes also called
\emph{spherical Macdonald functions} because of their relation to the spherical Hecke algebra,%
\begin{equation}
Q_{\lambda }(x;0,t)=\sum_{|T|=\lambda }\varphi _{T}(0,t)x^{T}\quad \text{%
and\quad }P_{\lambda }(x;0,t)=\sum_{|T|=\lambda }\psi _{T}(0,t)x^{T}\;,
\label{HLdef}
\end{equation}%
where the coefficients $\varphi _{\lambda /\mu }(0,t)=:\varphi _{\lambda
/\mu }(t),~\psi _{\lambda /\mu }(0,t)=:\psi _{\lambda /\mu }(t)$ now have
the simpler form%
\begin{equation}
\varphi _{\lambda /\mu }(t)=\left\{
\begin{array}{cc}
\prod_{i\in I_{\lambda /\mu }}(1-t^{\lambda _{i}^{\prime }-\lambda
_{i+1}^{\prime }}), & \text{if }\lambda -\mu \text{ is a horizontal strip}
\\
0, & \text{otherwise}%
\end{array}%
\right.  \label{phi}
\end{equation}%
and%
\begin{equation}
\psi _{\lambda /\mu }(t)=\left\{
\begin{array}{cc}
\prod_{i\in J_{\lambda /\mu }}(1-t^{\mu _{i}^{\prime }-\mu _{i+1}^{\prime
}}), & \text{if }\lambda /\mu \text{ is a horizontal strip} \\
0, & \text{otherwise}%
\end{array}%
\right. \ .  \label{psi}
\end{equation}%
Here the index set $I_{\lambda /\mu }$ contains all integers $i$ for which $%
\theta _{i}^{\prime }=1$ and $\theta _{i+1}^{\prime }=0$ with $\theta
^{\prime }=\lambda ^{\prime }-\mu ^{\prime }$ being the transposed
skew-diagram. In contrast, $J_{\lambda /\mu }$ consists of the integers $i$
for which $\theta _{i}^{\prime }=0$ and $\theta _{i+1}^{\prime }=1$. The
integers $m_{i}(\lambda)=\lambda _{i}^{\prime }-\lambda
_{i+1}^{\prime }$ and $m_{i}(\mu)=\mu _{i}^{\prime }-\mu
_{i+1}^{\prime }$ are the multiplicities of the part $i$ in the partitions $%
\lambda $ and $\mu $, respectively. Note that
\begin{equation}
b_{\mu }(t)\varphi _{\lambda /\mu }(t)=b_{\lambda }(t)\psi _{\lambda /\mu
}(t),\qquad b_{\lambda }(t)=\prod_{i>0}(t)_{m_{i}(\lambda )}\;.
\label{bphi=bpsi}
\end{equation}%
In what follows we will omit the dependence on the second parameter in the
notation and simply write $Q_{\lambda }(x;t)=Q_{\lambda }(x;0,t)$ and $%
P_{\lambda }(x;t)=P_{\lambda }(x;0,t)$. Denote by $R_{ij}$ the familiar
raising and lowering operators of the ring of symmetric functions, $%
R_{ij}\lambda =(\lambda _{1},\ldots ,\lambda _{i}+1,\ldots ,\lambda
_{j}-1,\ldots )$. Then we have the expression
\begin{equation}
Q_{\lambda }=\prod_{i<j}\frac{1-R_{ij}}{1-tR_{ij}}~g_{\lambda },\qquad
g_{\lambda }=g_{\lambda _{1}}g_{\lambda _{2}}\ldots ,  \label{Q2g}
\end{equation}%
which expresses the Hall-Littlewood $Q$-function as polynomial in the $%
g_{r}=g_{r}(0,t)$.\smallskip

\item[\emph{Demazure Characters and $q$-Whittaker functions.}] In light of the definition (\ref{HLdef})
it is natural to consider also the complementary limit of Macdonald
polynomials setting now $t=0$. Define $\psi _{\lambda /\mu }^{\prime
}(q):=\psi _{\lambda /\mu }(q,0)$ then%
\begin{equation}
\psi _{\lambda /\mu }^{\prime }(q)=\left\{
\begin{array}{cc}
\prod_{i>0}\QATOPD[ ] {\lambda _{i}-\lambda _{i+1}}{\lambda _{i}-\mu _{i}}%
_{q}, & \text{if }\lambda /\mu \text{ is a horizontal strip} \\
0, & \text{otherwise}%
\end{array}%
\right.  \label{psiprime}
\end{equation}%
and, similarly, setting $\varphi _{\lambda /\mu }^{\prime }(q):=\varphi
_{\lambda /\mu }(q,0)$ one finds%
\begin{equation}
\varphi _{\lambda /\mu }^{\prime }(q)=\left\{
\begin{array}{cc}
\frac{1}{(q;q)_{\lambda _{1}-\mu _{1}}}\prod_{i>0}\QATOPD[ ] {\mu _{i}-\mu
_{i+1}}{\lambda _{i+1}-\mu _{i+1}}_{q}, & \text{if }\lambda /\mu \text{ is a
horizontal strip} \\
0, & \text{otherwise}%
\end{array}%
\right. \;.  \label{phiprime}
\end{equation}%
N.B. the identities $\varphi _{\lambda /\mu }^{\prime }(q)=\frac{b_{\lambda
}(q,0)}{b_{\mu }(q,0)}~\psi _{\lambda /\mu }(q,0)=\frac{b_{\mu ^{\prime
}}(0,q)}{b_{\lambda ^{\prime }}(0,q)}~\psi _{\lambda /\mu }^{\prime }(q)$
hold with $\mu ^{\prime },\lambda ^{\prime }$ denoting the conjugate
partitions of $\lambda ,\mu $. We shall denote the resulting skew functions
by%
\begin{equation}
P_{\lambda /\mu }^{\prime }(x;q)=\sum_{|T|=\lambda /\mu }\psi _{T}^{\prime
}(q)x^{T}\quad \text{and}\quad Q_{\lambda /\mu }^{\prime
}(x;q)=\sum_{|T|=\lambda /\mu }\varphi _{T}^{\prime }(q)x^{T}\,.  \label{cHLdef}
\end{equation}%
As mentioned in the introduction for $\mu=\emptyset$ these specialisations of Macdonald functions coincide with certain Demazure characters \cite{Sanderson}. In \cite{Gerasimovetal} they have been named {\em $q$-deformed} or simply {\em $q$-Whittaker functions}. For general $\mu$ these links have not been established and, therefore, I shall refer to them simply as skew Macdonald functions, although it will always be understood that $t=0$.
Applying the involution $%
\omega :\Lambda \rightarrow \Lambda $ defined via $s_{\lambda }\mapsto
s_{\lambda ^{\prime }},$ where $s_{\lambda }=\det (e_{\lambda _{i}^{\prime
}-i+j})$ is the Schur function, we obtain from the Macdonald functions $%
P_{\lambda }^{\prime },Q_{\lambda }^{\prime }$ the so-called \emph{modified
Hall-Littlewood functions}. The latter form dual bases of the ordinary HL
functions with respect to the standard inner product $\langle s_{\lambda
},s_{\mu }\rangle =\delta _{\lambda \mu }$. That is, we have $\tilde{Q}%
_{\lambda ^{\prime }}=\omega P_{\lambda }^{\prime }$ and $\tilde{P}_{\lambda
^{\prime }}=\omega Q_{\lambda }^{\prime }$with $\langle \tilde{Q}_{\lambda
},P_{\mu }\rangle =\langle \tilde{P}_{\lambda },Q_{\mu }\rangle =\delta
_{\lambda \mu }$.
\end{description}

\subsubsection{Hall algebra, coproduct and skew Hall-Littlewood functions}

One of the reasons for the prominence of Hall-Littlewood functions is their
connection with the Hall algebra. When the variable $t$ is evaluated as the
cardinality of a finite field it is well-known that the coefficients in the
product expansion (\ref{Hall}) of Hall-Littlewood $P$-functions are related
to Hall polynomials, the structure constants of Steinitz's Hall algebra. For
generic $t$ the expansion coefficients $f_{\lambda \mu }^{\nu }(t)$ are
polynomials which vanish identically unless the Littlewood-Richardson
coefficient $f_{\lambda \mu }^{\nu }(0)=c_{\lambda ^{\prime }\mu ^{\prime
}}^{\nu ^{\prime }}=c_{\lambda \mu }^{\nu }$ is nonzero.

Skew Hall-Littlewood functions are intimately linked to the product
expansion (\ref{Hall}) through the following construction: endow $\Lambda
(t)=\Lambda (0,t)$ with the coproduct $\Delta :\Lambda (t)\rightarrow
\Lambda (t)\otimes \Lambda (t)$ which is the projective limit of the natural
embedding $\mathbb{C}[x_{1},\ldots ,x_{2n}]^{\mathfrak{S}_{2n}}%
\hookrightarrow \mathbb{C}[x_{1},\ldots ,x_{n}]^{\mathfrak{S}_{n}}\otimes
\mathbb{C}[x_{n+1},\ldots ,x_{2n}]^{\mathfrak{S}_{n}}$. The specialisation
of the bilinear form (\ref{Macprod}) at $q=0$ yields the unique inner
product $\left\langle ~\cdot ~,~\cdot ~\right\rangle _{t}$ on $\Lambda (t)$
such that%
\begin{equation}
\left\langle f,gh\right\rangle _{t}=\left\langle \Delta f,g\otimes
h\right\rangle _{t}\qquad \text{and\qquad }\left\langle
p_{m},p_{n}\right\rangle _{t}=\delta _{m,n}\frac{m}{t^{m}-1}\;,
\label{HLproduct}
\end{equation}%
where $p_{m}=\sum_{i}x_{i}^{m}$ is the $m^{\text{th}}$ power sum, the latter
generate the $\mathbb{Q}$-algebra $\Lambda ^{\mathbb{Q}}=\Lambda \otimes _{%
\mathbb{Z}}\mathbb{Q}$ of symmetric functions. The co-algebra of symmetric
functions $(\Lambda (t),\Delta )$ can be turned into a bi-algebra with
respect to the co-unit $\varepsilon (f):=f(0,0,\ldots )$ and the first inner
product identity in (\ref{HLproduct}) ensures that this bi-algebra is
self-dual. In particular, one has the identities%
\begin{equation}
\Delta Q_{\lambda }=\sum_{\mu }Q_{\lambda /\mu }\otimes Q_{\mu }\qquad \text{%
and\qquad }Q_{\lambda /\mu }=\sum_{\nu }f_{\mu \nu }^{\lambda }Q_{\nu },
\label{HLcop}
\end{equation}%
where $\left\langle Q_{\lambda /\mu },P_{\nu }\right\rangle
_{t}=\left\langle Q_{\lambda },P_{\mu }P_{\nu }\right\rangle _{t}=f_{\mu \nu
}^{\lambda }(t)$ are the coefficients in (\ref{Hall}).

Using the well-known duality relation of Macdonald functions \cite[VI.5,
Equation (5.1), p327]{Macdonald} the analogous Hopf algebra structure can be
defined on the Macdonald $P^{\prime },Q^{\prime }$-functions: set $q=t$ and
define an automorphism $\omega _{t}:\Lambda (t,0)\rightarrow \Lambda (0,t)$
through the following table%
\begin{equation}
\begin{tabular}{||c||c|c|c|c|c|c|}
\hline
$F$ & $Q_{\lambda }^{\prime }$ & $P_{\lambda }^{^{\prime }}$ & $e_{\lambda }$
& $g_{\lambda }^{\prime }$ & $s_{\lambda }$ & $S_{\lambda }^{\prime }$ \\
\hline
$\omega _{t}F$ & $P_{\lambda ^{\prime }}$ & $Q_{\lambda ^{\prime }}$ & $%
g_{\lambda }$ & $e_{\lambda }$ & $S_{\lambda ^{\prime }}$ & $s_{\lambda
^{\prime }}$ \\ \hline
\end{tabular}
\label{omega}
\end{equation}%
where $F_{\lambda }^{\prime }:=F_{\lambda }(t,0)$ and $F_{\lambda ^{\prime
}}:=F_{\lambda ^{\prime }}(0,t)$ for $F=Q,P,g$. Each single column in the
table fixes $\omega _{t}$ uniquely as all the displayed functions are bases
in $\Lambda (t,0)$ and $\Lambda (t,0)$ respectively. Here we have introduced
the dual functions
\begin{equation}
S_{\lambda }=\det (g_{\lambda _{i}-i+j})\qquad \text{and}\qquad S_{\lambda
}^{\prime }=\det (g_{\lambda _{i}-i+j}^{\prime })\;.  \label{SS'}
\end{equation}%
of the Schur function with respect to the two inner products obtained from (%
\ref{Macprod}), $\left\langle S_{\lambda },s_{\mu }\right\rangle
_{0,t}=\langle s_{\lambda ^{\prime }},S_{\mu ^{\prime }}^{\prime }\rangle
_{q,0}=\delta _{\lambda \mu }$.

\begin{theorem}[Macdonald]
The map $\omega _{t}:\Lambda (t,0)\rightarrow \Lambda (0,t)$ defined via $%
Q_{\lambda }(t,0)\mapsto P_{\lambda ^{\prime }}(0,t)$ (resp. $P_{\lambda
}(t,0)\mapsto Q_{\lambda ^{\prime }}(0,t)$) is a $\mathbb{C}(t)$-bialgebra
isomorphism which preserves the inner product (\ref{HLproduct}). Thus, in
particular, we have that%
\begin{equation}
Q_{\lambda }^{\prime }Q_{\mu }^{\prime }=\sum_{\nu }f_{\lambda ^{\prime }\mu
^{\prime }}^{\nu ^{\prime }}(t)Q_{\nu }^{\prime }\quad \text{and\quad }%
P_{\lambda /\mu }^{\prime }=\sum_{\nu }f_{\mu ^{\prime }\nu ^{\prime
}}^{\lambda ^{\prime }}(t)P_{\nu }^{\prime }\;,  \label{dualHall}
\end{equation}%
where $\lambda ^{\prime },\mu ^{\prime },\nu ^{\prime }$ are the conjugate
partitions of $\lambda ,\mu ,\nu $.
\end{theorem}

\subsubsection{Generalised Cauchy identities}

We recall the following well known generalisation of Cauchy's identity to
Hall-Littlewood functions \cite[III.4]{Macdonald},%
\begin{eqnarray*}
\prod_{i,j}\frac{1-tx_{i}y_{j}}{1-x_{i}y_{j}} &=&\sum_{\lambda }g_{\lambda
}(x;t)m_{\lambda }(y) \\
&=&\sum_{\lambda }s_{\lambda }(x)S_{\lambda }(y;t)=\sum_{\lambda }Q_{\lambda
}(x;t)P_{\lambda }(y;t)\;.
\end{eqnarray*}%
Applying the inverse of the automorphism $\omega _{t}$ to the functions in
the $x$ variables once we obtain%
\begin{eqnarray}
\prod_{i,j}(1+x_{i}y_{j}) &=&\sum_{\lambda }e_{\lambda }(x)m_{\lambda }(y)
\notag \\
&=&\sum_{\lambda }s_{\lambda ^{\prime }}(x)s_{\lambda }(y)=\sum_{\lambda
}P_{\lambda ^{\prime }}^{\prime }(x;t)P_{\lambda }(y;t)  \label{C1}
\end{eqnarray}%
and doing so for a second time yields%
\begin{eqnarray}
\prod_{i,j}\frac{1}{(x_{i}y_{j};t)_{\infty }} &=&\sum_{\lambda }g_{\lambda
}^{\prime }(x;t)m_{\lambda }(y)  \notag \\
&=&\sum_{\lambda }S_{\lambda }^{\prime }(x;t)s_{\lambda }(y)=\sum_{\lambda
}Q_{\lambda ^{\prime }}^{\prime }(x;t)P_{\lambda ^{\prime }}^{\prime }(x;t)
\label{C2}
\end{eqnarray}%
where in order to arrive at the last relation we have swapped $x$ and $y$%
-variables.

\section{q-bosons and Yang-Baxter algebras}

In this section we introduce the basic noncommutative algebraic structure,
the $q$-oscillator or boson algebra, from which we will construct
step-by-step all the other relevant algebraic objects in the following
order: a $q$-Schur algebra of $\hat{U}_{n}$, solutions to the Yang-Baxter
equation and their associated Yang-Baxter algebras as well as noncommutative
analogues of symmetric polynomials.

\subsection{q-boson algebra}

There exist different versions of the $q$-boson algebra, also called the $q$%
-oscillator or Heisenberg algebra in the literature; see for instance
Chapter 5 in \cite{KlimykSchmuedgen} as well as references therein. We shall
work with the following version (we assume henceforth that $q^{\pm 1}$
exist); compare with the symmetric $q$-oscillator algebra in \cite[5.1.1,
Definition 2 and 5.1.2 ]{KlimykSchmuedgen}.

\begin{definition}[$q$-deformed boson algebra]
The $q$-boson algebra $\mathcal{H}_{q}$ is the unital, associative $%
\mathbb{C}(q)$-algebra defined in terms of the generators $\{q^{\pm N},\beta
,\beta ^{\ast }\}$ and the algebraic relations,%
\begin{eqnarray}
q^{N}q^{-N} &=&q^{-N}q^{N}=1,\quad q^{N}\beta =\beta q^{N-1},\qquad
q^{N}\beta ^{\ast }=\beta ^{\ast }q^{N+1},  \label{H1} \\
\beta \beta ^{\ast }-\beta ^{\ast }\beta &=&(1-q^{2})q^{2N},\qquad \beta
\beta ^{\ast }-q^{2}\beta ^{\ast }\beta =1-q^{2},  \label{H2}
\end{eqnarray}%
where $q^{\pm N}$ denote generators and $q^{pN+x}$ is shorthand for $(q^{\pm
N})^{p}q^{x}$.
\end{definition}

Note that (\ref{H2}) implies the relations
\begin{equation}
\beta ^{\ast }\beta =1-q^{2N}\qquad \text{and}\qquad \beta \beta ^{\ast
}=1-q^{2(N+1)}\;.  \label{crucial}
\end{equation}%
The proof of the following proposition is contained in \cite[5.1.1,
Proposition 1]{KlimykSchmuedgen}.

\begin{proposition}[basis of the $q$-boson algebra]
The set $\{(\beta ^{\ast })^{p}q^{rN},q^{rN}\beta ^{s}:p,s\in \mathbb{N}%
,~r\in \mathbb{Z},\mathbb{~}(r,s)\neq (-1,0)\}$ forms a basis of $\mathcal{H}%
_{q}$.
\end{proposition}

In what follows we will consider the $n$-fold tensor product of the $q$%
-oscillator algebra and denote by $\beta _{i}$, $\beta _{i}^{\ast }$, $%
q^{\pm N_{i}},$ $i=1,2,...,n$ the generators which belong to the $i^{\text{th%
}}$ factor of the tensor product $\mathcal{H}_{q}^{\otimes n}$. The set of
generators $\{\beta _{i},\beta _{i}^{\ast },q^{\pm N_{i}}\}_{i=1}^{n}$ then
obeys the following relations:%
\begin{eqnarray}
\beta _{i}\beta _{j}-\beta _{j}\beta _{i} &=&\beta _{i}^{\ast }\beta
_{j}^{\ast }-\beta _{j}^{\ast }\beta _{i}^{\ast
}=q^{N_{i}}q^{N_{j}}-q^{N_{j}}q^{N_{i}}=0  \label{qboson1} \\
q^{N_{i}}\beta _{j} &=&\beta _{j}q^{N_{i}-\delta _{ij}},\qquad
q^{N_{i}}\beta _{j}^{\ast }=\beta _{j}^{\ast }q^{N_{i}+\delta _{ij}},
\label{qboson2} \\
\beta _{i}\beta _{j}^{\ast }-\beta _{j}^{\ast }\beta _{i} &=&\delta
_{ij}(1-q^{2})q^{2N_{i}},\qquad \beta _{i}\beta _{i}^{\ast }-q^{2}\beta
_{i}^{\ast }\beta _{i}=1-q^{2}\;.  \label{qboson3}
\end{eqnarray}

The following proposition is a generalisation of the case $U_{q}\mathfrak{sl}%
(2)$ discussed in \cite[5.1.1, Proposition 3]{KlimykSchmuedgen}; compare
also with \cite{Hayashi} where closely related homomorphisms are discussed
for general affine Lie algebras.

\begin{proposition}[Jordan-Schwinger realisation]
Let $n>2$ and $z$ be an indeterminate with $\bar{z}=z^{-1}$. There exists a
homomorphism $h:U_{n}\rightarrow \mathcal{H}_{q}^{\otimes n}$ such that
\begin{equation}
E_{i}\mapsto -\frac{q^{-N_{i}}\beta _{i}^{\ast }\beta _{i+1}}{q-q^{-1}}%
,\qquad F_{i}\mapsto -\frac{q^{-N_{i+1}}\beta _{i}\beta _{i+1}^{\ast }}{%
q-q^{-1}},\qquad K_{i}^{\pm 1}\mapsto q^{\pm N_{i}}  \label{h1}
\end{equation}%
where $i=1,2,...,n-1$. This homomorphism can be extended to the affine
algebra $h_{z}:\hat{U}_{n}\rightarrow \mathcal{H}_{q}^{\otimes n}\otimes
\mathbb{C}[z,z^{-1}]$ by setting
\begin{equation}
E_{n}\mapsto -z\frac{q^{-N_{n}}\beta _{n}^{\ast }\beta _{1}}{q-q^{-1}}%
,\qquad F_{n}\mapsto -z^{-1}\frac{q^{-N_{1}}\beta _{n}\beta _{1}^{\ast }}{%
q-q^{-1}},\qquad K_{n}^{\pm 1}\mapsto q^{\mp N_{n}}\;.  \label{h2}
\end{equation}
\end{proposition}

\begin{proof}
The proof is a straightforward computation exploiting repeatedly the
relations (\ref{qboson1}), (\ref{qboson2}), (\ref{qboson3}) and (\ref%
{crucial}).
\end{proof}

The following infinite-dimensional highest-weight module, the Fock space, is
for instance discussed in \cite[Chapter 5, Section 5.2]{KlimykSchmuedgen} we
refer the reader to \emph{loc. cit.} for a proof. We will be using for
convenience the bra-ket notation from physics.

\begin{proposition}[Fock space]
Let $\mathcal{I}\subset \mathcal{H}_{q}$ be the two-sided ideal generated by
$\beta $ and $q^{N}-1$ and set $\mathcal{F}=\mathcal{H}_{q}/\mathcal{I}$.
(i) $\mathcal{F}$ has highest weight vector $|0\rangle =1+\mathcal{I}$ and
the set $\{~|m\rangle :=(\beta ^{\ast })^{m}/(q^{2})_{m}|0\rangle ~|\ m\in
\mathbb{Z}_{\geq 0}\}$ forms a basis. The following relations hold,%
\begin{equation}
q^{N}|m\rangle =q^{m}|m\rangle ,\quad \beta ^{\ast }|m\rangle
=(1-q^{2m+2})|m+1\rangle ,\quad \beta |m\rangle =|m-1\rangle \;.
\label{Fock}
\end{equation}%
(ii) The module $\mathcal{F}$ is simple as long as $q$ is not evaluated at a
root of unity.
\end{proposition}

In what follows we will also make use of the dual basis $\{\langle
m|\}_{m\in \mathbb{Z}_{\geq 0}}\subset \mathcal{\tilde{F}}$, i.e. $\langle
m|m^{\prime }\rangle =\delta _{m,m^{\prime }}$ and%
\begin{equation}
\langle m|q^{N}=q^{m}\langle m|,\quad \langle m|\beta ^{\ast
}=(1-q^{2m})\langle m-1|,\quad \langle m|\beta =\langle m+1|\;.
\label{dualFock}
\end{equation}%
It will be sometimes convenient to employ the following vector space
isomorphism $\imath :\mathcal{\tilde{F}}\rightarrow \mathcal{F}$ between the
Fock space and its dual: map the bra-vector $\langle m|$ onto the ket-vector
$(q^{2})_{m}|m\rangle $. This induces a scalar product on $\mathcal{F}$,
which by abuse of notation we also denote by $\langle \;|\;\rangle $ and
which we assume to be antilinear in the first factor. With respect to this
inner product $\beta $, $\beta ^{\ast }$ are adjoints of each other and $%
(q^{\pm N})^{\ast }=q^{\pm N}$.

Considering the $n$-fold tensor product $\mathcal{F}^{\otimes n}$ we can
parametrize the standard basis $\{|m_{1},\ldots ,m_{n}\rangle
:=|m_{1}\rangle \otimes \cdots \otimes |m_{n}\rangle :m_{i}\in \mathbb{Z}%
_{\geq 0}\}\subset \mathcal{F}^{\otimes n}$ in terms of partitions $\lambda
\in \mathcal{A}_{k,n}^{+}$: denote by $\mathcal{F}_{k}^{\otimes n}\subset
\mathcal{F}^{\otimes n}$ the subspace spanned by $\{~|\lambda \rangle
:\lambda \in \mathcal{A}_{k,n}^{+}\}$, where $m_{i}(\lambda )$ is the
multiplicity of the part $i$ in $\lambda $ and $|\lambda \rangle
:=|m_{1}(\lambda )\rangle \otimes \cdots \otimes |m_{n}(\lambda )\rangle $.
Obviously, we have $\mathcal{F}^{\otimes n}=\tbigoplus_{k\geq 0}\mathcal{F}%
_{k}^{\otimes n}$ with $\mathcal{F}_{0}^{\otimes n}=\mathbb{C}%
(q)|\varnothing \rangle \cong \mathbb{C}(q)$. We denote by $\cup _{k\geq
0}\{\langle \lambda |~:\lambda \in \mathcal{A}_{k,n}^{+}\}$\ the
corresponding dual basis with $\langle \lambda |\mu \rangle =\delta
_{\lambda \mu }$. N.B. the above vector space isomorphism $\imath :\mathcal{%
\tilde{F}}\rightarrow \mathcal{F}$ generalises trivially to $\imath :%
\mathcal{\tilde{F}}^{\otimes n}\rightarrow \mathcal{F}^{\otimes n}$ by
setting $\langle \lambda |~\mapsto b_{\lambda }(q^{2})|\lambda \rangle $.

\begin{remark}
The physical interpretation of the operators $\beta $, $\beta ^{\ast }$ is
that they annihilate and create a $q$-boson, respectively. The tensor
product $\mathcal{F}^{\otimes n}$ then describes highly
correlated quantum particles on a one-dimensional lattice with $n$-sites
with $m_{i}$ being the occupation number at site $i$.
\end{remark}

Because of the homomorphism (\ref{h1}) we can view $\mathcal{F}^{\otimes n}$
as a $U_{n}$-module. This module is reducible and there is a natural
decomposition into irreducible submodules. To describe the latter we recall
some known facts first; compare with \cite[Section 3, page 7]{Brundan}.

Recall that the vector representation $V=\mathbb{C}\{v_{1},...,v_{n}\}$ of $%
U_{n}$ associated with the fundamental weight $\omega _{1}$ is given by%
\begin{equation}
E_{i}v_{r}=\delta _{i,r-1}v_{r-1},\qquad F_{i}v_{r}=\delta
_{r,i}v_{r+1},\qquad K_{i}v_{r}=q^{\delta _{i,r}}v_{r},  \label{Vrep}
\end{equation}%
with$\;i=1,2,\ldots ,n-1$. One easily verifies, that $\bar{v}_{r}=v_{r}$ is
a compatible bar involution; see the discussion in \cite[Section 3]{Brundan}
on how to induce compatible bar involutions on tensor products of $V$%
.\smallskip

Define a right action of the Hecke algebra $H_{k}$ on $V^{\otimes k}$ with $%
M_{\mu }.T_{j}:=\mathcal{R}_{j,j+1}^{-1}M_{\mu }$ where $\{M_{\mu }:=v_{\mu
_{k}}\otimes \cdots \otimes v_{\mu _{2}}\otimes v_{\mu _{1}}:1\leq \mu
_{i}\leq n\}$ is the standard basis %
in $V^{\otimes k}$ and $%
\mathcal{R}^{-1}:V\otimes V\rightarrow V\otimes V$ is given by
\begin{equation}
\mathcal{R}^{-1}(v_{r}\otimes v_{s})=\left\{
\begin{array}{cc}
v_{s}\otimes v_{r}, & \text{if }r<s \\
q^{-1}v_{r}\otimes v_{r}, & \text{if }r=s \\
v_{s}\otimes v_{r}-(q-q^{-1})v_{r}\otimes v_{s}, & \text{if }r>s%
\end{array}%
\right.
\end{equation}%
for $r,s=1,\ldots ,n$ and we have $T_{i}^{2}=(q^{-1}-q)T_{i}+1$ as well as $%
T_{i}T_{i+1}T_{i}=T_{i+1}T_{i}T_{i+1}$. Given a permutation $\sigma \in
\mathfrak{S}_{k}$ set as usual $T_{\sigma }=T_{i_{1}}\cdots T_{i_{r}}$ where
$\sigma _{i_{1}}\cdots \sigma _{i_{r}}$ is the reduced expression of $\sigma
$ into elementary transpositions. Employing this action of the Hecke algebra
we now discuss two different versions of $q$-analogues of the symmetric
tensor algebra which we will then identify with $\mathcal{F}^{\otimes n}$
and its dual $\mathcal{\tilde{F}}^{\otimes n}$.

The $k^{\,\text{th}}$\emph{\ divided symmetric power} $S^{k}(V)$ is defined as $%
S^{k}(V)=V^{\otimes k}.X_{k}$ where
\begin{equation}
X_{k}=\sum_{w\in \mathfrak{S}_{k}}q^{\ell (w_{k})-\ell (w)}T_{w}  \label{Xk}
\end{equation}%
is bar-invariant and $w_{k}$ is the longest element in $\mathfrak{S}_{k}$
with $\ell (w_{k})=k(k-1)/2$. Note that $T_{i}X_{k}=X_{k}T_{i}=q^{-1}X_{k}$.
Set $S(V):=\tbigoplus_{k\in \mathbb{Z}_{\geq 0}}S^{k}(V)$.

In contrast the \emph{quantum symmetric tensor algebra} $\tilde{S}(V)$ of $V$
is the tensor algebra $T(V):=\tbigoplus_{k\geq 0}V^{\otimes k}$ divided by
the two sided ideal $I$ generated from the elements $\{v_{j}\otimes
v_{i}-q^{-1}v_{i}\otimes v_{j}:1\leq i<j\leq n\}$. Denote by $\tilde{S}%
^{k}(V)$ the $k$-th homogeneous component, that is the invariant subspace
under the natural action of the Hecke algebra on $V^{\otimes k}$.

\begin{proposition}[quantum symmetric tensor algebra]
There exist $U_{n}$-module isomorphisms $\mathcal{F}^{\otimes n}\cong S(V)$
and $\mathcal{\tilde{F}}^{\otimes n}\cong \tilde{S}(V)$ such that $\mathcal{F%
}_{k}^{\otimes n}$ and $\mathcal{\tilde{F}}_{k}^{\otimes n}$ are mapped onto
$S^{k}(V)$ and $\tilde{S}^{k}(V)$, respectively.
\end{proposition}

\begin{proof}
Following \cite[Section 5]{Brundan} define two bases%
\begin{equation}
X_{\lambda }=\frac{1}{[m_{1}]_q!\cdots \lbrack m_{n}]_q!}M_{\lambda }X_{k}\quad
\quad \text{and}\quad \text{\quad }\tilde{X}_{\lambda }:=\pi _{k}(M_{\lambda
})  \label{canbasis}
\end{equation}%
in $S^{k}(V)$ and $\tilde{S}^{k}(V)$, respectively, where $\lambda \in
\mathcal{A}_{k,n}^{+},$ $m_{i}=m_{i}(\lambda )$, $[m]_q:=(q^m-q^{-m})/(q-q^{-1})$ and $\pi _{k}:V^{\otimes
k}\twoheadrightarrow \tilde{S}^{k}(V)$ is the quotient map. We claim that
the maps $\mathcal{F}_{k}^{\otimes n}\ni |\lambda \rangle \mapsto X_{\lambda
}$ and $\mathcal{\tilde{F}}_{k}^{\otimes n}\ni \langle \lambda |\mapsto
\tilde{X}_{\lambda }$ are $U_{n}$-module isomorphisms with $U_{n}$ acting on
$\mathcal{F}_{k}^{\otimes n}$ via the homomorphism $h$ in (\ref{h1}) and on $%
\mathcal{\tilde{F}}_{k}^{\otimes n}$ via $h^{\ast }\circ \Theta $, where $%
\Theta $ is the algebra anti-automorphism $\Theta (E_{i})=F_{i},\;\Theta
(F_{i})=E_{i},\;\Theta (K_{i})=K_{i}$ and $h^{\ast }$ is the map obtained by
taking the adjoint of the image under $h$.

Exploiting the coproduct (\ref{coprod}) we compute the following action on a
monomial basis vector%
\begin{multline*}
\Delta ^{k}(E_{i})M_{\lambda ^{\prime }}=q^{m_{i+1}-1}M_{(\ldots ,\underset{%
m_{i}+1}{\underbrace{\tiny i,\ldots ,i}},\underset{m_{i+1}-1}{\underbrace{\small %
i+1,\ldots ,i+1}},\ldots )} \\
+q^{m_{i+1}-2}M_{(\ldots ,\underset{m_{i}}{\underbrace{\small i,\ldots ,i}},i+1,i,%
\underset{m_{i+1}-2}{\underbrace{\small i+1,\ldots ,i+1}},\ldots )}+\cdots \\
+M_{(\ldots ,\underset{m_{i}}{\underbrace{\small i,\ldots ,i}},\underset{m_{i+1}-1}{%
\underbrace{i+1,\ldots ,i+1}},i,\ldots )}
\end{multline*}%
Employing that for any permutation $\mu $ of $\lambda $ we have $M_{\mu
}X_{k}=q^{-\ell (\mu ,\lambda )}X_{\lambda }$ where $\ell (\mu ,\lambda )$
is the length of the shortest permutation which brings $\mu $ into $\lambda $
(see \cite[Section 5, eqn (5.3)]{Brundan}) we find that $\Delta
^{k}(E_{i})M_{\lambda }.X_{k}=[m_{i+1}]_{q}M_{\nu }.X_{k}$ as well as $%
\Delta ^{k}(E_{i})\pi _{k}(M_{\lambda })=[m_{i+1}]_{q}\pi _{k}(M_{\nu })$,
where $\nu \in \mathcal{A}_{k,n}^{+}$ is the partition obtained by removing
a part $(i+1)$ and adding a part $i$. Thus, $\Delta ^{k}(E_{i})X_{\lambda
}=[m_{i}+1]_{q}X_{\nu }$ and $\Delta ^{k}(E_{i})\tilde{X}_{\lambda
}=[m_{i+1}]_{q}\tilde{X}_{\nu }$. In comparison we have according to the
homomorphism (\ref{h1}) that%
\begin{equation*}
\frac{\beta _{i}^{\ast }\beta _{i+1}q^{-N_{i}}}{1-q^{2}}|\lambda \rangle
=[m_{i}+1]_{q}|\nu \rangle \quad \quad \text{and}\quad \quad \langle \lambda
|\frac{\beta _{i}\beta _{i+1}^{\ast }q^{-N_{i+1}}}{1-q^{2}}=\langle \nu
|[m_{i+1}]_{q}\;.
\end{equation*}%
The computation for the generator $F_{i}$ is similar and it is trivial for $%
K_{i}$.
\end{proof}

\begin{remark}
Let $\iota :S^{k}(V)\hookrightarrow V^{\otimes k}$ be the inclusion map and $%
\pi _{k}:V^{\otimes k}\twoheadrightarrow \tilde{S}^{k}(V)$ the quotient map.
There exists a bilinear form $\langle ~\cdot ~|~\cdot ~\rangle $ on $%
V^{\otimes k}$ which induces a pairing $\langle ~\cdot ~|~\cdot ~\rangle :%
\tilde{S}^{k}(V)\times S^{k}(V)\rightarrow \mathbb{C}(q)$ by setting $%
\langle Y|X\rangle :=\langle Y|\iota (X)\rangle =\langle \pi
_{k}(Y),X\rangle $ for $X\in S^{k}(V),Y\in V^{\otimes k}$; see \cite[Section
5, para after eqn (5.2)]{Brundan} and references therein. This pairing
coincides with bracket of the Fock space and its dual, i.e. $\langle
M_{\lambda }|\iota (X_{\mu })\rangle =\langle \pi _{k}(M_{\lambda })|X_{\mu
}\rangle =\left\langle \lambda |\mu \right\rangle =\delta _{\lambda \mu }$.
\end{remark}

We are now turning to the affine algebra and consider for simplicity $%
\mathbf{\hat{U}}_{n}$ instead of $\hat{U}_{n}$. First we remind the reader
that $V$ can be turned into a so-called evaluation module $V(a)=V\otimes
\mathbb{C}[a,a^{-1}]$ for $\mathbf{\hat{U}}_{n}$ setting%
\begin{equation}
E_{n}v_{r}=a~\delta _{r,1}v_{n},\qquad F_{n}v_{r}=a^{-1}\delta
_{r,n}v_{1},\qquad K_{n,1}v_{r}=q^{\delta _{r,n}-\delta _{r,1}}v_{r}\;.
\end{equation}%
We now have the following:

\begin{proposition}[Kirillov-Reshetikhin module isomorphism]
\label{KRmodule}The vector space $\mathcal{F}_{k}^{\otimes n}\otimes \mathbb{%
C}[z,z^{-1}]$ viewed as $\mathbf{\hat{U}}_{n}$-module\ is isomorphic to the
Kirillov-Reshetikhin module $W^{1,k}=W(k\omega _{1})$. That is, it coincides
with the irreducible submodule of highest weight $k\omega _{1}$ of the
following (reducible) tensor product of $\mathbf{\hat{U}}_{n}$ evaluation
modules,
\begin{equation}
V(zq^{-k+1})\otimes V(zq^{-k+3})\otimes \cdots \otimes V(zq^{k-1})\;.
\label{KR}
\end{equation}
\end{proposition}

\begin{proof}
Employing results in \cite{CP95,CP96,CP98} on the classification of
finite-dimensional type 1 representations of $\mathbf{\hat{U}}_{n}$ ($\tilde{%
K}=\tilde{K}_{1}^{-1}\tilde{K}_{2}^{-1}\cdots \tilde{K}_{n-1}^{-1}$) and the
previous result $\mathcal{F}_{k}^{\otimes n}\cong S^{k}(V),$ it suffices to
compute the evaluation parameters. Consider in $\mathcal{A}_{k,n}^{+}$ the
partitions $\lambda =n^{k}=(n,n,\ldots ,n)$ and $\mu =(n,n,\ldots ,n,1)$.
Under the previously stated isomorphism of $U_{n}$-modules $|\lambda \rangle
$ and $|\mu \rangle $ are mapped to the following vectors in $S^{k}(V)$%
\begin{equation*}
|\lambda \rangle \mapsto \frac{1}{[k]!}v_{n}\otimes \cdots \otimes
v_{n}\quad \text{and\quad }|\mu \rangle ~\mapsto \frac{1}{[k-1]!}%
\sum_{i=1}^{k}q^{i-1}v_{n}\otimes \cdots \otimes \underset{i}{v_{1}}\otimes
\cdots \otimes v_{n}\;.
\end{equation*}%
The $\mathbf{\hat{U}}_{n}$-action on $\mathcal{F}_{k}^{\otimes n}\otimes
\mathbb{C}[z,z^{-1}]$ yields (compare with (\ref{h2}))%
\begin{equation*}
h_{z}(E_{0})|\mu \rangle =z[k]_{q}|\lambda \rangle \quad \text{and\quad }%
h_{z}(F_{0})|\lambda \rangle =z^{-1}|\mu \rangle \;.
\end{equation*}%
In comparison the $\mathbf{\hat{U}}_{n}$-action on $V(z_{1})\otimes
V(z_{2})\otimes \cdots \otimes V(z_{k})$ gives
\begin{gather*}
\Delta ^{k}(E_{0})v_{n}\otimes \cdots \otimes \underset{i}{v_{1}}\otimes
\cdots \otimes v_{n}=z_{i}q^{k-i}v_{n}\otimes \cdots \otimes v_{n}, \\
\Delta ^{k}(F_{0})v_{n}\otimes \cdots \otimes
v_{n}=\sum_{i=1}^{k}z_{i}^{-1}q^{i-1}v_{n}\otimes \cdots \otimes \underset{i}%
{v_{1}}\otimes \cdots \otimes v_{n}\;.
\end{gather*}%
Hence, acting with $X_{k}$ from the right we must have that $%
z_{i}=zq^{-k+2i-1}$ as asserted. From this it follows that the Drinfeld
polynomials \cite{CP95,CP96,CP98} are given by $%
P_{1}(z)=(1-zq^{-k+1})(1-zq^{-k+3})\cdots (1-zq^{k-1})$ and $P_{l}=1$ for $%
l\neq 1$ which fixes the Kirillov-Reshetikhin module up to isomorphism.
\end{proof}

\subsection{Solutions to the Yang-Baxter equation}

We now discuss three particular solutions to the Yang-Baxter equation in
terms of the $q$-boson algebra. The first one has been previously obtained
by Bogoliubov, Izergin and Kitanine \cite{Bogoliubovetal}, the others are
new. They are special limits of solutions related to the $U_{q}\widehat{%
\mathfrak{sl}}(2)$ R-matrix \cite{Korff}.

Let $u$ be an invertible variable, called the spectral variable. Define \cite%
{Bogoliubovetal}%
\begin{equation}
L(u)=\left(
\begin{array}{cc}
1 & u~\beta ^{\ast } \\
\beta & u~1%
\end{array}%
\right) \in \limfunc{End}[\mathbb{C}^{2}(u)]\otimes \mathcal{H}_{q}\ ,
\label{L}
\end{equation}%
where the notation is shorthand for $L(u)=\bigl(%
\begin{smallmatrix}
1 & 0 \\
0 & 0%
\end{smallmatrix}%
\bigr)\otimes 1+u\bigl(%
\begin{smallmatrix}
0 & 1 \\
0 & 0%
\end{smallmatrix}%
\bigr)\otimes \beta ^{\ast }+\bigl(%
\begin{smallmatrix}
0 & 0 \\
1 & 0%
\end{smallmatrix}%
\bigr)\otimes \beta +u\bigl(%
\begin{smallmatrix}
0 & 0 \\
0 & 1%
\end{smallmatrix}%
\bigr)\otimes 1$.

\begin{proposition}[Bogoliubov, Izergin, Kitanine]
The L-operator satisfies the Yang-Baxter equation%
\begin{equation}
R_{12}(u,v)L_{1}(u)L_{2}(v)=L_{2}(v)L_{1}(u)R_{12}(u,v),  \label{ybe}
\end{equation}%
where $R\in \limfunc{End}(\mathbb{C}(u)^{2}\otimes \mathbb{C}(v)^{2})$ $%
\cong \limfunc{End}\mathbb{C}(u,v)^{4}$ is given by
\begin{equation}
R(u,v)=\left(
\begin{array}{cccc}
\frac{u-tv}{u-v} & 0 & 0 & 0 \\
0 & t & \frac{1-t}{u-v}u & 0 \\
0 & \frac{1-t}{u-v}v & 1 & 0 \\
0 & 0 & 0 & \frac{u-tv}{u-v}%
\end{array}%
\right) \   \label{R}
\end{equation}%
with respect to the basis $\{v_{0}\otimes v_{0},v_{0}\otimes
v_{1},v_{1}\otimes v_{0},v_{1}\otimes v_{1}\}$ and we have set $t=q^{2}$.
\end{proposition}

\begin{proof}
Observe that $L$ is a $2\times 2$ matrix with entries
\begin{equation*}
\langle \sigma ^{\prime }|L(u)|\sigma \rangle =u^{\sigma }\frac{\beta
^{\sigma ^{\prime }}(\beta ^{\ast })^{\sigma }-t(\beta ^{\ast })^{\sigma
}\beta ^{\sigma ^{\prime }}}{1-t},\quad \sigma ^{\prime }\sigma =0,1\;.
\end{equation*}%
Making a case-by-case distinction and using repeatedly that $\beta \beta
^{\ast }-t\beta ^{\ast }\beta =1-t$ the proof is a straightforward
computation.
\end{proof}

We now define a second solution to the Yang-Baxter equation, which we then
relate to the first, in terms of an \textquotedblleft infinite-dimensional
matrix\textquotedblright . Setting $\mathcal{F}((u))=\mathbb{C}((u))\otimes
\mathcal{F}$ let $L^{\prime }(u)\in \limfunc{End}[\mathcal{F}((u))]\otimes
\mathcal{H}_{q}$ be given by ($t=q^{2}$)%
\begin{equation}
\langle m^{\prime }|L^{\prime }(u)|m\rangle =\frac{u^{m}(\beta ^{\ast
})^{m}\beta ^{m^{\prime }}}{(t)_{m}}\;,  \label{L'}
\end{equation}%
where $|m\rangle $ respectively $\langle m^{\prime }|$ label the basis and
dual basis in the Fock space $\mathcal{F}$ of the $q$-boson algebra $%
\mathcal{H}_{q}$. Define another operator $R^{\prime }\in \limfunc{End}[%
\mathcal{F}((u))\otimes \mathcal{F}((v))]$ via%
\begin{equation}
\langle m_{1}^{\prime },m_{2}^{\prime }|R^{\prime }(u,v)|m_{1},m_{2}\rangle
=\left( \frac{u}{v}\right) ^{m_{1}}\QATOPD[ ] {m_{2}^{\prime }}{m_{1}}%
_{t}(u/v;t)_{m_{2}^{\prime }-m_{1}}~\delta _{m_{1}+m_{2},m_{1}^{\prime
}+m_{2}^{\prime }}  \label{R'}
\end{equation}%
where the $q$-deformed binomial coefficient is zero if $m_{2}^{\prime
}<m_{1} $, hence only terms survive for which $m_{2}^{\prime }\geq m_{1}$.

\begin{proposition}
We have the identity%
\begin{equation}
R_{12}^{\prime }(u,v)L_{1}^{\prime }(u)L_{2}^{\prime }(v)=L_{2}^{\prime
}(v)L_{1}^{\prime }(u)R_{12}^{\prime }(u,v)\;.  \label{ybe'}
\end{equation}%
Moreover, the $R^{\prime }$-operator is invertible: let $P:\mathcal{F}%
\otimes \mathcal{F}\rightarrow \mathcal{F}\otimes \mathcal{F}$ be the flip
operator $P|m_{1},m_{2}\rangle =|m_{2},m_{1}\rangle $, then $R^{\prime
}(u,v)PR^{\prime }(v,u)P=1$.
\end{proposition}

\begin{proof}
Exploiting (\ref{H1}), (\ref{H2}) and (\ref{crucial}) one first proves via
induction the relation%
\begin{equation*}
\beta ^{a}\tilde{\beta}^{b}=\sum_{r=0}^{\min (a,b)}t^{(a-r)(b-r)}\QATOPD[ ] {%
b}{r}\frac{[a]!}{[a-r]!}~\tilde{\beta}^{b-r}\beta ^{a-r},\qquad \tilde{\beta}%
:=\frac{\beta ^{\ast }}{1-t}
\end{equation*}%
with $[x]:=(1-t^{x})/(1-t)$ and then verifies the assertion after a somewhat
tedious but straightforward computation whose details we omit.
\end{proof}

To relate the two solutions, $L$ and $L^{\prime }$, we need yet another
operator $R^{\prime \prime }\in \limfunc{End}[\mathbb{C}(u)^{2}\otimes
\mathcal{F}(v)]$ defined as%
\begin{equation}
R^{\prime \prime }(u,v)=L(u/v)+%
\begin{pmatrix}
u/v~q^{2N} & 0 \\
0 & 0%
\end{pmatrix}%
=%
\begin{pmatrix}
1+u/v~q^{2N} & u/v~\beta ^{\ast } \\
\beta & u/v\cdot 1%
\end{pmatrix}%
~.  \label{R''}
\end{equation}

\begin{proposition}
The $R^{\prime \prime }$ operator satisfies the identity%
\begin{equation}
R_{12}^{\prime \prime }(u,v)L_{1}(u)L_{2}^{\prime }(v)=L_{2}^{\prime
}(v)L_{1}(u)R_{12}^{\prime \prime }(u,v)  \label{ybe''}
\end{equation}%
and possesses the inverse%
\begin{equation}
(R^{\prime \prime })^{-1}(u,v)=\frac{1}{1+u/v}%
\begin{pmatrix}
q^{-2N} & -q^{-2N}\beta ^{\ast } \\
-v/u~\beta q^{-2N} & 1+v/u~q^{-2N-2}%
\end{pmatrix}%
\;.  \label{Rtilde}
\end{equation}
\end{proposition}

\begin{proof}
Once more the claimed identities are a direct consequence of the $q$%
-boson algebra relations (\ref{H1}), (\ref{H2}) and (\ref{crucial}).
\end{proof}

\subsection{The Yang-Baxter algebra}

The Yang-Baxter equation is naturally endowed with a coproduct: one easily
verifies that given the solutions $L,L^{\prime }$ the operators $\Delta
L=L_{2}L_{1}\in \limfunc{End}[\mathbb{C}^{2}(u)]\otimes \mathcal{H}%
_{q}^{\otimes 2}$ and $\Delta L^{\prime }=L_{2}^{\prime }L_{1}^{\prime }\in
\limfunc{End}[\mathcal{F}((u))]\otimes \mathcal{H}_{q}^{\otimes 2}$ are also
solutions to the quantum Yang-Baxter equation. Repeating the argument it is
therefore natural to consider the so-called monodromy matrices
\begin{equation}
T(u)=\Delta ^{n}L(u)=L_{n}(u)L_{n-1}(u)\cdots L_{1}(u)\in \limfunc{End}[%
\mathbb{C}^{2}(u)]\otimes \mathcal{H}_{q}^{\otimes n}  \label{T}
\end{equation}%
and%
\begin{equation}
T^{\prime }(u)=\Delta ^{n}L^{\prime }(u)=L_{n}^{\prime }(u)L_{n-1}^{\prime
}(u)\cdots L_{1}^{\prime }(u)\in \limfunc{End}[\mathcal{F}((u))]\otimes
\mathcal{H}_{q}^{\otimes n}\;.  \label{T'}
\end{equation}%
Much of the discussion which is to follow will focus on the matrix elements
of these two operators. We start the discussion with (\ref{T}).

Rewrite $T(u)=\left(
\begin{array}{cc}
A(u) & B(u) \\
C(u) & D(u)%
\end{array}%
\right) $ and for $\mathcal{O}=A,B,C,D\in \mathbb{C}[u]\otimes \mathcal{H}%
_{q}^{\otimes n}$ introduce the series expansions $\mathcal{O}%
(u)=\sum_{r\geq 0}u^{r}\mathcal{O}_{r}$. Note that the latter terminate for $%
r>n$ according to the definition (\ref{T}) and the $L$-operator (\ref{L}).

\begin{definition}
The Yang-Baxter algebra $\mathfrak{A}\subset \mathcal{H}_{q}^{\otimes n}$\
is the algebra generated by $\{A_{r},B_{r},C_{r},D_{r}\}_{r\geq 0}$ subject
to the commutation relations imposed by the Yang-Baxter equation (\ref{ybe}).
\end{definition}

Let $u,v$ be two independent variables then one easily checks that (\ref{ybe}%
) entails the relations
\begin{equation}
\mathcal{O}(u)\mathcal{O}(v)=\mathcal{O}(v)\mathcal{O}(u),\qquad \mathcal{O}%
=A,B,C,D  \label{aba0}
\end{equation}%
and setting ($t=q^{2}$)
\begin{eqnarray}
(u-v)A(u)B(v) &=&(ut-v)B(v)A(u)+(1-t)vB(u)A(v),  \label{aba1} \\
\left( u-v\right) D(u)B(v) &=&(u-tv)B(v)D(u)-(1-t)vB(u)D(v),  \label{aba2} \\
C(u)B(v)-tB(v)C(u) &=&\frac{(1-t)v}{u-v}~[A(v)D(u)-A(u)D(v)]\ .  \label{aba3}
\end{eqnarray}

\begin{proposition}
Introduce the co-product $\Delta :\mathfrak{A}\rightarrow \mathfrak{A}\times
\mathfrak{A}$,%
\begin{eqnarray}
\Delta A(x) &=&A(x)\otimes A(x)+C(x)\otimes B(x),  \notag \\
\Delta B(x) &=&B(x)\otimes A(x)+D(x)\otimes B(x),  \notag \\
\Delta C(x) &=&A(x)\otimes C(x)+C(x)\otimes D(x),  \notag \\
\Delta D(x) &=&B(x)\otimes C(x)+D(x)\otimes D(x)  \label{YBcop}
\end{eqnarray}%
and the co-unit $\varepsilon :\mathfrak{A}\rightarrow \mathbb{C}$,
\begin{equation*}
\varepsilon (A)=\varepsilon (D)=1\quad \text{and\quad }\varepsilon
(C)=\varepsilon (B)=0\,.
\end{equation*}%
Then $(\mathfrak{A},\Delta ,\varepsilon )$ is a bialgebra. That is, we have
the identities
\begin{equation*}
(\Delta \otimes 1)\Delta =(1\otimes \Delta )\Delta \qquad \text{and\qquad }%
\left( \varepsilon \otimes 1\right) \Delta =(1\otimes \varepsilon )\Delta
\cong 1\ .
\end{equation*}
\end{proposition}

\begin{proof}
All stated bi-algebra axioms are easily checked via a straightforward
computation.
\end{proof}

\begin{lemma}
We have the dependencies%
\begin{eqnarray}
(1-t)B(u) &=&u~A(u)\beta _{1}^{\ast }-ut\beta _{1}^{\ast }A(u)=D(u)\beta
_{n}^{\ast }-t\beta _{n}^{\ast }D(u)  \label{BAD} \\
(1-t)uC(u) &=&u\beta _{n}A(u)-utA(u)\beta _{n}=\beta _{1}D(u)-tD(u)\beta _{1}
\label{CAD}
\end{eqnarray}
\end{lemma}

\begin{proof}
The assertion is easily proved via induction employing the coproduct (\ref%
{YBcop}).
\end{proof}

In analogy with (\ref{T}) decompose the monodromy matrix as the sum of the
form $T^{\prime }(u)=\sum_{X,Y}X(u)\otimes Y$ with $X(u)\in \limfunc{End}%
\mathcal{F}((u)),$ $Y\in \mathcal{H}_{q}^{\otimes n}$ and consider the
matrix elements of the monodromy matrix, $T_{m^{\prime }m}^{\prime
}(u):=\sum_{X,Y}\langle m^{\prime }|X(u)|m\rangle Y$. Then as a direct
consequence of (\ref{ybe''}) we have the following identities.

\begin{corollary}
\label{YBQ}One deduces the following commutation relations of $T_{m^{\prime
}m}^{\prime }(u)$ with the Yang-Baxter algebra generators,%
\begin{multline}
(ut^{m^{\prime }}+v)A(u)T_{m^{\prime }m}^{\prime }(v)-(ut^{m}+v)T_{m^{\prime
}m}^{\prime }(v)A(u)= \\
v(1-t^{m})T_{m^{\prime },m-1}^{\prime }(v)B(u)-uC(u)T_{m^{\prime
}-1,m}^{\prime }(v)  \label{AT'}
\end{multline}%
\begin{equation}
T_{m^{\prime }m-1}^{\prime }(v)B(u)=(ut^{m^{\prime }}+v)B(u)T_{m^{\prime
}m-1}^{\prime }(v)
-uT_{m^{\prime }m}^{\prime }(v)A(u)+uD(u)T_{m^{\prime }-1,m-1}^{\prime }(v)
\label{BT'}
\end{equation}%
\begin{multline}
uC(u)T_{m^{\prime }-1,m}^{\prime }(v)=(ut^{m}+v)T_{m^{\prime }-1,m}^{\prime
}(v)C(u) \\
-v(1-t^{m})T_{m^{\prime }-1,m-1}^{\prime }(v)D(u)+v(1-t^{m^{\prime
}})A(u)T_{m^{\prime }m}^{\prime }(v)  \label{CT'}
\end{multline}%
\begin{equation}
uD(u)T_{m^{\prime }m}^{\prime }(v)-uT_{m^{\prime }m}^{\prime }(v)D(u)=
uT_{m^{\prime }m+1}^{\prime }(v)C(u)-v(1-t^{m^{\prime }})B(u)T_{m^{\prime
}+1,m}^{\prime }(v)  \label{DT'}
\end{equation}
\end{corollary}

\subsection{Affine quantum plactic polynomials}

We now state simple polynomial expressions for the matrix elements of both
monodromy matrices, (\ref{T}) and (\ref{T'}), in the generators of an
algebra $\mathcal{\hat{U}}_{n}^{-}\subset h_{z}(\hat{U}_{q}\mathfrak{b}^{-})$
which is contained in the image of the lower Borel algebra $\hat{U}_{q}%
\mathfrak{b}^{-}\subset \hat{U}_{n}$ under the homomorphism (\ref{h2}).

\begin{corollary}
Denote by $\mathcal{\hat{U}}_{n}^{-}\subset \mathcal{H}_{q}^{\otimes n}$ the
subalgebra generated by the letters
\begin{equation}
a_{i}=\beta _{i+1}^{\ast }\beta _{i},\quad i=1,\dots,n-1\qquad \text{and\qquad }a_{n}=z\beta
_{1}^{\ast }\beta _{n}\;.  \label{qplacticrep}
\end{equation}%
Then we have non-local commutativity,%
\begin{equation}
a_{i}a_{j}=a_{j}a_{i}\quad \text{for}\quad |i-j|\func{mod}n>1\;,
\label{nonlocal}
\end{equation}%
and what we call the \textquotedblleft quantum Knuth
relations\textquotedblright\ ($t=q^{2}$)%
\begin{eqnarray}
a_{i+1}a_{i}^{2}+ta_{i}^{2}a_{i+1} &=&(1+t)a_{i}a_{i+1}a_{i},  \notag \\
a_{i+1}^{2}a_{i}+ta_{i}a_{i+1}^{2} &=&(1+t)a_{i+1}a_{i}a_{i+1}\;,
\label{qKnuth}
\end{eqnarray}%
where all indices are understood modulo $n$. Denote by $\mathcal{U}_{n}^{-}$
the non-affine algebra generated by $\{a_{1},\ldots ,a_{n-1}\}$.
\end{corollary}

\begin{proof}
The assertion is a direct consequence of the quantum Serre relations (\ref%
{qSerre}) and the homomorphisms (\ref{h1}) and (\ref{h2}).
\end{proof}

\begin{remark}
Setting formally $t=0$ we recover the faithful representation of the local
affine plactic algebra considered in \cite[Def 5.4 and Prop 5.8]{KS}.
Setting $t=0$ and $a_{n}=0$ one obtains a representation of the local finite
plactic algebra; compare with \cite{FG}. The (non-local) plactic algebra was
introduced in \cite{LasSchutz} and is intimately linked to the
Robinson-Schensted-Knuth correspondence \cite{FultonYT}. We therefore shall
refer to $\mathcal{\hat{U}}_{n}^{-}$ as \emph{affine quantum plactic algebra}%
. Note that this definition is different from the construction in \cite[%
Section 4.7]{KrobThibon} using the quantum coordinate ring.
\end{remark}

To describe the matrix elements of the monodromy matrices we require the
notion of cyclically ordered words $i_{1}\ldots i_{r}$ with letters $%
i_{j}\in \{1,\ldots ,n\}$. A word $i_{1}\ldots i_{r}$ is \emph{%
anti-clockwise cyclically ordered} if for any two indices $i_{j}$, $i_{k}$
with $i_{k}=i_{j}+1$ modulo $n$, the $i_{j}$ occurs to the right of $i_{k}$.
(In case $i_{k}\neq i_{j}+1$ the order does not matter because of %
\eqref{nonlocal}.) The origin of the name becomes obvious if we identify the
letter $i_{j}$ with labels of the nodes of the $\widehat{\mathfrak{sl}}(n)$
Dynkin diagram: there are two circle segments connecting the two points. If
they are not of the same length we choose the shorter one and the
anti-clockwise order is the same as the intuitively defined anti-clockwise
order with respect to this segment. For any word $i_{1}\ldots i_{r}$ not
containing all the letters in $\{1,\ldots ,n\}$, there is a unique
anti-clockwise cyclically ordered word which differs only by a permutation;
compare with \cite[Section 5.3]{KS}.

\begin{proposition}[noncommutative elementary symmetric polynomials]
One has the identity%
\begin{equation}
\boldsymbol{E}(u):=A(u)+zD(u)=\sum_{r=0}^{n}u^{r}\boldsymbol{e}_{r},
\label{ncE}
\end{equation}%
where we set $\boldsymbol{e}_{n}=z~1$ and for $r<n$%
\begin{equation}
\boldsymbol{e}_{r}=(1-t)^{1-r}\sum_{w=i_{1}\cdots
i_{r}}[a_{i_{1}},[a_{i_{2}},...[a_{i_{r-1}},a_{i_{r}}]_{t}...]_{t}]_{t}
\label{nce}
\end{equation}%
with $[x,y]_{t}:=xy-t~yx$ and the sum runs over all cyclically ordered words
$w$ with distinct elements.
\end{proposition}

\begin{example}
Set $n=4$ and $r=3$, then the affine quantum plactic elementary symmetric
polynomial reads%
\begin{equation*}
(1-t)^{2}\boldsymbol{e}%
_{r}=[a_{3},[a_{2},a_{1}]_{t}]_{t}+[a_{4},[a_{3},a_{2}]_{t}]_{t}+[a_{1},[a_{4},a_{3}]_{t}]_{t}+[a_{2},[a_{1},a_{4}]_{t}]_{t}\;.
\end{equation*}
\end{example}

\begin{remark}
In the commutative case, making the formal replacement $a_{i}\rightarrow
x_{i}$ with $x_{i}x_{j}=x_{j}x_{i}$ for all $i,j$, the noncommutative
polynomial $\boldsymbol{e}_{r}$ becomes the standard, commutative elementary
symmetric polynomial.
\end{remark}

\begin{proof}
Employing the coproduct of the Yang-Baxter algebra one easily verifies via
induction the formulae%
\begin{eqnarray*}
A(u) &=&\sum_{0\leq r\leq \frac{n}{2}}~\sum_{1\leq
i_{1}<j_{1}<...<i_{r}<j_{r}\leq n}u^{(j_{1}-i_{1})+\cdots
+(j_{r}-i_{r})}a_{i_{1},j_{1}}\cdots a_{i_{r},j_{r}}, \\
D(u) &=&\sum_{0\leq r\leq \frac{n}{2}}~\sum_{1\leq
i_{1}<j_{1}<...<i_{r}<j_{r}\leq n}u^{n-(j_{1}-i_{1})-\cdots
-(j_{r}-i_{r})}a_{j_{1},i_{1}}\cdots a_{j_{r},i_{r}}\;,
\end{eqnarray*}%
where $a_{i,j}:=\beta _{i}\beta _{j}^{\ast }$. The remainder of the proof is
now a consequence of the following lemma which can be proved via a
straightforward computation.

\begin{lemma}
For $i<j$ we have%
\begin{equation}
a_{i,j}=(1-q)^{i-j+1}[a_{j-1},[a_{j-2},...[a_{i+1},a_{i}]_{t}...]_{t}]_{t}
\label{aij}
\end{equation}%
and%
\begin{equation}
a_{j,i}=z^{-1}(1-q)^{i-j-1}[a_{i-1},[a_{i-2},...[a_{1},[a_{n},[a_{n-1},...[a_{j+1},a_{j}]_{t}...]_{t}]_{t}\ .
\label{aji}
\end{equation}
\end{lemma}
\end{proof}

Also the matrix elements of the other monodromy matrix (\ref{T'}) can be
expressed as noncommutative analogues of a family of symmetric functions;
compare with (\ref{g'}) and (\ref{G'}).

\begin{proposition}[noncommutative Rogers-Szeg\"{o} polynomials]
One has the formal power series expansion%
\begin{equation}
\boldsymbol{G}^{\prime }(u):=\sum_{m\geq 0}z^{m}T_{m,m}^{\prime
}=\sum_{r\geq 0}u^{r}\boldsymbol{g}_{r}^{\prime },\qquad \boldsymbol{g}%
_{r}^{\prime }:=\sum_{\lambda \vdash r}\frac{\boldsymbol{m}_{\lambda }}{%
(t)_{\lambda }}~,  \label{ncG'}
\end{equation}%
where the quantum plactic analogues of the \emph{symmetric monomial functions%
\footnote{%
Note that these polynomials do not commute in general.}} are defined for
partitions $\lambda $ of length $\leq n$ as%
\begin{equation}
\boldsymbol{m}_{\lambda }:=\sum_{w}(z\beta _{1}^{\ast })^{\lambda
_{w_{n}}}a_{1}^{\lambda _{w_{1}}}\cdots a_{n-1}^{\lambda _{w_{n-1}}}\beta
_{n}^{\lambda _{w_{n}}}  \label{ncm}
\end{equation}%
with the last sum running over all \emph{distinct} permutations of $\lambda $%
.
\end{proposition}

\begin{proof}
According to its definition it is easy to verify that the matrix elements $%
T_{m^{\prime },m}^{\prime }(u)\in \mathbb{C}[\![u]\!]\otimes \mathcal{H}%
_{q}^{\otimes n}$ of the monodromy matrix can be written as%
\begin{equation}
T_{m^{\prime },m}^{\prime }(u)=\frac{u^{m}}{(t)_{m}}\sum_{m_{1},\ldots
,m_{n-1}\geq 0}\frac{u^{m_{1}+\cdots +m_{n-1}}}{(t)_{m_{1}}\cdots
(t)_{m_{n-1}}}~(\beta _{1}^{\ast })^{m}a_{1}^{m_{1}}a_{2}^{m_{2}}\cdots
a_{n-1}^{m_{n-1}}\beta _{n}^{m^{\prime }}\;.  \label{T'matrix}
\end{equation}%
Setting $m=m^{\prime }$ the assertion follows.
\end{proof}

Note that despite the sum in (\ref{ncG'}) being infinite, only a finite
number of terms survive when acting on a vector in $\mathcal{F}^{\otimes n}$%
. Thus, the \textquotedblleft generating function\textquotedblright\ $\boldsymbol{G}'$, which is a formal series here, becomes a well defined operator in $\limfunc{End}(%
\mathcal{F}^{\otimes n})$; see (\ref{ncG'2}) below.

\begin{corollary}[Integrability]
All of the quantum plactic symmetric polynomials defined above commute
pairwise, i.e. for any $r,r^{\prime }\geq 0$ we have%
\begin{equation}
\boldsymbol{e}_{r}\boldsymbol{e}_{r^{\prime }}=\boldsymbol{e}_{r}\boldsymbol{%
e}_{r^{\prime }},\quad \boldsymbol{g}_{r}^{\prime }\boldsymbol{g}_{r^{\prime
}}^{\prime }=\boldsymbol{g}_{r^{\prime }}^{\prime }\boldsymbol{g}_{r},\quad
\boldsymbol{e}_{r}\boldsymbol{g}_{r^{\prime }}^{\prime }=\boldsymbol{g}%
_{r^{\prime }}^{\prime }\boldsymbol{e}_{r}\;.  \label{integrable}
\end{equation}%
This statement remains in particular true for $z=0$, that is, for the finite
plactic polynomials obtained by setting formally $a_{n}\equiv 0$ in $%
\boldsymbol{e}_{r}$ and considering only partitions $\lambda $ with length $%
<n$ in (\ref{ncm}) to compute $\boldsymbol{g}_{r}^{\prime }$.
\end{corollary}

\begin{proof}
This is a direct consequence of the Yang-Baxter equations (\ref{ybe}), (\ref%
{ybe'}), (\ref{ybe''}) and observing that the $R$-matrices (\ref{R}), (\ref%
{R'}), (\ref{R''}) are compatible with quasi-periodic boundary conditions,
\begin{equation*}
\lbrack R,\sigma \otimes 1+1\otimes \sigma ]=[R^{\prime },N\otimes
1+1\otimes N]=[R^{\prime \prime },\sigma \otimes 1+1\otimes N]=0,
\end{equation*}%
where $\sigma =\bigl(%
\begin{smallmatrix}
0 & 0 \\
0 & 1%
\end{smallmatrix}%
\bigr)$ and $N|m\rangle =m|m\rangle $ is the occupation or number operator.
Using the above relations one shows with the help of (\ref{ybe}) that%
\begin{equation*}
z^{\sigma _{1}}T_{1}(u)z^{\sigma _{2}}T_{2}(v)=R_{12}^{-1}(u/v)z^{\sigma
_{2}}T_{2}(v)z^{\sigma _{1}}T_{1}(u)R_{12}(u/v)
\end{equation*}%
and, hence, $\boldsymbol{E}(u)\boldsymbol{E}(v)=\boldsymbol{E}(v)\boldsymbol{%
E}(u)$ after taking the partial trace of the monodromy matrices on both
sides. Similarly, one derives with the help of (\ref{ybe'}) and (\ref{ybe''}%
) that $\boldsymbol{G}^{\prime }(u)\boldsymbol{G}^{\prime }(v)=\boldsymbol{G}%
^{\prime }(v)\boldsymbol{G}^{\prime }(u)$ and $\boldsymbol{E}(u)\boldsymbol{G%
}^{\prime }(v)=\boldsymbol{G}^{\prime }(v)\boldsymbol{E}(u)$. Making a power
series expansion with respect to the variables $u,v$ and comparing
coefficients on both sides, the assertion follows.

The proof of the statement for $z=0$ is now an immediate consequence of the
definitions (\ref{ncE}) and (\ref{ncG'}). Alternatively, it follows from the
first relation (\ref{AT'}) in Corollary \ref{YBQ} which yields $%
A(u)T_{0,0}^{\prime }(v)=T_{0,0}^{\prime }(v)A(u)\boldsymbol{\ }$when
setting $m=m^{\prime }=0$. Moreover, one easily deduces from (\ref{ybe}) and
(\ref{ybe'}) that $A(u)A(v)=A(v)A(u)$ and $T_{0,0}^{\prime
}(u)T_{0,0}^{\prime }(v)=T_{0,0}^{\prime }(v)T_{0,0}^{\prime }(u)$.
\end{proof}

Returning briefly to the case of commuting variables, we can expect a
functional relationship between the generating functions%
\begin{equation*}
G^{\prime }(u)=\prod_{i>0}\frac{1}{(ux_{i};t)_{\infty }}\qquad \text{and}%
\qquad E(u)=\prod_{i>0}(1+ux_{i})\;.
\end{equation*}%
Namely, one verifies without difficulty that $G^{\prime }(u)E(-u)=G^{\prime
}(ut)$ has to hold by formally manipulating the infinite products in $%
G^{\prime }$. The following proposition states the noncommutative analogue
of this relation which due to the periodic boundary conditions, i.e. the use
of the affine instead of the finite quantum plactic algebra, contains an
additional term.

\begin{proposition}[functional equation]
\label{TQeqn}The following functional equation in $\mathcal{H}_{q}^{\otimes
n}\otimes \mathbb{C}[\![u]\!]$ is valid,%
\begin{equation}
\boldsymbol{E}(-u)\boldsymbol{G}^{\prime }(u)=\boldsymbol{G}^{\prime
}(uq^{2})+z(-u)^{n}\boldsymbol{G}^{\prime
}(uq^{-2})\prod_{i=1}^{n}q^{2N_{i}}\;.  \label{TQ}
\end{equation}
\end{proposition}

\begin{proof}
We start by considering the kernel $W\subset \mathbb{C}^{2}\otimes \mathcal{F%
}$ of $\tilde{R}(-1)$. The latter is spanned by the vectors $%
w_{0}=|0,0\rangle $ and $w_{m}=|0,m\rangle +|1,m-1\rangle $ for $m>0$. Here $%
|\sigma ,m\rangle =|\sigma \rangle \otimes |m\rangle $ and we have adopted
the bra-ket notation also for $\mathbb{C}^{2}$ with $\sigma =0,1$. From the
Yang-Baxter equation (\ref{ybe''}) we infer that
\begin{equation*}
L_{13}(-u)L_{23}^{\prime }(u)W\otimes \mathcal{H}_{q}\otimes \mathbb{C}%
[\![u]\!]\subset W\otimes \mathcal{H}_{q}\otimes \mathbb{C}[\![u]\!]\;.
\end{equation*}%
A straightforward computation yields that for any $X\in \mathcal{H}%
_{q}\otimes \mathbb{C}[\![u]\!]$%
\begin{multline*}
L_{13}(-u)L_{23}^{\prime }(u)w_{m}\otimes X= \\
u^{m}\sum_{m^{\prime }\geq 0}\left[ q^{2m}|0,m^{\prime }\rangle \otimes
\frac{\beta ^{\ast m}\beta ^{m^{\prime }}}{(q^{2})_{m}}X+|1,m^{\prime
}\rangle \otimes (\frac{\beta \beta ^{\ast m}\beta ^{m^{\prime }}}{%
(q^{2})_{m}}-\frac{\beta ^{\ast (m-1)}\beta ^{m^{\prime }}}{(q^{2})_{m-1}})X%
\right] = \\
=u^{m}q^{2m}\sum_{m^{\prime }\geq 0}w_{m^{\prime }}\otimes \frac{\beta
^{\ast m}\beta ^{m^{\prime }}}{(q^{2})_{m}}X,
\end{multline*}%
where we have used in the second line that $\beta \beta ^{\ast
m}=(1-q^{2m})\beta ^{\ast (m-1)}+\beta ^{\ast m}\beta $; see (\ref{crucial}%
). Thus, we can identify the action $L_{13}(-u)L_{23}^{\prime }(u)$ on $%
W\otimes \mathcal{H}_{q}\otimes \mathbb{C}[\![u]\!]$ with the action of $%
L^{\prime }(uq^{2})$ on $\mathcal{F}[\![u]\!]\otimes \mathcal{H}_{q}\mathbb{\ }$%
using the isomorphism $W\cong \mathcal{F}$ with $w_{m}\mapsto |m\rangle $.

Let us now turn to the quotient space $(\mathbb{C}^{2}\otimes \mathcal{F})/W$
which is spanned by the vectors $\bar{w}_{m}:=|1,m\rangle $ with $m\geq 0$.
We again calculate the action on the basis vectors for any $X\in \mathcal{H}%
_{q}\otimes \mathbb{C}[\![u]\!]$ and find%
\begin{multline*}
L_{13}(-u)L_{23}^{\prime }(u)\bar{w}_{m}\otimes X= \\
-\frac{u^{m+1}(1-q^{2m+2})}{(q^{2})_{m}}\sum_{m^{\prime }\geq 0}w_{m^{\prime
}}\otimes \beta ^{\ast (m+1)}\beta ^{m^{\prime }}X-\frac{u^{m+1}}{(q^{2})_{m}%
}\sum_{m^{\prime }\geq 0}\bar{w}_{m^{\prime }}\otimes \beta ^{\ast m}\beta
^{m^{\prime }}X \\
+\frac{u^{m+1}(1-q^{2m+2})}{(q^{2})_{m}}\sum_{m^{\prime }>0}\bar{w}%
_{m^{\prime }-1}\otimes \beta ^{\ast (m+1)}\beta ^{m^{\prime }}X= \\
-\frac{u^{m+1}}{(q^{2})_{m}}\sum_{m^{\prime }\geq 0}q^{-2m^{\prime }}\bar{w}%
_{m^{\prime }}\otimes \beta ^{\ast m}\beta ^{m^{\prime }}q^{2N}X+\ldots
\end{multline*}%
where in the last line we have omitted terms in $W\otimes \mathcal{H}_{q}$
since they are sent to zero under the quotient map. To arrive at this result
use the relation $\beta ^{\ast }\beta ^{m^{\prime }}=\beta ^{m^{\prime
}-1}(1-q^{2(N-m^{\prime }+1)})$ which follows from (\ref{crucial}). Using
both formulae for the action of $L_{13}(-u)L_{23}^{\prime }(u)$, one now
easily shows that the trace of the product
\begin{equation*}
T_{0}(-u)T_{0^{\prime }}^{\prime }(u)=L_{0n}(-u)L_{0^{\prime }n}^{\prime
}(u)\cdots L_{01}(-u)L_{0^{\prime }1}^{\prime }(u),
\end{equation*}%
with $0,0^{\prime }$ now labelling the factors in $\mathbb{C}^{2}\otimes
\mathcal{F}$ and $i=1,\ldots ,n$ the factors in $\mathcal{H}_{q}^{\otimes n}$%
, can be written as the sum of the traces of the monodromy matrices $%
T^{\prime }(uq^{2})$ and $(-u)^{n}T^{\prime }(uq^{-2})q^{2(N_{1}+\cdots
+N_{n})}$. Hence, the assertion follows.
\end{proof}

\subsection{Noncommutative Macdonald polynomials and Cauchy identities}

Employing known formulae for commutative symmetric functions we now define
polynomials in the generators of the quantum plactic algebra $\mathcal{\hat{U%
}}_{n}^{-}$ which can interpreted as noncommutative analogues of the
Macdonald functions $P_{\lambda }^{\prime },~Q_{\lambda }^{\prime }$; see (%
\ref{cHLdef}).\smallskip

\noindent Let $K(t)=M(s,P)$ be the unitriangular transition matrix between
Schur and Hall-Littlewood $P$-functions and denote by $M(P,s)=K(t)^{-1}$ its
inverse.

\begin{definition}
We define the following noncommutative analogues of Macdonald polynomials%
\begin{equation}
\boldsymbol{P}_{\lambda ^{\prime }}^{\prime }:=\sum_{\mu \geq \lambda }%
\boldsymbol{s}_{\mu ^{\prime }}K_{\mu \lambda }(t),\qquad \boldsymbol{s}%
_{\lambda }:=\det (\boldsymbol{e}_{\lambda _{i}^{\prime }-i+j})_{1\leq
i,j\leq \ell }  \label{ncschur}
\end{equation}%
and%
\begin{equation}
\boldsymbol{Q}_{\lambda }^{\prime }=\sum_{\mu ^{\prime }\leq \lambda
^{\prime }}K_{\lambda ^{\prime }\mu ^{\prime }}^{-1}(t)\boldsymbol{S}_{\mu
}^{\prime },\qquad \boldsymbol{S}_{\lambda }^{\prime }:=\det (\boldsymbol{g}%
_{\lambda _{i}-i+j}^{\prime })_{1\leq i,j\leq \ell }\;,  \label{ncQ'S'}
\end{equation}%
where $\boldsymbol{s}_{\lambda }$ is the noncommutative Schur polynomial and
$\boldsymbol{S}_{\lambda }^{\prime }$ its noncommutative dual; see (\ref{SS'}%
).
\end{definition}

The matrix elements $K_{\mu \lambda }(t)$ are the celebrated
Kostka-Foulkes polynomials for which explicit formulae are known; see e.g.
\cite[III.6 p242 and Examples 4, 7, pp 243-245 ]{Macdonald}\ and
references therein. The matrix elements of the inverse matrix, $K^{-1}(t)$,
and hence $\boldsymbol{Q}_{\lambda }^{\prime }$ can also be explicitly
computed: denote by $R_{ij}$ the familiar raising and lowering operators of
the ring of symmetric functions, $R_{ij}\lambda =(\lambda _{1},\ldots
,\lambda _{i}+1,\ldots ,\lambda _{j}-1,\ldots )$. Define $R_{ji}^{\prime
}F_{\lambda }:=F_{(R_{ji}\lambda ^{\prime })^{\prime }}$, then
\begin{equation*}
\boldsymbol{Q}_{\lambda }^{\prime }:=\prod_{\lambda _{i}^{\prime }>\lambda
_{j}^{\prime }}(1-tR_{ji}^{\prime })\boldsymbol{S}_{\lambda }^{\prime }
\end{equation*}%
compare with \cite[III.6]{Macdonald}.

\begin{remark}
Note that the polynomials $\boldsymbol{s}_{\lambda },\boldsymbol{S}_{\lambda
}^{\prime }$ and $\boldsymbol{Q}_{\lambda }^{\prime }$ are well-defined
because of (\ref{integrable}). In particular, we have for any $\lambda ,\mu $
that
\begin{equation}
\boldsymbol{S}_{\lambda }^{\prime }\boldsymbol{S}_{\mu }^{\prime }=%
\boldsymbol{S}_{\mu }^{\prime }\boldsymbol{S}_{\lambda }^{\prime },\quad
\boldsymbol{s}_{\lambda }\boldsymbol{s}_{\mu }=\boldsymbol{s}_{\mu }%
\boldsymbol{s}_{\lambda },\quad \boldsymbol{S}_{\lambda }^{\prime }%
\boldsymbol{s}_{\mu }=\boldsymbol{s}_{\mu }\boldsymbol{S}_{\lambda }^{\prime
}\;.
\end{equation}%
Note that in general some of the relations are different from the case of
the usual symmetric functions in commuting variables: the functional
relation (\ref{TQ}) implies that%
\begin{equation}
\sum_{a+b=c}(-1)^{a}\boldsymbol{e}_{a}\boldsymbol{g}_{b}^{\prime }=%
t^c\boldsymbol{g}_{c}^{\prime }+zt^{N_{\text{tot}}}\boldsymbol{g}_{c-n}^{\prime
}\label{TQ2}
\end{equation}%
with $N_{\text{tot}}=N_{1}+\cdots +N_{n}$. Only if we set $z=0$ the above
relation coincides with the known one for commuting variables, which is
easily derived with help of the generating functions (\ref{E}) and (\ref{G'}%
). Thus, in general we have that $\boldsymbol{P}_{\lambda }^{\prime }\neq
b_{\lambda ^{\prime }}(t)\boldsymbol{Q}_{\lambda }^{\prime }$ which is
different from the commutative case.

In contrast, the case $z=0$ corresponding to the finite quantum plactic
algebra has relations completely analogus to the commutative case, that is,
the finite quantum plactic polynomials $\boldsymbol{s}_{\lambda },%
\boldsymbol{S}_{\lambda }^{\prime }$ and $\boldsymbol{P}_{\lambda }^{\prime
},\boldsymbol{Q}_{\lambda }^{\prime }$ with $a_{n}$ set formally to zero
behave just as their commutative counterparts. In particular, it then
follows from the above definitions that for $\lambda $ being a horizontal or
vertical $r$-strip one has $\boldsymbol{Q}_{(r)}^{\prime }=\boldsymbol{g}%
_{r}^{\prime }\;$and$\;(1-t)\boldsymbol{Q}_{(1^{r})}^{\prime }=\boldsymbol{e}%
_{r}$.
\end{remark}

The definition of the above noncommutative polynomials is motivated by the
following noncommutative analogues of Cauchy identities; compare with (\ref%
{C1}) and (\ref{C2}).

\begin{corollary}[noncommutative Cauchy identities]
We have the expansions%
\begin{equation}
\boldsymbol{E}(x_{1})\cdots \boldsymbol{E}(x_{k})=
\sum_{\lambda}m_{\lambda }(x)\boldsymbol{e}_{\lambda }=
\sum_{\lambda}s_{\lambda }(x)\boldsymbol{s}_{\lambda ^{\prime
}}=\sum_{\lambda}P_{\lambda }(x;t)\boldsymbol{P}%
_{\lambda ^{\prime }}^{\prime }  \label{ncC1}
\end{equation}%
where $\boldsymbol{e}_{\lambda }:=\boldsymbol{e}_{\lambda _{1}}\boldsymbol{e}%
_{\lambda _{2}}\cdots$ and%
\begin{equation}
\boldsymbol{G}^{\prime }(x_{1})\cdots \boldsymbol{G}^{\prime
}(x_{n-1})=\sum_{\lambda}m_{\lambda }(x)\boldsymbol{g}_{\lambda
}^{\prime }=\sum_{\lambda}s_{\lambda }(x)\boldsymbol{S}_{\lambda
}^{\prime }=\sum_{\lambda}P_{\lambda }^{\prime }(x;t)\boldsymbol{Q}%
_{\lambda }^{\prime },  \label{ncC2}
\end{equation}%
where $\boldsymbol{g}_{\lambda }^{\prime }:=\boldsymbol{g}%
_{\lambda _{1}}^{\prime }\boldsymbol{g}_{\lambda _{2}}^{\prime }\cdots$.
\end{corollary}

\begin{proof}
The first equalities in (\ref{ncC1}) and (\ref{ncC2}) are obvious and follow
from the definition of $\boldsymbol{e}_{r}$ and $\boldsymbol{g}_{r}^{\prime
} $ as the expansion coefficients of the matrices (\ref{ncE}) and (\ref{ncG'}%
), respectively. To prove the other identities note that the definitions of $%
\boldsymbol{s}_{\lambda ^{\prime }},\boldsymbol{P}_{\lambda ^{\prime
}}^{\prime }$ and $\boldsymbol{S}_{\lambda }^{\prime }$, $\boldsymbol{Q}%
_{\lambda }^{\prime }$ are the same as in the case of commuting variables.
Thus, the expansions follow from the same arguments as in the case of
commuting variables using the known transition matrices $M(s,m)=K^{-1}(1)$
\cite[III.6, Table on page 241]{Macdonald} and $M(P,s)=K(t)^{-1}$.
\end{proof}

\section{Hall-Littlewood functions and the Yang-Baxter algebra}

We now describe the action of the Yang-Baxter algebra employing the language
of statistical mechanics: we show that the generators $A,B,C,D$ can be used
to compute partition functions of certain statistical vertex models defined
on a finite square lattice. These partition functions turn out to be skew
Hall-Littlewood functions. More precisely, the Yang-Baxter algebra
generators play the role of row-to-row transfer matrices, i.e. their matrix
elements yield the partition functions of a single lattice row with
appropriate boundary conditions imposed on the lattice. We state for each
generator a bijection between lattice configurations of the respective
statistical model and (semi-standard) skew tableaux, showing that the
Boltzmann weights of the statistical model coincide with the coefficient
functions (\ref{phi}) and (\ref{psi}) appearing in the definition of skew
Hall-Littlewood functions. This section can be seen as a preparatory step to
motivate our construction of cylindric Hall-Littlewood functions in the
subsequent section.

\begin{figure}[tbp]
\begin{equation*}
\includegraphics[scale=0.35]{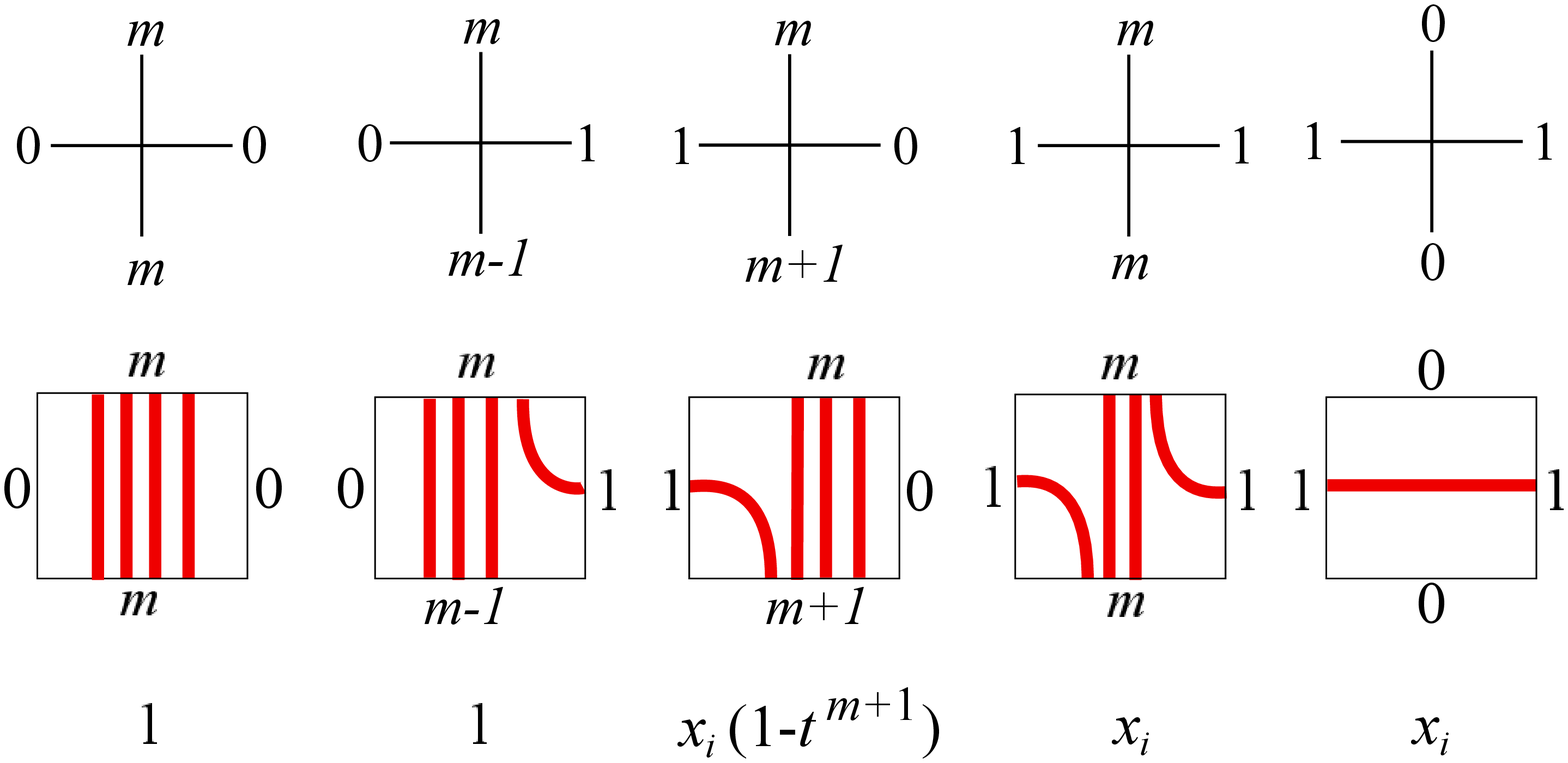}
\end{equation*}%
\caption{Depiction of the allowed vertex configurations associated with
the $L$-operator (\protect\ref{L}); see top row. Horizontal edges can have values $%
0 $ and $1$ only, while vertical edges can have values $m\in \mathbb{Z}%
_{\geq 0}$. Note the constraint that the sum of the values at the N and E
edge equals the sum of the W and S values. Below, in the middle row, the vertex
configurations are described in terms of nonintersecting paths. Listed at the bottom are
the associated weights where $i$ is the row index of the lattice.}
\label{fig:qbosonvertex}
\end{figure}

\subsection{The statistical vertex model associated with $L$}

Set $\mathbb{I}_{r}=\{0,1,2,\ldots ,r+1\}\subset \mathbb{Z}$ and consider
the set $\mathbb{L}:=\mathbb{I}_{\ell }\times \mathbb{I}_{n}$. We call $%
\mathbb{\dot{L}}:=\{\langle i,j\rangle \in \mathbb{L}~|~i\neq 0,\ell +1$ and
$j\neq 0,n+1\}$ the set of \emph{interior lattice points} and $\partial
\mathbb{L}:=\mathbb{L}\backslash \mathbb{\dot{L}}$ the set of \emph{boundary
lattice points,} identifying $\mathbb{\dot{L}}$\ with a square lattice of $%
\ell $ lattice rows and $n$ lattice columns; see the definition below. The $%
i $th lattice row with $i=1,\ldots ,\ell $ is the set $\{\langle r,s\rangle
\in \mathbb{L}:r=i\}$ and the $j$th lattice column with $j=1,\ldots ,n$ is
the set $\{\langle r,s\rangle \in \mathbb{L}:s=j\}$. We will keep the number
of lattice columns fixed throughout our discussion, while $\ell $ can vary.

\begin{definition}[lattice edges]
Let $u,v\in \mathbb{L}$. Then we call the pair

\begin{itemize}
\item $(u,v)$ a \emph{horizontal edge} if $u_{1}=v_{1}$ and $u_{2}+1=v_{2}$.
We refer to $u$ and $v$ as the start and end point of the edge,
respectively. Denote by $\mathbb{E}_{h}\subset \mathbb{L}\times \mathbb{L}$
the set of all horizontal edges which start or end at an interior lattice
point.

\item $(u,v)$ a \emph{vertical edge} if $u_{1}+1=v_{1}$ and $u_{2}=v_{2}$.
Similar as before we call $u$ and $v$ the start and end point and denote by $%
\mathbb{E}_{v}\subset \mathbb{L}\times \mathbb{L}$ the set of all vertical
edges which start or end at an interior lattice point.
\end{itemize}

We call the horizontal and vertical edges which start (end) at a boundary
lattice point and end (start) at an interior point \emph{outer} horizontal
and vertical edges, respectively.
\end{definition}

We will often refer to the horizontal edges in the $i$th lattice row ($%
i=1,\ldots ,\ell $) or the vertical edges in the $j$th lattice column ($%
j=1,\ldots ,n$). By this we shall mean horizontal edges which start or end
at a point $\langle i,r\rangle $ and vertical edges which start or end at a
point $\langle s,j\rangle $, respectively. By the horizontal edges in the $j$%
th lattice column we shall mean those edges which end in the $j$th lattice
column. Similarly the upper and lower vertical edges in the $i$th lattice
row will be those vertical edges which respectively end and start in the $i$%
th row.

We now assign so-called statistical variables to the lattice edges.

\begin{definition}[lattice \& vertex configurations]
A \emph{horizontal }and \emph{vertical edge configuration} are maps $\gamma
_{h}:\mathbb{E}_{h}\rightarrow \{0,1\}$ and $\gamma _{v}:\mathbb{E}%
_{v}\rightarrow \mathbb{Z}_{\geq 0}$, respectively. A pair $\gamma =(\gamma
_{h},\gamma _{v})$ is simply called a \emph{lattice configuration}. Given $%
\gamma $ and a lattice point $\langle i,j\rangle \in \mathbb{\dot{L}}$, we
call the images of the horizontal and vertical edges which either start or
end at this lattice point \emph{the vertex configuration }at $\langle
i,j\rangle $ and denote it by $\gamma _{\langle i,j\rangle }$.
\end{definition}

Each lattice configuration $\gamma $ can be assigned a statistical weight as
follows: given a vertex configuration $\gamma _{\langle i,j\rangle
}=\{\sigma ,m,\sigma ^{\prime },m^{\prime }\}$, where $\sigma ,\sigma
^{\prime }=0,1$ are the images of the W and E horizontal edges and $%
m,m^{\prime }\in \mathbb{Z}_{\geq 0}$ the images of the N and S vertical
edges, define%
\begin{equation}
\limfunc{wt}(\gamma ):=\prod_{\langle i,j\rangle \in \mathbb{\dot{L}}}%
\limfunc{wt}(\gamma _{\langle i,j\rangle }),\qquad \limfunc{wt}(\gamma
_{\langle i,j\rangle }):=\langle \sigma ,m|L(x_{i})|\sigma ^{\prime
},m^{\prime }\rangle \;.  \label{Lweight}
\end{equation}%
Here $\langle \sigma ,m|L(x_{i})|\sigma ^{\prime },m^{\prime }\rangle $ are
the matrix elements of the solution (\ref{L}) to the Yang-Baxter equation
with $\langle \sigma ,m|\sigma ^{\prime },m^{\prime }\rangle :=\delta
_{\sigma ,\sigma ^{\prime }}\langle m|m^{\prime }\rangle $ and $%
x_{i},\,i=1,\ldots ,\ell $ are some abstract commutative variables which
only depend on the row index. The nonzero weights for the allowed vertex
configurations are listed in Figure \ref{fig:qbosonvertex}. Note that only
the weights of vertex configurations for interior points contribute to the
weight of a lattice configuration.

\begin{remark}
Alternatively, each lattice configuration can be described in terms of
non-intersecting paths, closely related to the infinitely-friendly walker
model in \cite{Guttmann2,Guttmann3}, by mapping vertex configurations onto
those path configurations shown in Figure \ref{fig:qbosonvertex}. However,
we will not employ this path picture in the proofs but instead continue to
focus on an algebraic description in terms of vertex models.
\end{remark}

In what follows we will impose \emph{boundary conditions} on the lattice,
that is we will prescribe certain values for the outer horizontal and
vertical edges. Namely, given $\mu \in \mathcal{A}_{k,n}^{+}$, $\lambda \in
\mathcal{A}_{k^{\prime },n}^{+}$ and $\sigma ,\tau =0,1$ consider the set $%
\Gamma _{\lambda ,\mu }^{\sigma ,\tau }$ of lattice configurations $\gamma $
where the outer vertical edge starting at $\langle 0,j\rangle $ has value $%
m_{j}(\mu )$ and the outer vertical edge ending at $\langle \ell +1,j\rangle
$ value $m_{j}(\lambda )$ for all $1\leq j\leq n$. Fix the values of the
outer horizontal edges either starting at $\langle i,0\rangle $ or ending at
$\langle i,n+1\rangle $ to be $\sigma $ or $\tau $ respectively for all $%
1\leq i\leq \ell $.

\begin{definition}[partition function]
The \emph{partition function }is the weighted sum $Z_{\lambda ,\mu }^{\sigma
,\tau }(x_{1},\ldots ,x_{\ell })=\sum_{\gamma \in \Gamma _{\lambda ,\mu
}^{\sigma ,\tau }}\limfunc{wt}(\gamma )$ over those lattice configurations $%
\gamma $ which satisfy the boundary conditions specified by $\lambda ,\mu
,\sigma ,\tau $ as just described.
\end{definition}

\begin{lemma}
We have the following identities between partition functions and matrix
elements,%
\begin{eqnarray}
Z_{\lambda ,\mu }^{0,0}(x_{1},\ldots ,x_{\ell }) &=&\langle \lambda
|A(x_{1})\cdots A(x_{\ell })|\mu \rangle ,\quad \lambda ,\mu \in \mathcal{A}%
_{k,n}^{+}  \label{ZA} \\
Z_{\lambda ,\mu }^{1,0}(x_{1},\ldots ,x_{\ell }) &=&\langle \lambda
|B(x_{1})\cdots B(x_{\ell })|\mu \rangle ,\quad \lambda \in \mathcal{A}%
_{k+\ell ,n}^{+},\;\mu \in \mathcal{A}_{k,n}^{+}  \label{ZB} \\
Z_{\lambda ,\mu }^{0,1}(x_{1},\ldots ,x_{\ell }) &=&\langle \lambda
|C(x_{1})\cdots C(x_{\ell })|\mu \rangle ,\quad \lambda \in \mathcal{A}%
_{k-\ell ,n}^{+},\;\mu \in \mathcal{A}_{k,n}^{+}  \label{ZC} \\
Z_{\lambda ,\mu }^{1,1}(x_{1},\ldots ,x_{\ell }) &=&\langle \lambda
|D(x_{1})\cdots D(x_{\ell })|\mu \rangle ,\quad \lambda ,\mu \in \mathcal{A}%
_{k,n}^{+}  \label{ZD}
\end{eqnarray}
\end{lemma}

\begin{proof}
This is a direct consequence of the definition (\ref{Lweight}) and the
monodromy matrix (\ref{T}).
\end{proof}

\begin{example}
Figure \ref{fig:Bexample} shows an example for an allowed configuration for
the $B$-operator when $n=6,\,k=4,\,\ell =3$ and $\mu =(5,5,4,3),\,\lambda
=(6,6,5,4,3,2,1)$. Because in each row one path is entering on the outer
horizontal edge from the left and none is exiting on the right, the level $%
k=\sum_{i=1}^{n}m_{i}$ increases by one in each row and, thus, the partition
function is only nonzero for $\lambda \in \mathcal{A}_{k+\ell ,n}^{+}$.
\end{example}

\begin{figure}[tbp]
\begin{equation*}
\includegraphics[scale=0.33]{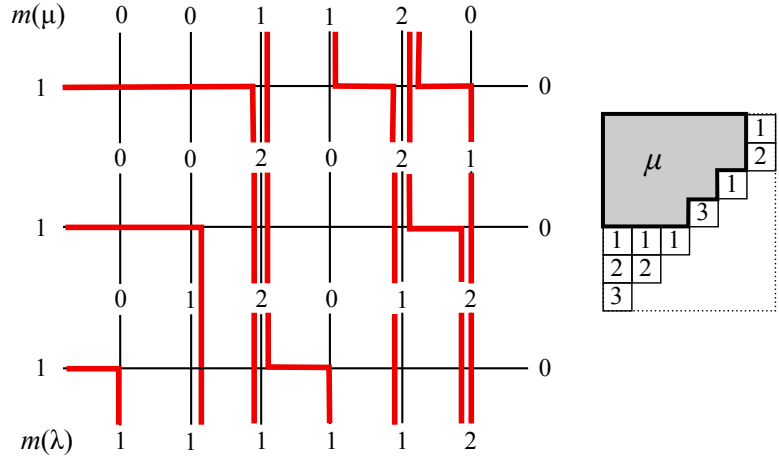}
\end{equation*}%
\caption{Graphical depiction of a sum over vertex configurations for the $B$%
-operator of the Yang-Baxter algebra (\protect\ref{ybe}) with $%
n=6,\,k=4,\,\ell =3$ and $\protect\mu=(5,5,4,3),\,\protect\lambda %
=(6,6,5,4,3,2,1)$. Depicted on the right is the corresponding skew tableau
under the bijection described in the proof of Lemma \protect\ref{bijection1}%
. }
\label{fig:Bexample}
\end{figure}

\paragraph{Lattice-tableau bijection 1: the $A$ and $B$-operator}

\begin{lemma}
\label{bijection1}Let $\lambda \in \mathcal{A}_{k+\ell ,n}^{+}$ and $\mu \in
\mathcal{A}_{k,n}^{+}$. There exists a bijection $\gamma \mapsto T(\gamma )$
between lattice configurations $\gamma \in \Gamma _{\lambda ,\mu }^{1,0}$
and skew tableaux $T$ of shape $\lambda /\mu $ such that $\limfunc{wt}%
(\gamma )=\varphi _{T(\gamma )}x^{T(\gamma )}$, where $\varphi _{T}$ is
defined in (\ref{phi}). The analogous statement holds for $\lambda ,\mu \in
\mathcal{A}_{k,n}^{+}$ and $\gamma \in \Gamma _{\lambda ,\mu }^{0,0}$.
\end{lemma}

\begin{proof}
We will make use of the fact that a (semi-standard) tableau $T$ of shape $%
\lambda /\mu $ is equivalent to a sequence of partitions $(\mu =\lambda
^{(0)},\lambda ^{(1)},\ldots ,\lambda ^{(\ell )}=\lambda )$ such that $%
\lambda ^{(i+1)}/\lambda ^{(i)}$ is a horizontal strip; see e.g. \cite[%
Chapter I, Section 1]{Macdonald}. It will therefore suffice to prove the
bijection for a single horizontal strip and a single row configuration.

\begin{figure}[tbp]
\begin{equation*}
\includegraphics[scale=0.37]{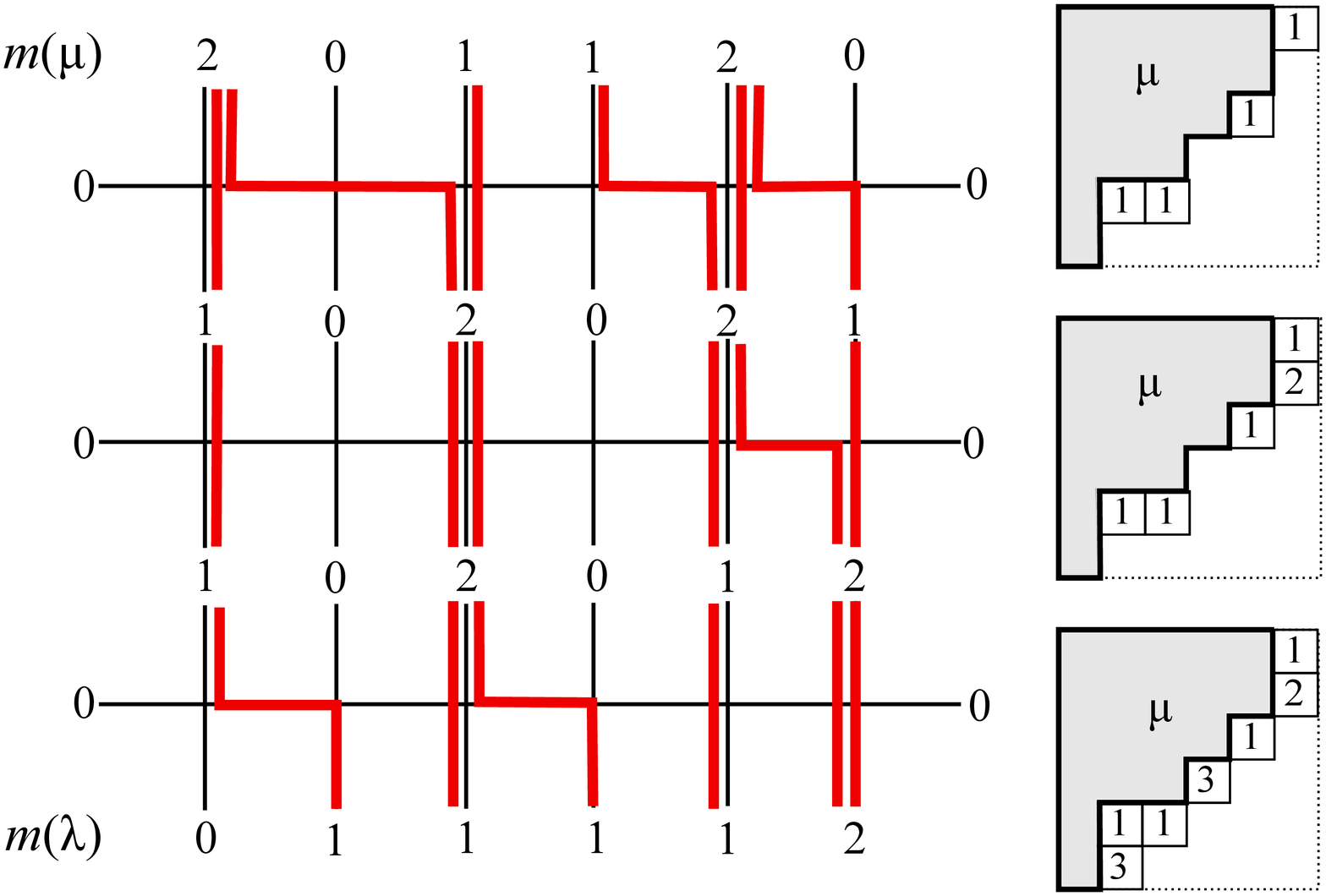}
\end{equation*}%
\caption{Graphical depiction of a sum over vertex configurations for the $A$%
-operator of the Yang-Baxter algebra (\protect\ref{ybe}) with $n=k=6,\,\ell
=3$ and $\protect\mu=(5,5,4,3,1,1),\,\protect\lambda=(6,6,5,4,3,2)$.
Depicted on the right are the corresponding skew diagrams for each lattice
row; see the proof of Lemma \protect\ref{bijection1}.}
\label{fig:Aexample}
\end{figure}

Assume $\ell =1$, that is consider one lattice row only. Fix an allowed
horizontal edge configuration $\sigma =(\sigma _{1},\sigma _{2},\ldots
,\sigma _{n+1})\in \{0,1\}^{n+1}$ in $\Gamma _{\lambda ,\mu }^{1,0}$.
According to Figure \ref{fig:qbosonvertex} $m_{i}(\lambda )-m_{i}(\mu
)=\sigma _{i}-\sigma _{i+1}$ for allowed configurations, where $%
m_{i}(\lambda )=\lambda _{i}^{\prime }-\lambda _{i+1}^{\prime }$ and $%
m_{i}(\mu )=\mu _{i}^{\prime }-\mu _{i+1}^{\prime }$. Moreover, we must have
that $0\leq \sigma _{i+1}\leq \min (1,m_{i}(\mu ))$ with $\sigma _{1}=1$ and
$\sigma _{n+1}=0$. Hence,
\begin{equation}
\sigma _{i+1}=\sigma _{i}+m_{i}(\mu )-m_{i}(\lambda )=\sigma
_{1}+\sum_{j=1}^{i}(m_{j}(\mu )-m_{j}(\lambda ))  \label{rowconfig}
\end{equation}%
and the horizontal edge value $\sigma _{i}$ in the $i$th lattice column is
given by $\sigma _{i}=\lambda _{i}^{\prime }-\mu _{i}^{\prime }=0,1$ because
$\mu _{1}^{\prime }=\sum_{j=1}^{n}m_{j}(\mu )=k$ while $\lambda _{1}^{\prime
}=\sum_{j=1}^{n}m_{j}(\lambda )=k+1$. Thus, $\lambda /\mu $ is a horizontal
strip with a box in the $i$th column of the skew diagram if $\sigma _{i}=1$.

Conversely, it is easy to see that each horizontal strip $\lambda /\mu $
with $\lambda _{1},\mu _{1}\leq n$ which has a box in the first column
defines a unique allowed row configuration of the lattice employing the same
formulae in reverse order.

Recall from the definition (\ref{phi}), that $\varphi _{\lambda /\mu }$
contains a factor $(1-t^{m_{i}(\lambda )})$ if there is a box added in the $%
i $th column of the skew diagram but none in the ($i+1$)th. In terms of the
correspondence between row configurations and horizontal strips this means
that the value of the horizontal\emph{\ }edge in $i$th lattice column is one
and zero in the ($i+1$)th. Thus, we obtain the weight of the third vertex
configuration shown in Figure \ref{fig:qbosonvertex}. Finally, it is obvious
that the sum $r=\sum_{i=1}^{n+1}\sigma _{i}$ over the horizontal edge values
gives the length of the horizontal strip. From this we now easily deduce
that $\limfunc{wt}(\gamma (T))=\varphi _{\lambda /\mu }x^{r}$ as desired.

Assume now $\ell \geq 1$ then it follows that a lattice configuration
defines a sequence of horizontal strips of the type just described and hence
we end up with skew tableaux of $\lambda /\mu $ with $\ell $ boxes in the
first column: starting from the first lattice row on the top add a box
labelled with one to each column of the Young diagram of $\mu $ whenever a
horizontal path edge occurs in the lattice column with the same number. Then
continue and do the same for the second lattice row labeling boxes now with
2, and so on. Conversely, given a skew tableau $T$ of shape $\lambda /\mu $
with $\lambda \in \mathcal{A}_{k+\ell ,n}^{+}$ and $\mu \in \mathcal{A}%
_{k,n}^{+}$, the first column of the skew tableau $T$ of shape $\lambda /\mu
$ must be of height $\ell $ with boxes labelled from 1 to $\ell $. It thus
decomposes into a sequence of allowed row configurations. The equality
between the weight $\limfunc{wt}(\gamma )$ of the lattice configuration $%
\gamma =\gamma (T)$ corresponding to the tableau $T$ and the value of the
coefficient function $\varphi _{T}$ is now a direct consequence of the fact
that $\varphi _{T}=\prod_{i\geq 0}\varphi _{\lambda ^{(i+1)}/\lambda ^{(i)}}$
with $\lambda ^{(i+1)}/\lambda ^{(i)}$ being the horizontal strips
determined by $T$ and that the weight $\limfunc{wt}(\gamma )$ of the lattice
configuration $\gamma $ is the product of the weights of its row
configurations, see (\ref{Lweight}).

The construction of the bijection for $\lambda ,\mu \in \mathcal{A}%
_{k,n}^{+} $ and $\gamma \in \Gamma _{\lambda ,\mu }^{0,0}$ is completely
analogous and only differs in the boundary conditions imposed on the square
lattice: now the left \emph{and} the right boundary edges are set to zero,
so one has $\sigma _{1}=\sigma _{n+1}=0$ in each row configuration. An
example is shown in Figure \ref{fig:Aexample}. From the graphical depiction
of the weights of vertex configurations in Figure \ref{fig:qbosonvertex} it
is evident that these boundary conditions imply that the level $%
k=\sum_{i=1}^{n}m_{i}(\lambda )$ is preserved in each lattice row: the sum
over the values of the outer vertical edges on the top must equal the sum of
the values of the outer vertical edges at the bottom. This means that we
have to choose $\lambda $ and $\mu \in \mathcal{A}_{k,n}^{+}$ to obtain
allowed lattice configurations. Conversely, only skew diagrams $\lambda /\mu
$ which have no boxes in the first column correspond to such lattice
configurations. We omit the remainder of the proof as it now closely follows
along the previous lines.
\end{proof}

\paragraph{Lattice-tableau bijection 2: the $C$ and $D$-operator}

\begin{figure}[tbp]
\begin{equation*}
\includegraphics[scale=0.35]{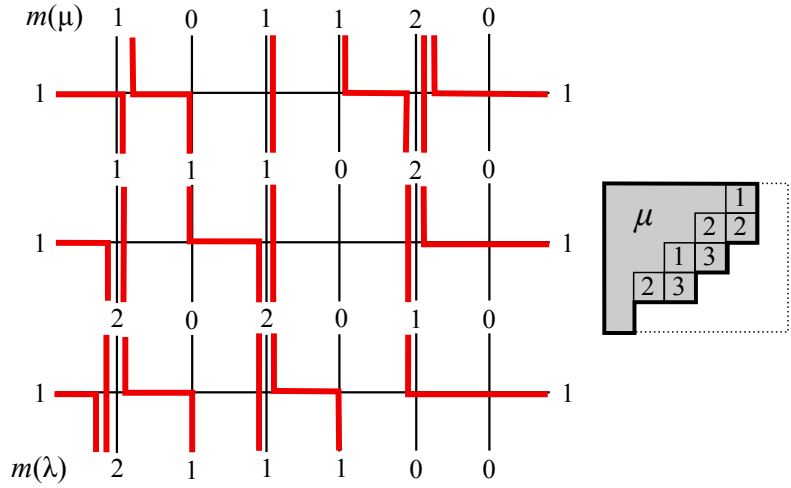}
\end{equation*}%
\caption{Graphical depiction of a sum over vertex configurations for the $D$%
-operator of the Yang-Baxter algebra (\protect\ref{ybe}) with $%
n=6,\,k=5,\,\ell =3$ and $\protect\mu=(5,5,4,3,1),\,\protect\lambda%
=(4,3,2,1,1)$. Depicted on the right is the corresponding skew tableau; see
the proof of Lemma \protect\ref{bijection2}.}
\label{fig:Dexample}
\end{figure}

\begin{lemma}
\label{bijection2}Let $\lambda \in \mathcal{A}_{k-\ell ,n}^{+}$ and $\mu \in
\mathcal{A}_{k,n}^{+}$ with $\ell \leq k$. There exists a bijection $\gamma
\mapsto T(\gamma )$ between lattice configurations $\gamma \in \Gamma
_{\lambda ,\mu }^{0,1}$ and Young tableaux $T$ of shape $\mu /\lambda $ such
that $\limfunc{wt}(\gamma )=\psi _{T(\gamma )}x^{T(\gamma )}$. A similar
statement holds for $\gamma \in \Gamma _{\lambda ,\mu }^{1,1}$ and $\lambda
,\mu \in \mathcal{A}_{k,n}^{+}$.
\end{lemma}

\begin{proof}
We will discuss the $D$-operator case only, the generalization to the $C$%
-operator will then be obvious. The outer horizontal edges on the left and
right boundary now all have value one, $\sigma =\tau =1$; compare with the
definition of the $D$-operator via (\ref{T}). By the same reasoning as
before it suffices to consider the simplest case $\ell =1.$

Employing the relations $m_{i}(\lambda )-m_{i}(\mu )=\sigma _{i}-\sigma
_{i+1}$, $0\leq \sigma _{i+1}\leq \min (1,m_{i}(\mu ))$ and $\sigma
_{1}=\sigma _{n+1}=1$ for allowed row configurations according to Figure \ref%
{fig:qbosonvertex}, we now infer with the help of (\ref{rowconfig}) that $%
1-\sigma _{i}=\mu _{i}^{\prime }-\lambda _{i}^{\prime }=0,1$ since $\lambda
,\mu \in \mathcal{A}_{k,n}^{+}$. Thus, we obtain a horizontal strip $\mu
/\lambda $ which has a box in the $i$th column of the skew diagram if $%
\sigma _{i}=0$. Conversely, given a horizontal strip of shape $\mu /\lambda $
with $\lambda ,\mu \in \mathcal{A}_{k,n}^{+}$ we must have $\sigma
_{1}=\sigma _{n+1}=1$ since there can be no boxes in the first or $(n+1)$th
column.

The equality between weights $\limfunc{wt}(\gamma )$ of row configurations $%
\gamma $ and the coefficient function $\psi _{\mu /\lambda }$ defined in (%
\ref{psi}) follows again from their definitions: $\psi _{\mu /\lambda }$
contains a factor $(1-t^{m_{i}(\lambda )})$ if the skew diagram $\mu
/\lambda $ has a box in the ($i+1)$th column but not in the preceding one.
Comparison with the vertex configurations in Figure \ref{fig:qbosonvertex}
shows that this matches again with the third vertex configuration and its
weight. Furthermore, we obviously have that $r=n+1-\sum_{i=1}^{n+1}\sigma
_{i},$ where $\sigma _{i}$ are the horizontal edge values, equals the length
of the horizontal strip and, thus, $\limfunc{wt}(\gamma )=x_{1}^{n-r}\psi
_{\mu /\lambda }$.

The general case of $\ell \geq 1$ is obtained again by writing each tableau $%
T$ as a sequence of horizontal strips: label the lowest box \emph{within}
each column of the Young diagram of $\mu $ with $\ell $ whenever there is no
horizontal edge in the corresponding lattice column. (We employ the same
numbering convention of columns in the Young diagram and the lattice as
before.) Do the same for the second lattice row labeling boxes now with $%
\ell -1$. Continue with this procedure up to the last row. An example is
provided in Figure \ref{fig:Dexample}. The equality of lattice configuration
weights with $(x_{1}\cdots x_{\ell })^{n}\psi _{T}x^{-T}$ is once more
immediate from the definition of $\psi _{T}$ and the fact that lattice
configurations are products of row configurations.

Finally, matrix elements for the $C$-operator are obtained by setting $%
\sigma _{1}=0$ and $\sigma _{n+1}=1$ in each row. Allowed lattice
configurations can therefore only occur if $\lambda \in \mathcal{A}_{k-\ell
,n}^{+}$ and $\mu \in \mathcal{A}_{k,n}^{+}$ and the corresponding skew
tableaux $\mu /\lambda $ will have $\ell $ boxes in the first column. The
remainder of the proof now follows along the same lines as before.
\end{proof}

\subsection{Skew Hall-Littlewood functions as partition functions}

The following proposition and corollary summarise the findings of the two
previous lemmata with $\varphi _{\lambda /\mu }$ and $\psi _{\mu /\lambda }$
defined in (\ref{phi}) and (\ref{psi}), respectively.

\begin{proposition}[Pieri-type formulae]
\label{YBPieri}Let $\mu \in \mathcal{A}_{k,n}^{+}$ and set $t=q^{2}$. Then
the action of the Yang-Baxter operators can be expressed as%
\begin{equation}
A_{r}|\mu \rangle =\sum_{\substack{ \lambda -\mu =(r),  \\ \lambda \in
\mathcal{A}_{k,n}^{+}}}\varphi _{\lambda /\mu }(t)|\lambda \rangle ,\text{%
\qquad }B_{r}|\mu \rangle =\sum_{\substack{ \lambda -\mu =(r),  \\ \lambda
\in \mathcal{A}_{k+1,n}^{+}}}\varphi _{\lambda /\mu }(t)|\lambda \rangle
\label{ABPieri}
\end{equation}%
and%
\begin{equation}
C_{r}|\mu \rangle =\sum_{\substack{ \mu -\lambda =(n-r),  \\ \lambda \in
\mathcal{A}_{k-1,n}^{+}}}\psi _{\mu /\lambda }(t)|\lambda \rangle ,\text{%
\qquad }D_{r}|\mu \rangle =\sum_{\substack{ \mu -\lambda =(n-r),  \\ \lambda
\in \mathcal{A}_{k,n}^{+}}}\psi _{\mu /\lambda }(t)|\lambda \rangle \ ,
\label{CDPieri}
\end{equation}%
where the notation $\lambda -\mu =(r)$ and $\mu -\lambda =(n-r)$ means that
the skew diagrams $\lambda /\mu $ and $\mu /\lambda $ are horizontal strips
of length $r$ and $n-r$, respectively. In particular we have $A_{r}^{\ast
}=D_{n-r}$ and $B_{r}^{\ast }=C_{n-r}$ with respect to the inner product on
the Fock space $\mathcal{F}^{\otimes n}$.
\end{proposition}
The multiple action of the Yang-Baxter algebra generators $%
A,B,C,D$ can be described in terms of skew Hall-Littlewood functions.
\begin{corollary}
\label{YBHL} Let $%
\mu \in \mathcal{A}_{k,n}^{+}$ and $x_{1},...,x_{\ell }$ be some generic
variables, then we have%
\begin{eqnarray}
A(x_{1})\cdots A(x_{\ell })|\mu \rangle &=&\sum_{\lambda \in \mathcal{A}%
_{k,n}^{+}}Q_{\lambda /\mu }(x_{1},...,x_{\ell };t)|\lambda \rangle ,
\label{AHL} \\
B(x_{1})\cdots B(x_{\ell })|\mu \rangle &=&\sum_{\lambda \in \mathcal{A}%
_{k+\ell ,n}^{+}}Q_{\lambda /\mu }(x_{1},...,x_{\ell };t)|\lambda \rangle ,
\label{BHL} \\
C(x_{1})\cdots C(x_{\ell })|\mu \rangle &=&(x_{1}\cdots x_{\ell
})^{n}\sum_{\lambda \in \mathcal{A}_{k-\ell ,n}^{+}}P_{\mu /\lambda
}(x_{1}^{-1},...,x_{\ell }^{-1};t)|\lambda \rangle ,  \label{CHL} \\
D(x_{1})\cdots D(x_{\ell })|\mu \rangle &=&(x_{1}\cdots x_{\ell
})^{n}\sum_{\lambda \in \mathcal{A}_{k,n}^{+}}P_{\mu /\lambda
}(x_{1}^{-1},...,x_{\ell }^{-1};t)|\lambda \rangle \ .  \label{DHL}
\end{eqnarray}
\end{corollary}

\begin{remark}
Note in particular that for $|\emptyset \rangle :=|(0,0,\ldots )\rangle $
(the pseudo-vacuum) we have the following expansion%
\begin{equation}
B(x_{1})\cdots B(x_{k})|\emptyset \rangle =\sum_{\lambda \in \mathcal{A}%
_{k,n}^{+}}Q_{\lambda }(x_{1},...,x_{k};t)|\lambda \rangle \ .
\label{BetheHL}
\end{equation}%
For this special case the connection with Hall-Littlewood functions has been
first dicussed in \cite[Section 3, Propositions 3 and 4]{Tsilevich}.
However, the formula (\ref{BetheHL}) differs from the one in \cite[Section
3, Propositions 3 and 4]{Tsilevich}: in our result the sums are restricted
to partitions of a fixed length $\ell (\lambda )=k\geq 0$. Without this
restriction one obtains erroneous results.
\end{remark}

\section{Cylindric Hall-Littlewood functions}

So far we have considered lattice configurations where the horizontal edge
values on the left and on the right boundary are fixed to certain values in
each row. Now we wish to consider (quasi) periodic boundary conditions, i.e.
we move onto the cylinder by identifying the first and the last lattice
column.

\subsection{Periodic boundary conditions}

Given $\lambda ,\mu \in \mathcal{A}_{k,n}^{+}$ we now restrict $1\leq $ $%
\ell \leq k$ and consider the set $\Gamma _{\lambda ,\mu }$ of lattice
configurations $\gamma $ where the outer vertical edge starting at $\langle
0,j\rangle $ has value $m_{j}(\mu )$ and the outer vertical edge ending at $%
\langle \ell +1,j\rangle $ value $m_{j}(\lambda )$ for all $1\leq j\leq n$
as before. Moreover, $\gamma _{h}$ should satisfy $\gamma _{h}(\langle
i,0\rangle ,\langle i,1\rangle )=\gamma _{h}(\langle i,n\rangle ,\langle
i,n+1\rangle )$ for all $1\leq i\leq \ell $. That is, the values of the
outer horizontal edges either starting at $\langle i,0\rangle $ or ending at
$\langle i,n+1\rangle $ need to be the same. Clearly this condition is
equivalent to considering the cylindric lattice $\mathbb{L}^{\text{cyl}}:=%
\mathbb{I}_{\ell }\times \mathbb{Z}_{n}$ with the obvious generalisations of
the definitions of horizontal, vertical edges and lattice configurations.

Denote by $\Gamma _{\lambda /d/\mu }:=\{\gamma \in \Gamma _{\lambda }^{\mu
}:\sum_{i=1}^{\ell }\gamma _{h}(\langle i,0\rangle ,\langle i,1\rangle )=d\}$
the subset of configurations where the sum over the values of the left (or
right) outer horizontal edges is $d$ and let
\begin{equation}
Z_{\lambda /d/\mu }(x_{1},\ldots ,x_{\ell }):=\sum_{\gamma \in \Gamma
_{\lambda /d/\mu }}\limfunc{wt}(\gamma )
\end{equation}%
be the partition function for the lattice with periodic boundary conditions
in the horizontal direction. We introduce a formal variable $z$ keeping
track of the winding number around the cylinder; the latter is the same
variable as the one used in (\ref{qplacticrep}) and (\ref{ncE}).

\begin{lemma}
We have the identity%
\begin{equation}
\langle \lambda |\boldsymbol{E}(x_{1})\cdots \boldsymbol{E}(x_{\ell })|\mu
\rangle =\sum_{d\geq 0}z^{d}Z_{\lambda /d/\mu }(x_{1},\ldots ,x_{\ell
}),\qquad \ell \leq k\;,  \label{ZE}
\end{equation}%
and, hence, in light of definition (\ref{ncE}) that%
\begin{equation}
Z_{\lambda /d/\mu }(x_{1},\ldots ,x_{\ell })=\sum_{0<i_{1}<\cdots <i_{d}\leq
\ell }\langle \lambda |A(x_{1})\cdots D(x_{i_{1}})\cdots D(x_{i_{d}})\cdots
A(x_{\ell })|\mu \rangle \;.  \label{ZEd}
\end{equation}%
In particular, the partition function $Z_{\lambda /d/\mu }(x_{1},\ldots
,x_{\ell })$ is symmetric in the variables $(x_{1},\ldots ,x_{\ell })$.
\end{lemma}

\begin{proof}
The first identity follows from the definition (\ref{ncE}). Due to (\ref%
{integrable}), which implies $\boldsymbol{E}(x_{i})\boldsymbol{E}(x_{j})=%
\boldsymbol{E}(x_{j})\boldsymbol{E}(x_{i})$, we know that the sum in (\ref%
{ZE}) must be symmetric in $(x_{1},\ldots ,x_{\ell })$. But since $z$ is an
independent formal variable (the weights of the lattice configurations do
not depend on $z$) we can conclude that each single summand must be
symmetric as well.
\end{proof}

\begin{remark}
The partition function (\ref{ZE}) for $z=1$ is related to the so-called $q$%
-boson model discussed by Bogoliubov, Izergin and Kitanine in \cite%
{Bogoliubovetal}. This model is \emph{quantum integrable} and the associated
discrete Hamiltonians are
\begin{equation}
H_{r}^{\pm }:=-\frac{\boldsymbol{e}_{r}\pm \boldsymbol{e}_{n-r}}{2},\qquad
1\leq r\leq n/2\;.  \label{IoM}
\end{equation}
In particular, the operator%
\begin{equation}
H_{1}^{+}=-\frac{1}{2}\sum_{i\in \mathbb{Z}_{n}}\left( \beta _{i}\beta
_{i+1}^{\ast }+\beta _{i}^{\ast }\beta _{i+1}\right) ,  \label{qbosonHam}
\end{equation}%
describes nearest neighbour hopping of certain highly-correlated quantum
particles, so-called $q$-bosons, on a one-dimensional periodic lattice (i.e.
a circle) with $n$ sites. The interaction between the particles is encoded
in the $q$-deformation. This model can also be interpreted as a quantization
of the Ablowitz-Ladik hierarchy in integrable systems. The latter describes
a discrete version of the nonlinear Schr\"{o}dinger model.
\end{remark}

\begin{corollary}
According to Proposition \ref{YBPieri} the action of $\boldsymbol{E}%
(u)=\sum_{r\geq 0}u^{r}\boldsymbol{e}_{r}$ is given by%
\begin{equation}
\boldsymbol{e}_{r}|\mu \rangle =\sum_{\substack{ \lambda -\mu =(r),  \\ %
\lambda \in \mathcal{A}_{k,n}^{+}}}\varphi _{\lambda /\mu }(t)|\lambda
\rangle +z\sum_{\substack{ \mu -\lambda =(n-r),  \\ \lambda \in \mathcal{A}%
_{k,n}^{+}}}\psi _{\mu /\lambda }(t)|\lambda \rangle \;.  \label{Epieri}
\end{equation}%
and for $\bar{z}=z^{-1},~\bar{u}=u$ we have $\boldsymbol{E}^{\ast
}(u)=z^{-1}u^{n}\boldsymbol{E}(u^{-1})$.
\end{corollary}

We now wish to derive an expansion of the partition function $Z_{\lambda
/d/\mu }$ of the $q$-boson model on the cylinder,which is analogous to the
formulae appearing in Corollary \ref{YBHL}. Before we can do so, we need
some additional combinatorial notions.

\subsection{Cylindric loops and cylindric skew tableaux}

Recall that a skew diagram can be seen as subset of $\mathbb{Z}^{2}$, $%
\lambda /\mu :=\left\{ \langle i,j\rangle \in \mathbb{Z}^{2}:1\leq i\leq
\ell (\lambda ),\;\mu _{i}<j\leq \lambda _{i}\right\} $. Inspired by the
discussion in \cite{GesselKrattenthaler}, \cite{Postnikov} and \cite%
{McNamara} we now define \emph{cylindric} \emph{skew diagrams}. The borders
of such cylindric skew diagram are so-called cylindric loops, which are a
periodic generalisation of the Ferrer's shapes or Young diagrams of ordinary
partitions.

\begin{definition}[cylindric loop]
Let $\lambda \in \mathcal{A}_{k,n}^{+}$ and $r\in \mathbb{Z}$. Define the
associated cylindric loop $\lambda \lbrack r]=(\ldots ,\lambda \lbrack
r]_{-1},\lambda \lbrack r]_{0},\lambda \lbrack r]_{1},\lambda \lbrack
r]_{2},\ldots )$ as the sequence obtained from $\lambda \lbrack
r]_{i}:=\lambda \lbrack 0]_{i-r}$ with $\lambda \lbrack 0]_{i}:=\lambda _{i}$
for $1\leq i\leq k$ and $\lambda \lbrack 0]_{i+\ell }:=\lambda \lbrack
0]_{i}-\ell n$ outside this interval.
\end{definition}

In other words, $\lambda \lbrack 0]$ is the path in $\mathbb{Z\times Z}$
traced out when translating the outline of the Young diagram of $\lambda $
by the period vector $\Omega =(k,-n)$. The loop $\lambda \lbrack r]$ is then
its shift by the vector $(r,0)$. The reason for the name \textquotedblleft
cylindric loop\textquotedblright\ should be apparent from the definition;
see Figure \ref{fig:cylindric_loop} for an example.

\begin{figure}[tbp]
\begin{equation*}
\includegraphics[scale=0.4]{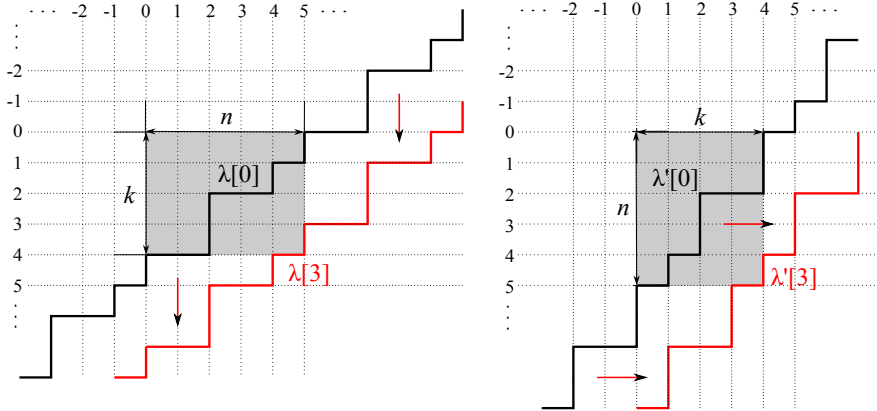}
\end{equation*}%
\caption{On the left are shown two cylindric loops constructed from the
partition $\protect\lambda=(5,4,2,2)$. The loop $\protect\lambda[0]$ is
simply a periodic continuation of the outline of the Young diagram of $%
\protect\lambda$. The loop $\protect\lambda[3]$ is then obtained by shifting
three times in the direction of the vector $(1,0)$. On the left are the
corresponding conjugate cylindric loops.}
\label{fig:cylindric_loop}
\end{figure}

\begin{definition}[cylindric skew diagram]
Let $\lambda ,\mu \in \mathcal{A}_{k,n}^{+}$. Define the $(n,k)$-restricted
\emph{cylindric skew diagram of degree }$d\geq 0$ as the set%
\begin{equation}
\lambda /d/\mu :=\{\langle i,j\rangle \in \mathbb{Z}\times \mathbb{Z}%
~|~\lambda \lbrack d]_{i}\geq j>\mu \lbrack 0]_{i}\}~.
\label{cylindric_diagram}
\end{equation}%
A cylindric skew diagram which in each column (row) has at most one square
is called a cylindric horizontal (vertical) strip.
\end{definition}

Broadly speaking the cylindric skew diagram is the \textquotedblleft
periodic continuation\textquotedblright\ of the skew diagram $\Lambda
^{(d)}/\mu $ by the period vector $\Omega =(k,-n)$, where $\Lambda ^{(d)}$
is obtained by adding $d$ parts of size $n$ to $\lambda $. Letting $\lambda
,\mu \in \mathcal{A}_{k,n}^{+}$ it follows that the cylindric diagram $%
\lambda /d/\mu $ can have at most $k$ squares in each column and at most $n$
elements in each row. Hence, we call such diagrams $(n,k)$-restricted%
\footnote{%
Postnikov called such restricted cylindric shapes \textquotedblleft
toric\textquotedblright ; see \cite[Def 3.2]{Postnikov} and
\cite[Paragraph after eqn (4.1) on page 289]{McNamara}.} and shall
henceforth only consider such.

The \emph{fundamental region} of a cylindric skew diagram $\lambda /d/\mu $
is the following finite set of squares between the two loops $\lambda
\lbrack d]$ and $\mu \lbrack 0]$,
\begin{equation*}
\lbrack \lambda /d/\mu ]_{0}:=\left\{ \langle i,j\rangle \in \mathbb{Z}%
^{2}:1\leq i\leq k+d,\;\mu _{i}[0]<j\leq \lambda \lbrack d]_{i},\;\ell \in
\mathbb{Z}\right\} \;.
\end{equation*}%
Conversely, given a $(n,k)$-restricted cylindric skew diagram $\Theta $ one
can reconstruct the partitions $\lambda ,\mu \in \mathcal{A}_{k,n}^{+}$ and
degree $d$. It is easiest to explain this on a concrete example like the one
shown in Figure \ref{fig:cylindric_tableau}.

One first identifies a fundamental region, that is one fixes a bounding box
of height $k$ and width $n$ in the $\mathbb{Z\times Z}$ lattice such that
the squares contained in it generate the entire diagram upon shifting with $%
(k,-n)$. Then one defines $\mu $ to be the partition whose Young diagram is
cut out by the bounding box and the upper boundary of the cylindric skew
diagram. To obtain the partition $\lambda $ consider first the partition $%
\Lambda $ obtained by adding to $\mu $ all the boxes of the cylindric skew
diagram within the specified $n$-strip. Then remove from $\Lambda $ rows of
size $n$ until the associated Young diagram has height $k$, i.e. fits into
the bounding box. The degree $d$ is the number of $n$-rows removed. Note
that for $d=0$ we recover the familiar skew-diagram of two partitions, i.e. $%
[\lambda /0/\mu ]_{0}=\lambda /\mu $.

\begin{remark}
The set $\lambda /d/\mu $ is a particular way of parametrizing certain \emph{%
cylindric skew diagrams} or \emph{cylindric skew shapes }with period vector $%
\Omega =(k,-n)$; compare with the general definition of (reversed) cylindric
plane partitions in \cite{GesselKrattenthaler} and \cite[Section 3]{McNamara}
for a discussion in terms of oriented posets. While we use the same notation
as in \cite[Section 4]{McNamara} and \cite{Postnikov}, our parametrization
of a cylindric skew diagram in terms of two partitions $\lambda $, $\mu $
and a degree $d$ differs from the one used in \emph{loc. cit.} in order to
accommodate the description of the Verlinde algebra given in \cite{KS}.
\end{remark}

For our discussion we will also need the transposed or \emph{conjugate
cylindric skew diagram} which we denote by $\lambda ^{\prime }/d/\mu
^{\prime }$ and is defined as the set
\begin{equation}
\lambda ^{\prime }/d/\mu ^{\prime }:=\left\{ \langle i+\ell n,j-\ell
k\rangle \in \mathbb{Z}^{2}:1\leq i\leq n,\;\mu _{i}^{\prime }<j\leq \lambda
_{i}^{\prime }+d,\;\ell \in \mathbb{Z}\right\} \;.  \label{ccsd}
\end{equation}%
The latter can also be described in terms of\emph{\ conjugate cylindric loops%
} $\lambda ^{\prime }[r]$. Consider first the case $r=0$ and set $\lambda
^{\prime }[0]_{i}=\lambda _{i}^{\prime }$ for $1\leq i\leq n$ and $\lambda
^{\prime }[0]_{i+n}=\lambda _{i}-ik$ otherwise. The conjugate cylindric loop
$\lambda ^{\prime }[r]$ is then obtained by shifting $\lambda ^{\prime }[0]$
by the vector $(0,r)$, i.e. we set $\lambda ^{\prime }[r]_{i}=\lambda
^{\prime }[0]_{i}+r$. It is then straightforward to verify that $\lambda
^{\prime }/d/\mu ^{\prime }=\{\langle i,j\rangle \in \mathbb{Z}_{n}\times
\mathbb{Z}_{k}~|~\lambda ^{\prime }[d]_{i}\geq j>\mu ^{\prime }[0]_{i}\}$
and that the transposed or conjugate cylindric skew diagram is obtained by
transposing $\lambda /d/\mu $; see Figure \ref{fig:cylindric_loop}. In the
introduction we discussed these conjugate cylindric loops and skew diagrams.
\begin{definition}[cylindric skew tableau]
Let $\Theta $ be a $(n,k)$-restricted cylindric skew diagram. A cylindric
(semi-standard) skew tableau is a map $T:\Theta \rightarrow \mathbb{N}$ such
that for any $\langle i,j\rangle \in \Theta $ one has%
\begin{eqnarray}
T(i,j) &=&T(i+k,j-n),  \label{T1} \\
T(i,j) &<&T(i+1,j),\;\text{if}\;\langle i+1,j\rangle \in \Theta  \label{T2}
\\
T(i,j) &\leq &T(i,j+1),\;\text{if}\;\langle i,j+1\rangle \in \Theta \;.
\label{T3}
\end{eqnarray}%
The weight vector $\limfunc{wt}(T)=(t_{1},\ldots ,t_{k})$ of a cylindric
tableau is defined by setting $t_{i}$ to be the number of $i$-entries in $T$
in the fundamental region.
\end{definition}

In other words, a tableau $T$ of a (restricted) cylindric shape $\lambda
/d/\mu $ is a filling of the squares of the associated diagram with integers
such that in each row the numbers are weakly increasing (left to right),
while they strictly increase (top to bottom) in each column; see Figure \ref%
{fig:cylindric_tableau} for an example.

\begin{definition}
The \emph{cylindric Kostka number\footnote{%
Compare with the \textquotedblleft quantum Kostka number\textquotedblright\
introduced by Bertram, Ciocan-Fontanine, Fulton \cite{BCF} in the context of
the quantum cohomology ring of the Grassmannian.}} $K_{\lambda /d/\mu
,\theta }$ is defined to be the number of cylindric tableaux of weight $%
\limfunc{wt}(T)=\theta $, and specialises for $d=0$ to the ordinary Kostka
number $K_{\lambda /\mu ,\theta }$.
\end{definition}

\begin{figure}[tbp]
\begin{equation*}
\includegraphics[scale=0.45]{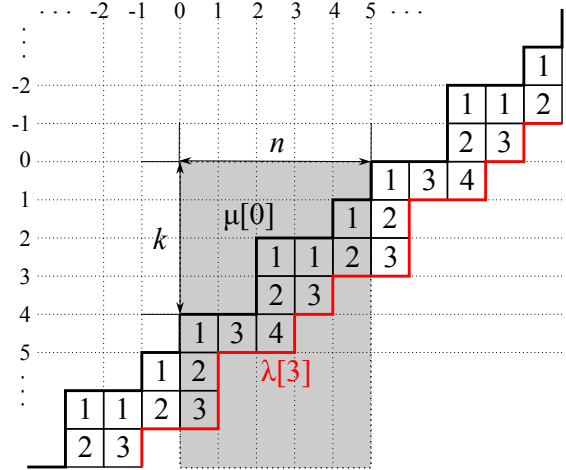}
\end{equation*}%
\caption{Example of a cylindric skew tableau. A set of representatives for
the cylindric skew shape is shown in grey. Fix the $(n,k)$-bounding box as
shown, then the upper boundary of the cylindric skew diagram cuts out the
partition $\protect\mu=(5,4,2,2)$ while the lower one gives the partition $%
\Lambda =(5,5,5,4,3,1,1)$. Remove 3 rows of size $n=5$ from the Young
diagram of $\Lambda $ to obtain the partition $\protect\lambda =(4,3,1,1)$
which now lies within the bounding box. Denote the cylindric skew diagram by
$\protect\lambda/d/\protect\mu$. The depicted filling of the boxes with
integers 1 to 4 yields an example of a cylindric tableau, the integers are
weakly increasing in each row and strictly increasing in each column.}
\label{fig:cylindric_tableau}
\end{figure}

\begin{remark}
Note that because of condition (\ref{T3}) not every ordinary skew tableau of
shape $\Lambda ^{(d)}/\mu $ (with $\Lambda ^{(d)}$ being the partition
obtained from $\lambda $ by adding $d$ parts of size $n$) gives rise to a
cylindric tableau via periodic continuation with respect to the vector $%
\Omega =(k,-n)$. For instance, consider the ordinary skew tableau%
\begin{equation*}
\young(::::3,::114,::23,134,2,3)
\end{equation*}
which is of the shape of the skew diagram shown in the grey box in Figure %
\ref{fig:cylindric_tableau}. This tableau does not give rise to a cylindric
skew tableau when periodically continued.
\end{remark}

\begin{lemma}
\label{cyltab}Any cylindric skew tableau $T$ of shape $\lambda /d/\mu $ with
$\lambda ,\mu \in \mathcal{A}_{k,n}^{+}$ is equivalent to a sequence of
cylindric loops
\begin{equation*}
(\lambda ^{(0)}[d_{0}=0]=\mu \lbrack 0],\lambda ^{(1)}[d_{1}],\ldots
,\lambda ^{(r)}[d_{r}=d]=\lambda \lbrack d])
\end{equation*}
with $\lambda ^{(a)}\in \mathcal{A}_{k,n}^{+}$\ and $0\leq d_{a}-d_{a-1}\leq
1$ such that
\begin{equation*}
\lambda ^{(a)}/(d_{a}-d_{a-1})/\lambda ^{(a-1)}:=\{\langle i,j\rangle \in
\mathbb{Z}\times \mathbb{Z}~|~\lambda _{i}^{(a)}[d_{a}]\geq j>\lambda
_{i}^{(a-1)}[d_{a-1}]\}
\end{equation*}%
is a cylindric horizontal strip.
\end{lemma}

\begin{proof}
Given $T$ define $\lambda ^{(a)}/(d_{a}-d_{a-1})/\lambda ^{(a-1)}$ to be the
set of squares filled with $a$ in $T$ and let $\lambda ^{(a)}[d_{a}],\lambda
^{(a-1)}[d_{a-1}]$ be its lower and upper boundary. Because of (\ref{T2}) $%
\lambda ^{(a)}/(d_{a}-d_{a-1})/\lambda ^{(a-1)}$ must be a horizontal strip
and, thus, we have $d_{a}-d_{a-1}=0,1$ otherwise there would be more than
one box in the first column of the skew diagram in the fundamental region.
We obviously have that $\lambda ^{(0)}=\mu $ and $\lambda ^{(r)}=\lambda $.
The converse statement, that such a sequence of cylindric loops gives a
cylindric tableau, is equally obvious. Requirement (\ref{T1}) is satisfied
because each $\lambda ^{(a)}/(d_{a}-d_{a-1})/\lambda ^{(a-1)}$ is cylindric
and, by definition, the degrees $d_{a}$ of the individual loops accumulate
to give $d$. (\ref{T2}) and (\ref{T3}) are satisfied since each cylindric
sub-diagram $\lambda ^{(a)}/(d_{a}-d_{a-1})/\lambda ^{(a-1)}$ is horizontal.
\end{proof}

\subsection{Cylindric Hall-Littlewood functions}

We now generalise the notion of skew Hall-Littlewood functions to cylindric
skew diagrams and then link them to the partition function of the
quasi-periodic $q$-boson model, analogous to Corollary \ref{YBHL}.

We start with the generalisation of the coefficient functions (\ref{phi})
and (\ref{psi}) to cylindric skew tableaux. A cylindric horizontal strip of
period $(k,-n)$ can obviously always be written as $\lambda /d/\mu $ with $%
\lambda ,\mu \in \mathcal{A}_{k,n}^{+}$ and either $d=0$ or $d=1$ depending
on the strip fitting within the $(n,k)$-bounding box or not. We introduce
the following generalisations of the functions (\ref{phi}) and (\ref{psi}):
let%
\begin{equation}
\Phi _{\lambda /d/\mu }(t)=\left\{
\begin{array}{cc}
\prod_{i\in I_{\lambda /d/\mu }}(1-t^{m_{i}(\lambda )}), & \lambda /d/\mu
\text{ is a horizontal strip} \\
0, & \text{else}%
\end{array}%
\right.  \label{Phi}
\end{equation}%
and%
\begin{equation}
\Psi _{\lambda /d/\mu }(t)=\left\{
\begin{array}{cc}
\prod_{i\in J_{\lambda /d/\mu }}(1-t^{m_{i}(\mu )}), & \lambda /d/\mu \text{
is a horizontal strip} \\
0, & \text{else}%
\end{array}%
\right. \;.  \label{Psi}
\end{equation}%
Here the sets $I_{\lambda /d/\mu }$ and $J_{\lambda /d/\mu }$ are defined as
follows: include $1\leq i\leq n$ in $I_{\lambda /d/\mu }$ if $\lambda
^{\prime }[d]_{i}-\mu ^{\prime }[0]_{i}=1$ and $\lambda ^{\prime
}[d]_{i+1}-\mu ^{\prime }[0]_{i+1}=0$, while $i\in J_{\lambda /d/\mu }$ if
and only if $\lambda ^{\prime }[d]_{i}-\mu ^{\prime }[0]_{i}=0$ and $\lambda
^{\prime }[d]_{i+1}-\mu ^{\prime }[0]_{i+1}=1$.

\begin{example}
Set $n=5,k=6$ and $\ell =3$. Consider the cylindric tableau shown on the
right in Figure \ref{fig:qbosonlatticeconfig}. Let us first compute the sets
$I_{\lambda ^{(i+1)}/d_{i}/\lambda ^{(i)}}$ for the horizontal strips
describing the first (top) skew tableau $T_{1}$. This cylindric skew tableau
corresponds to the sequence $(\mu \lbrack 0]=\lambda ^{(0)}[0],\lambda
^{(1)}[d_{1}],\lambda ^{(2)}[d_{2}],\lambda ^{(3)}[d_{3}]=\lambda \lbrack
d])$ with $\mu \lbrack 0]=(\ldots ,5,5,4,3,1,1,\ldots )$, $\lambda
^{(1)}[0]=(\ldots ,6,5,5,3,3,1,\ldots )$, $\lambda ^{(2)}[1]=(\ldots
,6,5,5,4,3,2,1,\ldots )$ and $\lambda \lbrack d=1]=(\ldots
,6,6,5,4,4,2,1,\ldots )$. We then find that $I_{\lambda ^{(1)}/0/\mu \lbrack
0]}=\{3,6\}$, $I_{\lambda ^{(2)}/1/\lambda ^{(1)}}=\{4\}$ and $I_{\lambda
^{(3)}/d/\lambda ^{(2)}}=\{4,6\}$.

Similarly, we find for the second (bottom) tableau $T_{2}$ in Figure \ref%
{fig:qbosonlatticeconfig} the sequence $(\lambda \lbrack 0]=\mu
^{(0)}[0],$ $\mu ^{(1)}[1],\mu ^{(2)}[1],$ $\mu ^{(3)}[2]=\mu \lbrack 2])$ with
the cylindric loops $\mu ^{(1)}[1]=(\ldots ,6,5,5,4,3,2,1,\ldots )$, $\mu
^{(2)}[1]=(\ldots ,6,6,5,5,3,3,1,\ldots )$. Following the prescription
outlined above we compute $J_{\mu ^{(1)}/1/\lambda }=\{4,6\}$, $J_{\mu
^{(2)}/1/\mu ^{(1)}}=\{4\}$ and $J_{\mu ^{(3)}/2/\mu ^{(2)}}=\{3,6\}$. As
the alert reader will have noticed these are the same sets as before but
calculated in reverse order, since both results ought to reproduce the \emph{%
same} Boltzmann weight $\Phi _{T_{1}}(t)=\Psi
_{T_{2}}(t)=(1-t)^{4}(1-t^{2})^{2}$ of the integrable $q$-boson lattice for
the lattice configuration shown on the left in Figure \ref%
{fig:qbosonlatticeconfig}; see Theorem \ref{qboson_Z}.
\end{example}

\begin{definition}[cylindric Hall-Littlewood functions]
Given a $(n,k)$-restricted skew diagram $\lambda /d/\mu $ with $\lambda ,\mu
\in \mathcal{A}_{k,n}^{+}$ define the (restricted) cylindric skew
Hall-Littlewood functions for $1\leq \ell \leq k$ as%
\begin{equation}
Q_{\lambda /d/\mu }(x_{1},\ldots ,x_{\ell };t):=\sum_{|T|=\lambda /d/\mu
}\Phi _{T}(t)x^{T},  \label{cylQ}
\end{equation}%
and%
\begin{equation}
P_{\lambda /d/\mu }(x_{1},\ldots ,x_{\ell };t):=\sum_{|T|=\lambda /d/\mu
}\Psi _{T}(t)x^{T}~,  \label{cylP}
\end{equation}%
where the sums run over all cylindric skew tableaux of shape $\lambda /d/\mu
$ and $\Phi _{T},\Psi _{T}$ are defined as the products of the functions (%
\ref{Phi}) and (\ref{Psi}) obtained when decomposing $T$ into a sequence of
horizontal strips. For $\ell >k$ we set both functions to zero.
\end{definition}

\begin{lemma}
As in the non-cylindric case we have the identity%
\begin{equation}
\Phi _{\lambda /d/\mu }=\frac{b_{\lambda }}{b_{\mu }}~\Psi _{\lambda /d/\mu }
\end{equation}%
and, thus, $Q_{\lambda /d/\mu }=(b_{\lambda }/b_{\mu })P_{\lambda /d/\mu }$;
compare with (\ref{bphi=bpsi}).
\end{lemma}

\begin{proof}
Set $\theta _{i}=\lambda ^{\prime }[d]_{i}-\mu ^{\prime }[0]_{i}$ and
observe that
\begin{equation*}
m_{i}(\lambda )=\lambda _{i}^{\prime }-\lambda _{i+1}^{\prime }=\lambda
^{\prime }[d]_{i}-\lambda ^{\prime }[d]_{i+1},
\end{equation*}%
\ and $\theta _{i}-\theta _{i+1}=m_{i}(\lambda )-m_{i}(\mu )$. Then%
\begin{eqnarray*}
\frac{b_{\lambda }(t)}{b_{\mu }(t)} &=&\frac{(t)_{m_{1}(\lambda
)}(t)_{m_{2}(\lambda )}\cdots (t)_{m_{n}(\lambda )}}{(t)_{m_{1}(\mu
)}(t)_{m_{2}(\mu )}\cdots (t)_{m_{n}(\mu )}} \\
&=&\prod_{\theta _{i}-\theta _{i+1}=1}(1-t^{m_{i}(\lambda )})\prod_{\theta
_{i}-\theta _{i+1}=-1}(1-t^{m_{i}(\mu )})^{-1}=\Phi _{\lambda /d/\mu }/\Psi
_{\lambda /d/\mu },
\end{eqnarray*}%
since $\theta _{n+1}=\lambda ^{\prime }[d]_{n+1}-\mu ^{\prime
}[0]_{n+1}=\lambda ^{\prime }[d]_{1}-\mu ^{\prime }[0]_{1}$.
\end{proof}

With these definitions at hand we can now rewrite the action of the transfer
matrix (\ref{ncE}) in terms of cylindric horizontal strips.

\begin{lemma}[cylindric horizontal strips]
\label{cylhorstrip}We have the following alternative formulae for the action
of the noncommutative elementary symmetric polynomials on the Fock space,%
\begin{equation}
\boldsymbol{e}_{r}|\mu \rangle =\sum_{\substack{ \lambda /d/\mu =(r),  \\ %
\lambda \in \mathcal{A}_{k,n}^{+}}}z^{d}\Phi _{\lambda /d/\mu }(t)|\lambda
\rangle =\sum_{\substack{ \mu /d/\lambda =(n-r),  \\ \lambda \in \mathcal{A}%
_{k,n}^{+}}}z^{1-d}\Psi _{\mu /d/\lambda }(t)|\lambda \rangle \;.
\label{nce2}
\end{equation}%
Note that the length of the horizontal strip in the fundamental region is
different in both cases and that $d=0$ or $1$.
\end{lemma}

\begin{proof}
To prove the assertion it suffices to show the identities $\psi _{\mu
/\lambda }=\Phi _{\lambda /1/\mu }$, $\varphi _{\lambda /\mu }=\Psi _{\mu
/1/\lambda }$ together with the fact that $\lambda /1/\mu $ being a
horizontal $r$-strip implies that $\mu -\lambda =(n-r)$ and $\mu /1/\lambda $
being a horizontal $n-r$ strip is equivalent to $\lambda -\mu =(r)$. The
last two claims are obvious from the definition of the cylindric skew
diagram, since $\lambda \lbrack 1]_{i}^{\prime }-\mu _{i}[0]^{\prime }=1$
means that $\mu _{i}^{\prime }-\lambda _{i}^{\prime }=0$ and vice versa.
Recall that $\lambda \lbrack 1]$ confined to the fundamental region is the
Young diagram of $\lambda $ with one row of length $n$ added. Similarly, we
have that $\mu \lbrack 1]_{i}^{\prime }-\lambda _{i}[0]^{\prime }=0,1$ is
equivalent to $\lambda _{i}^{\prime }-\mu _{i}{}^{\prime }=1,0$. The
equality of the coefficient functions is now a trivial consequence of the
definitions (\ref{phi}), (\ref{psi}) and (\ref{Phi}), (\ref{Psi}).
\end{proof}

\begin{figure}[tbp]
\begin{equation*}
\includegraphics[scale=0.3]{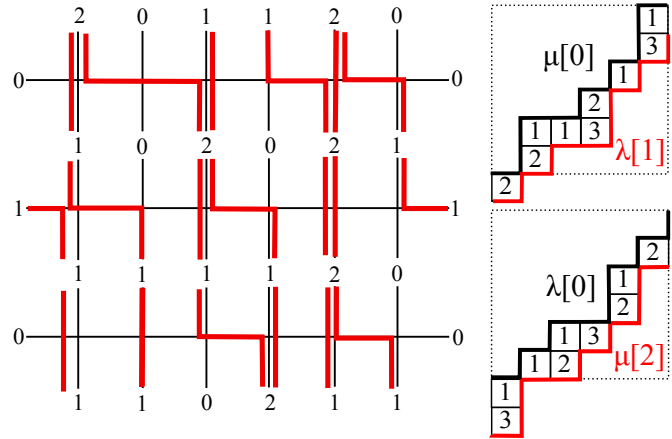}
\end{equation*}%
\caption{Example of a lattice configuration for the $q$-boson model with
periodic boundary conditions. On the right the two cylindric skew tableau
obtained from bijection 1 (top) and bijection 2 (bottom) in the proof of
Theorem \protect\ref{qboson_Z}. Here $\protect\lambda=(6,5,4,4,2,1)$ and $%
\protect\mu=(5,5,4,3,1,1)$.}
\label{fig:qbosonlatticeconfig}
\end{figure}

\begin{theorem}
\label{qboson_Z}The cylindric skew HL functions are symmetric and we have
the following identities with the quantum integrable $q$-boson model,%
\begin{eqnarray}
Z_{\lambda /d/\mu }(x_{1},\ldots ,x_{\ell }) &=&Q_{\lambda /d/\mu
}(x_{1},\ldots ,x_{\ell };t)  \label{ZE1} \\
&=&(x_{1}\cdots x_{\ell })^{n}P_{\mu /(\ell -d)/\lambda }(x_{1}^{-1},\ldots
,x_{\ell }^{-1};t)\;,  \label{ZE2}
\end{eqnarray}%
where $1\leq \ell \leq k$ as before.
\end{theorem}

\begin{proof}
We only need to prove the case $\ell =1$, since according to Lemma \ref%
{cyltab}, definitions (\ref{cylQ}), (\ref{cylP}) and the fact that the
weight of a lattice configuration equals the product of its row
configurations the general case with $\ell \geq 1$ then trivially follows.
We prove both assertions by formulating bijections between lattice
configurations and cylindric tableau starting with (\ref{ZE1}).\medskip\

\noindent \textbf{Lattice-cylindric tableau bijection 1}. Fix $\lambda ,\mu
\in \mathcal{A}_{k,n}^{+}$ and assume $\ell =1$. Given an allowed row
configuration $\sigma =(\sigma _{1},\ldots ,\sigma _{n},\sigma _{1})$ we
find from (\ref{rowconfig}) that $\sigma _{i}=\lambda \lbrack \sigma
_{1}]_{i}^{\prime }-\mu \lbrack 0]_{i}^{\prime }=0,1$ for $1\leq i\leq n$,
where $\lambda \lbrack \sigma _{1}]_{i}^{\prime }$, $\mu \lbrack
0]_{i}^{\prime }$ are the heights of the columns of the cylindric loops $%
\lambda \lbrack \sigma _{1}],\mu \lbrack 0]$ in the fundamental region.
Conversely, given a horizontal strip $\lambda /d/\mu $ we must have $d=0$ or
$1$, and, thus, we can define a unique row configuration via $\sigma
_{i}:=\lambda \lbrack \sigma _{1}]_{i}^{\prime }-\mu \lbrack 0]_{i}^{\prime
} $ with $\sigma _{1}:=\sigma _{n+1}:=d$. Using the first equality of (\ref%
{nce2}) in Lemma \ref{cylhorstrip} the identity (\ref{ZE1}) follows.

For the benefit of the reader we describe the bijection between the set $%
\Gamma _{\lambda /d/\mu }$ of allowed lattice configurations (w.r.t. the
vertex configurations shown in Figure (\ref{fig:qbosonvertex})) and the set
of cylindric skew tableaux of shape $\lambda /d/\mu $ also for the general
case when $\ell \geq 1$. It is instructive to consult the example shown in
Figure \ref{fig:qbosonlatticeconfig}. Beginning at the first lattice row add
a box with entry 1 in each column of the Young diagram of $\mu $ whenever
the horizontal edge in the corresponding lattice column has value one. (This
is analogous to the non-cylindric case discussed previously.) After arriving
at the $n$th lattice column one obtains a finite horizontal strip, since in
each column of the Young diagram at most one box is added. Continue the
resulting horizontal strip from the fundamental region by shifting with the
period vector $\Omega =(k,-n)$ to obtain the \emph{cylindric} horizontal
strip. Now do the same for the second lattice row labeling the boxes with 2
and so on. Continue until the last lattice row. The result is a uniquely
defined cylindric skew tableau of shape $\lambda /d/\mu $. This bijection
preserves weights and it generalises the one stated earlier in the context
of the $A$-operator of the Yang-Baxter algebra. \medskip

\noindent \textbf{Lattice-cylindric tableau bijection 2}. For given $\lambda
,\mu \in \mathcal{A}_{k,n}^{+}$ and $\ell =1$ we now set $1-\sigma _{i}=\mu
\lbrack 1-\sigma _{1}]_{i}^{\prime }-\lambda \lbrack 0]_{i}^{\prime }$ in
order to identify a row configuration with a cylindric strip $\mu /d/\lambda
$, where $d=1-\sigma _{1}$; compare with the second identity in (\ref{nce2})
in Lemma \ref{cylhorstrip}. This proves (\ref{ZE2}) for $\ell =1$ and the
general case then follows by the same arguments as before.

Again for completeness we state the bijection also in the general case $\ell
\geq 1$; it generalises the one for the $D$-operator from the previous
section. Beginning with the bottom lattice row starting at $\langle 0,\ell
\rangle $ place squares labelled with $\ell $ in each column of the extended
Young diagram of $\mu \lbrack \ell -d]$, if the value of horizontal edge in
the respective lattice column modulo $n$ is zero. Continue to the next row
and label the squares with $\ell -1$ and so on. The result is a unique
restricted skew cylindric tableau of shape $\mu /(\ell -d)/\lambda $.
\end{proof}


\begin{corollary}[Expansions of cylindric Hall-Littlewood functions]
Denote by $\mathcal{P}_{k,n}^{+}$ the set of partitions $\lambda $ with $%
\ell (\lambda )\leq k$ and $\lambda _{1}\leq n$. Then%
\begin{eqnarray}
Q_{\lambda /d/\mu }(x;t) &=&\sum_{\nu \in \mathcal{P}_{k,n}^{+}}\left\langle
\lambda |\boldsymbol{e}_{\nu }|\mu \right\rangle m_{\nu }(x)=\sum_{\nu \in
\mathcal{P}_{k,n}^{+}}\left\langle \lambda |\boldsymbol{s}_{\nu ^{\prime
}}|\mu \right\rangle s_{\nu }(x),  \label{cylQ2m&s} \\
&=&\sum_{\nu \in \mathcal{P}_{k,n}^{+}}\left\langle \lambda |\boldsymbol{P}%
_{\nu ^{\prime }}^{\prime }|\mu \right\rangle P_{\nu }(x;t),
\end{eqnarray}%
and%
\begin{equation}
\left\langle \nu |\boldsymbol{e}_{\lambda }|\mu \right\rangle =\sum_{\sigma
\in \mathcal{P}_{k,n}^{+}}\left\langle \nu |\boldsymbol{s}_{\sigma ^{\prime
}}|\mu \right\rangle K_{\sigma \lambda }=\sum_{\substack{ |T|=\nu /d/\mu  \\
\limfunc{wt}(T)=\lambda }}\Psi _{T}(t),  \label{ematrix}
\end{equation}%
where $K=K(1)$ is the matrix of Kostka numbers and, hence,%
\begin{equation*}
\left\langle \nu |\boldsymbol{s}_{\lambda ^{\prime }}|\mu \right\rangle
=\sum_{w\in \mathfrak{S}_{\ell }}\varepsilon (w)\sum_{\substack{ |T|=\nu
/d/\mu  \\ \limfunc{wt}T=\lambda (w)}}\Psi _{T}(t)\;
\end{equation*}%
with $\lambda (w)=(\lambda _{1}-1+w_{1},\ldots ,\lambda _{\ell }-\ell
+w_{\ell })$.
\end{corollary}

\begin{proof}
These expansion are an immediate consequence of the noncommutative Cauchy
identities (\ref{ncC1}).
\end{proof}

\begin{remark}
For $d=0$ (\ref{ematrix}) becomes $\sum_{T}\psi
_{T}(t)=b_{\lambda }(t)(K(t)^{-1}K)_{\lambda \mu }$ \cite[III.6, eqn (6.4)
on page 239]{Macdonald}. Note that in general the expansion of $Q_{\lambda
/d/\mu }(x;t)$ into Hall-Littlewood $Q$-functions does not yield polynomial
coefficients in $t$, that is, $\left\langle \lambda |\boldsymbol{P}_{\nu
^{\prime }}^{\prime }|\mu \right\rangle /b_{\nu ^{\prime }}(t)$ is in
general a rational function in $t$. In contrast, $\left\langle \lambda |%
\boldsymbol{P}_{\nu ^{\prime }}^{\prime }|\mu \right\rangle b_{\mu
}(t)/b_{\lambda }(t)$ is always polynomial in $t$.
\end{remark}

\begin{example}
Set $n=4$ and $k=3$. The following table summarises the expansion of the
cylindric Hall-Littlewood function $P_{(3,2,1)/2/(4,3,1)}$ in various bases
of the ring of symmetric functions.

\begin{equation*}
\begin{tabular}{|c|c|c|c|}
\hline\hline
$\lambda $ & $m_{\lambda }$ & $s_{\lambda }$ & $P_{\lambda }$ \\ \hline\hline
\multicolumn{1}{|l|}{$4,2,0$} & \multicolumn{1}{|l|}{$1-t$} &
\multicolumn{1}{|l|}{$1-t$} & \multicolumn{1}{|l|}{$1-t$} \\ \hline
\multicolumn{1}{|l|}{$4,1,1$} & \multicolumn{1}{|l|}{$2-3t+t^{3}$} &
\multicolumn{1}{|l|}{$1-2t+t^{3}$} & \multicolumn{1}{|l|}{$1-t-t^{2}+t^{3}$}
\\ \hline
\multicolumn{1}{|l|}{$3,3,0$} & \multicolumn{1}{|l|}{$2-3t+t^{3}$} &
\multicolumn{1}{|l|}{$1-2t+t^{3}$} & \multicolumn{1}{|l|}{$1-t-t^{2}+t^{3}$}
\\ \hline
\multicolumn{1}{|l|}{$3,2,1$} & \multicolumn{1}{|l|}{$%
\begin{tabular}{l}
${\small 6-14t+5t}^{2}$ \\
$\;\;{\small +9t}^{3}{\small -7t}^{4}{\small +t}^{5}$%
\end{tabular}%
$} & \multicolumn{1}{|l|}{$%
\begin{tabular}{l}
${\small 2-8t+5t}^{2}$ \\
$\;\;{\small +7t}^{3}{\small -7t}^{4}{\small +t}^{5}$%
\end{tabular}%
$} & \multicolumn{1}{|l|}{$%
\begin{tabular}{l}
${\small 2-5t+t}^{2}$ \\
$\;\;{\small +6t}^{3}{\small -5t}^{4}{\small +t}^{5}$%
\end{tabular}%
$} \\ \hline
\multicolumn{1}{|l|}{$2,2,2$} & \multicolumn{1}{|l|}{$%
\begin{tabular}{l}
${\small 10-26t+14t}^{2}{\small +16t}^{3}$ \\
$\;\;{\small -21t}^{4}{\small +7t}^{5}{\small +t}^{6}{\small -t}^{7}$%
\end{tabular}%
$} & \multicolumn{1}{|l|}{$%
\begin{tabular}{l}
${\small 1-3t+4t}^{2}{\small -7t}^{4}$ \\
$\;\;{\small +5t}^{5}{\small +t}^{6}{\small -t}^{7}$%
\end{tabular}%
$} & \multicolumn{1}{|l|}{$%
\begin{tabular}{l}
${\small 1-t-t}^{2}{\small -t}^{3}$ \\
$\;\;{\small +t}^{4}{\small +4t}^{5}{\small -3t}^{6}$%
\end{tabular}%
$} \\ \hline
\end{tabular}%
\end{equation*}

In comparison we have the following expansion coefficients for the ordinary
skew Hall-Littlewood function $P_{(4,4,3,2,1)/(4,3,1,0,0)}$:

\begin{equation*}
\begin{tabular}{|c|c|c|c|}
\hline\hline
$\lambda $ & $m_{\lambda }$ & $s_{\lambda }$ & $P_{\lambda }$ \\ \hline\hline
\multicolumn{1}{|l|}{${\small 4,2,0}$} & \multicolumn{1}{|l|}{${\small 1}$}
& \multicolumn{1}{|l|}{${\small 1}$} & \multicolumn{1}{|l|}{${\small 1}$} \\
\hline
\multicolumn{1}{|l|}{${\small 4,1,1}$} & \multicolumn{1}{|l|}{${\small 2-t-t}%
^{2}$} & \multicolumn{1}{|l|}{${\small 1-t-t}^{2}$} & \multicolumn{1}{|l|}{$%
{\small 1-t}^{2}$} \\ \hline
\multicolumn{1}{|l|}{${\small 3,3,0}$} & \multicolumn{1}{|l|}{${\small 2-t-t}%
^{2}$} & \multicolumn{1}{|l|}{${\small 1-t-t}^{2}$} & \multicolumn{1}{|l|}{$%
{\small 1-t}^{2}$} \\ \hline
\multicolumn{1}{|l|}{${\small 3,2,1}$} & \multicolumn{1}{|l|}{${\small %
7-7t-4t}^{2}{\small +4t}^{3}$} & \multicolumn{1}{|l|}{${\small 3-5t-2t}^{2}%
{\small +4t}^{3}$} & \multicolumn{1}{|l|}{${\small 3-2t-3t}^{2}{\small +2t}%
^{3}$} \\ \hline
\multicolumn{1}{|l|}{${\small 2,2,2}$} & \multicolumn{1}{|l|}{${\small %
12-15t-6t}^{2}{\small +11t}^{3}{\small -t}^{4}{\small -t}^{5}$} &
\multicolumn{1}{|l|}{${\small 1-3t+3t}^{3}{\small -t}^{4}{\small -t}^{5}$} &
\multicolumn{1}{|l|}{${\small 1-t}^{2}{\small -t}^{3}{\small +t}^{5}$} \\
\hline
\end{tabular}%
\end{equation*}
\end{example}

\begin{proposition}[inverse Kostka-Foulkes matrix]
Denote by $n^{k}$ the partition whose Young diagram consists of $k$ rows of
size $n$. Let $\lambda \in \mathcal{A}_{k,n}^{+}$ and $\tilde{\lambda}$ be
the partition with all $n$-parts removed. Then%
\begin{equation}
Q_{\lambda /d/n^{k}}(x_{1},\ldots ,x_{k};t)=Q_{\tilde{\lambda}}(x_{1},\ldots
,x_{k};t)
\end{equation}%
with $d=k-m_{n}(\lambda )$ and in particular we have that%
\begin{equation}
\langle \lambda |\boldsymbol{s}_{\mu ^{\prime }}|n^{k}\rangle =b_{\tilde{%
\lambda}}(t)K(t)_{\tilde{\lambda}\mu }^{-1}\;.
\end{equation}
\end{proposition}

\begin{proof}
Exploit the description of lattice configurations in terms of
non-intersecting paths using the correspondence shown in Figure \ref%
{fig:qbosonvertex} between vertex and paths configurations. According to our
assumptions there are $d=\sum_{i=1}^{n-1}m_{i}(\lambda )$ paths crossing the
boundary, none of which ends up in the $n$th lattice column. Therefore, we
must have for each cylindric tableau $T$ of shape $\lambda /d/n^{k}$ that $%
\Phi _{T}=\varphi _{\tilde{T}}$ with $\tilde{T}$ being the ordinary tableau
of shape $\tilde{\lambda}$ obtained by restricting $T$ to the fundamental
region. Thus, the first assertion follows. The second is then a simple
consequence of the definition (\ref{ncschur}) and the known transformation
matrix $M(Q,s)=b(t)K(t)^{-1}$; see \cite[III.6, Table on p241]{Macdonald}.
\end{proof}

\section{Cylindric Macdonald functions}

In this section we define cylindric analogues of the skew Macdonald
functions $Q_{\lambda /\mu }^{\prime }$ and $P_{\lambda /\mu }^{\prime }$;
see (\ref{cHLdef}) for their definition.

\begin{figure}[tbp]
\begin{equation*}
\includegraphics[scale=0.3]{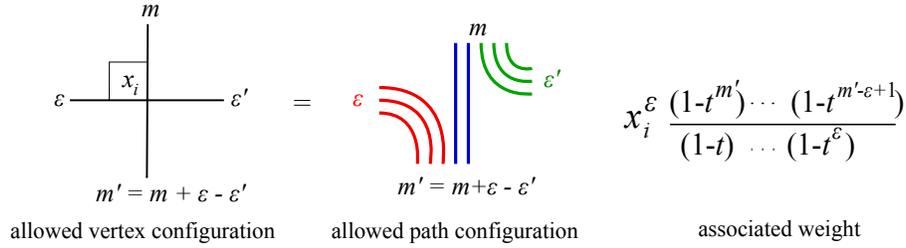}
\end{equation*}%
\caption{The vertex configurations and weights of the statistical model
associated with the second solution to the Yang-Baxter equation (\protect\ref%
{L'weight}). }
\label{fig:qdeserterweights}
\end{figure}

\subsection{The statistical vertex model associated with $L^{\prime }$}

We consider again a statistical vertex model defined on a cylinder but this
time we assume that the number of lattice rows of $\mathbb{L}$ lies in the
interval $1\leq \ell \leq n-1$. The setup and definitions remain the same as
previously with the following two exceptions:

\begin{enumerate}
\item A horizontal edge configuration is now a map $\gamma _{h}^{\prime }:%
\mathbb{E}_{h}\rightarrow \mathbb{Z}_{\geq 0}$, that is, vertical \emph{and}
horizontal edges can take values in the nonnegative integers. We shall label
values of horizontal edge with the Greek letter $\varepsilon $ and continue
to label vertical ones with the letter $m$. With this change a vertex
configuration $\gamma _{\langle i,j\rangle }^{\prime }$ is now a 4-tuple of
nonnegative integers, $\gamma _{\langle i,j\rangle }^{\prime }=\{\varepsilon
,m,\varepsilon ^{\prime },m^{\prime }\}$ giving the values of the W, N, E
and S edges centered at the interior lattice point $\langle i,j\rangle $,
respectively; see Figure \ref{fig:qdeserterweights}

\item We now fix the weights of the vertex configurations through the matrix
elements of the $L^{\prime }$-operator (\ref{L'}),%
\begin{equation}
\limfunc{wt}\nolimits^{\prime }(\gamma _{\langle i,j\rangle }^{\prime
}):=\langle \varepsilon ^{\prime },m^{\prime }|L^{\prime
}(x_{i})|\varepsilon ,m\rangle =u^{\varepsilon }\QATOPD[ ] {m^{\prime }}{%
\varepsilon }_{t}\delta _{m+\varepsilon ,m^{\prime }+\varepsilon ^{\prime
}}\;.  \label{L'weight}
\end{equation}%
As previously we then define the weight of a lattice configuration $\gamma
^{\prime }=(\gamma _{h}^{\prime },\gamma _{v}^{\prime })$ as the product of
vertex configurations over interior lattice points,
\begin{equation}
\limfunc{wt}\nolimits^{\prime }(\gamma ^{\prime })=\prod_{\langle i,j\rangle
\in \mathbb{\dot{L}}}\limfunc{wt}\nolimits^{\prime }(\gamma _{\langle
i,j\rangle }^{\prime })  \label{weight'prod}
\end{equation}%
and call a vertex or lattice configuration \textquotedblleft not
allowed\textquotedblright\ if the corresponding weight vanishes.
\end{enumerate}

\begin{remark}
The vertex configurations can again be interpreted in terms of
non-intersecting paths or \textquotedblleft infinitely-friendly
walkers\textquotedblright ; see Figure \ref{fig:qdeserterweights}. Note,
however, that this particular statistical model has not been previously
discussed in the literature.
\end{remark}

Shadowing closely our previous discussion of the $q$-boson model we consider
again for $\lambda ,\mu \in \mathcal{A}_{k,n}^{+}$ the set $\Gamma _{\lambda
,\mu }^{\prime }$ of periodic or cylindric lattice configurations, that is $%
\gamma _{h}^{\prime }(e_{i,0})=\gamma _{h}^{\prime }(e_{i,n})$ for each $%
\gamma ^{\prime }\in \Gamma _{\lambda ,\mu }^{\prime }$ where $e_{i,0}$ and $%
e_{i,n}$ are the horizontal edges starting at the points $\langle i,0\rangle
$ and $\langle i,n\rangle $, respectively. The values of the top and bottom
outer vertical edges are fixed in terms of the multiplicities $m_{j}(\mu )$
and $m_{j}(\lambda )$ as previously discussed and we set similar as before
\begin{equation*}
\Gamma _{\lambda /d/\mu }^{\prime }:=\{\gamma ^{\prime }\in \Gamma _{\lambda
,\mu }^{\prime }:\sum_{i=1}^{\ell }\gamma _{h}^{\prime }(\langle i,0\rangle
,\langle i,1\rangle )=d\}\ .
\end{equation*}%
That is, $\Gamma _{\lambda /d/\mu }^{\prime }$ is the subset of
configurations where the sum over the values of the left (or right) outer
horizontal edges is $d$. Define the partition function $Z_{\lambda /d/\mu
}^{\prime }(x_{1},\ldots ,x_{\ell }):=\sum_{\gamma \in \Gamma _{\lambda
/d/\mu }^{\prime }}\limfunc{wt}\nolimits^{\prime }(\gamma )$ then we have in
close analogy to our previous discussion the following identity.

\begin{lemma}
Let $\lambda ,\mu \in \mathcal{A}_{k,n}^{+}$. The following equality between
matrix element and partition functions is true,%
\begin{equation}
\langle \lambda |\boldsymbol{G}^{\prime }(x_{1})\cdots \boldsymbol{G}%
^{\prime }(x_{\ell })|\mu \rangle =\sum_{d\geq 0}z^{d}Z_{\lambda /d/\mu
}^{\prime }(x_{1},\ldots ,x_{\ell })\;.  \label{Z'G'}
\end{equation}%
Moreover, $Z_{\lambda /d/\mu }^{\prime }(x_{1},\ldots ,x_{\ell })$ is
symmetric in the $x_{i}$'s.
\end{lemma}

\begin{proof}
The assertion is again a direct consequence of the definitions (\ref{ncG'})
and (\ref{L'weight}). That the partition function is symmetric follows from $%
\boldsymbol{G}^{\prime }(x_{i})\boldsymbol{G}^{\prime }(x_{j})=\boldsymbol{G}%
^{\prime }(x_{j})\boldsymbol{G}^{\prime }(x_{i})$ which is a consequence of (%
\ref{ybe'}).
\end{proof}

Note that we can recover \textquotedblleft open boundary
conditions\textquotedblright , that is configurations on a finite strip with
the values of the outer horizontal edges all being zero, by setting $z=0$.

\begin{figure}[tbp]
\begin{equation*}
\includegraphics[scale=0.35]{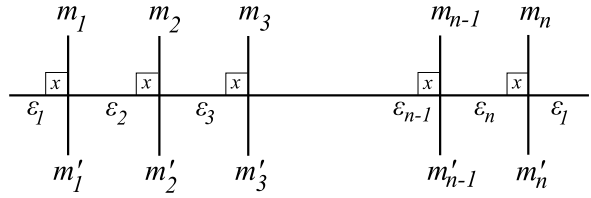}
\end{equation*}%
\caption{The numbering convention used in describing a row configuration of
the vertex model (\protect\ref{L'weight}). }
\label{fig:qdeserterrow}
\end{figure}

\subsection{Cylindric weight functions}

We now introduce the cylindric analogue of the coefficient function (\ref%
{phiprime}) which assigns a weight to each cylindric horizontal strip of the
form $\lambda ^{\prime }/d/\mu ^{\prime }$. Here $\lambda ^{\prime }/d/\mu
^{\prime }$ is the conjugate or transposed cylindric skew diagram of $%
\lambda /d/\mu $ defined earlier, that is we now consider the the case when $%
\lambda /d/\mu $ is a \emph{vertical} strip. For $\lambda ,\mu \in \mathcal{A%
}_{k,n}^{+}$ set%
\begin{equation}
\Phi _{\lambda ^{\prime }/d/\mu ^{\prime }}^{\prime }(t):=\prod_{j=1}^{n}%
\QATOPD[ ] {\mu _{j}^{\prime }-\mu _{j+1}^{\prime }}{\lambda ^{\prime
}[d]_{j+1}-\mu ^{\prime }[0]_{j+1}}_{t}  \label{Phiprime}
\end{equation}%
if $\lambda ^{\prime }/d/\mu ^{\prime }$ is a horizontal strip and zero
otherwise. To facilitate the comparison with (\ref{phiprime}) note that we
have%
\begin{equation}
\lambda ^{\prime }[d]_{1}-\mu ^{\prime }[0]_{1}=k+d-k=d-0=\lambda ^{\prime
}[d]_{n+1}-\mu ^{\prime }[0]_{n+1}
\end{equation}%
and, hence, the analogue of the additional factor $1/(t)_{\lambda
_{1}^{\prime }-\mu _{1}^{\prime }}$ appearing in (\ref{phiprime}) is
included in the definition (\ref{Phiprime}). Thus, the polynomial $\Phi
_{\lambda ^{\prime }/d/\mu ^{\prime }}^{\prime }$ can be interpreted as the
natural generalisation of $\varphi _{\lambda ^{\prime }/\mu ^{\prime
}}^{\prime }$ to cylindric skew diagrams. We also introduce the cylindric
counterpart for the second coefficient function (\ref{psiprime}) setting%
\begin{equation}
\Psi _{\lambda ^{\prime }/d/\mu ^{\prime }}^{\prime }(t)=\prod_{j=1}^{n}%
\QATOPD[ ] {\lambda _{j}^{\prime }-\lambda _{j+1}^{\prime }}{\lambda
^{\prime }[d]_{j}-\mu ^{\prime }[0]_{j}}_{t}\;.  \label{Psiprime}
\end{equation}%
The following equality shows that the two definitions (\ref{Phiprime}) and (%
\ref{Psiprime}) are consistent with each other.

\begin{lemma}
We have the identity
\begin{equation}
\Psi _{\lambda ^{\prime }/d/\mu ^{\prime }}^{\prime }(t)=\frac{b_{\lambda
}(t)}{b_{\mu }(t)}\Phi _{\lambda ^{\prime }/d/\mu ^{\prime }}^{\prime }(t)\;.
\label{bPsi'=bPhi'}
\end{equation}
\end{lemma}

\begin{proof}
A trivial rewriting using that $\lambda _{j}^{\prime }-\lambda
_{j+1}^{\prime }=\lambda ^{\prime }[0]_{j}-\lambda ^{\prime
}[0]_{j+1}=\lambda ^{\prime }[d]_{j}-\lambda ^{\prime }[d]_{j+1}$.
\end{proof}

\begin{definition}
Let $\lambda ,\mu \in \mathcal{A}_{k,n}^{+}$ and $d\geq 0,\;1\leq \ell \leq
n-1$. Then we define the (restricted) cylindric Macdonald functions $%
P_{\lambda ^{\prime }/d/\mu ^{\prime }}^{\prime }$ and $Q_{\lambda ^{\prime
}/d/\mu ^{\prime }}^{\prime }$ as follows,
\begin{equation}
P_{\lambda ^{\prime }/d/\mu ^{\prime }}^{\prime }(x_{1},\ldots ,x_{\ell
};t):=\sum_{|T|=\lambda ^{\prime }/d/\mu ^{\prime }}\Psi _{T}^{\prime
}(t)x^{T},  \label{cylP'}
\end{equation}%
and%
\begin{equation*}
Q_{\lambda ^{\prime }/d/\mu ^{\prime }}^{\prime }(x_{1},\ldots ,x_{\ell
};t):=\frac{b_{\mu }(t)}{b_{\lambda }(t)}~P_{\lambda ^{\prime }/d/\mu
^{\prime }}^{\prime }(x_{1},\ldots ,x_{\ell };t),
\end{equation*}%
where the sum runs over all conjugate cylindric skew tableaux $T$ of shape $%
\lambda ^{\prime }/d/\mu ^{\prime }$.
\end{definition}

The following proposition ties the transfer matrix (\ref{ncG'}) of the
second statistical model to cylindric vertical strips $\lambda /d/\mu $.

\begin{proposition}[cylindric vertical strips]
Let $\boldsymbol{G}^{\prime }(u):=\sum_{r\geq 0}u^{r}\boldsymbol{g}%
_{r}^{\prime }$ be the formal partial trace of the monodromy introduced
above. The latter is well-defined as an operator in $\limfunc{End}(\mathcal{F%
}^{\otimes n})$ since when acting on $|\mu \rangle ,~\mu \in \mathcal{A}%
_{k,n}^{+}$ for any $k\geq 0$ only a finite number of coefficients act
non-trivially. Namely, one has
\begin{equation}
\boldsymbol{g}_{r}^{\prime }|\mu \rangle =\sum_{\substack{ \lambda /d/\mu
=(1^{r}),  \\ \lambda \in \mathcal{A}_{k,n}^{+}}}z^{d}\Psi _{\lambda
^{\prime }/d/\mu ^{\prime }}^{\prime }(t)|\lambda \rangle \;,  \label{ncG'2}
\end{equation}%
where the second sum runs over all cylindric vertical strips of length $r$
with $0\leq d\leq \min (r,m_{n}(\mu ))$.
\end{proposition}

\begin{remark}
The operators $\boldsymbol{g}_{r}^{\prime }$ are closely related to the
discrete Laplacians considered in \cite{vanDiejen} and the functional
relation (\ref{TQ}) can be seen as a generalisation of Baxter's famous $TQ$%
-relation for the six-vertex or XXZ model with $\boldsymbol{E}$
corresponding to the transfer matrix $T$ and $\boldsymbol{G}^{\prime }$ to
Baxter's $Q$-operator. In fact, (\ref{ncG'}) is a generalisation of the XXZ $%
Q$-operator construction given in \cite{KorffQ} for \textquotedblleft
infinite spin\textquotedblright .
\end{remark}

\begin{proof}
Let $\mu \in \mathcal{A}_{k,n}^{+}$. According to the definitions (\ref{ncG'}%
) and (\ref{ncm}) the matrix elements of $\boldsymbol{g}_{r}^{\prime }$ give
the sum over all allowed row configurations. Fix an allowed horizontal edge
configuration $(\varepsilon _{1},\varepsilon _{2},\ldots ,\varepsilon
_{n},\varepsilon _{n+1})\in \mathbb{Z}_{\geq 0}^{n+1}$, that is $m_{j}(\mu
)\geq \varepsilon _{j+1}$ and $\varepsilon _{n+1}=\varepsilon _{1}$; see
Figure \ref{fig:qdeserterrow} for an illustration and note that we must have
in particular $\varepsilon _{1}\leq m_{n}$. Denote by $\lambda $ the
partition in $\mathcal{A}_{k,n}^{+}$ corresponding to the resulting
multiplicities $m_{j}^{\prime }:=m_{j}+\varepsilon _{j}-\varepsilon _{j+1}$
and by $\lambda ^{\prime }[\varepsilon _{1}]$ the associated conjugate
cylindric loop with $\lambda ^{\prime }[\varepsilon _{1}]_{n}=m_{n}^{\prime
}+\varepsilon _{1}=m_{n}+\varepsilon _{n}$. Then we have the following
relations between conjugate cylindric loops and horizontal edge values,%
\begin{equation}
m_{j}^{\prime }=\lambda ^{\prime }[\varepsilon _{1}]_{j}-\lambda ^{\prime
}[\varepsilon _{1}]_{j+1}\quad \text{and\quad }\varepsilon _{j}=\lambda
^{\prime }[\varepsilon _{1}]_{j}-\mu ^{\prime }[0]_{j},\quad j=1,...,n~.
\label{bijectionG'}
\end{equation}%
Furthermore, it follows from $\varepsilon _{1}\leq m_{n}$ that the conjugate
cylindric skew diagram $\lambda ^{\prime }/\varepsilon _{1}/\mu ^{\prime }$
is a horizontal strip of length $r=\sum_{j=1}^{n}\varepsilon _{j}$ in the
fundamental region, where $\varepsilon _{j}$ boxes are added in the $j$th
row of the Young diagram of $\mu ^{\prime }$. The associated weight of this
row configuration is according to (\ref{L'weight}) given by
\begin{equation}
x_{i}^{\varepsilon _{1}+\cdots +\varepsilon _{n}}\prod_{j=1}^{n}\QATOPD[ ] {%
\lambda _{j}^{\prime }-\lambda _{j+1}^{\prime }}{\lambda ^{\prime
}[\varepsilon _{1}]_{j}-\mu ^{\prime }[0]_{j}}_{t}=x_{i}^{\varepsilon
_{1}+\cdots +\varepsilon _{n}}\Psi _{\lambda ^{\prime }/\varepsilon _{1}/\mu
^{\prime }}^{\prime }
\end{equation}%
which for configurations with $\varepsilon _{1}=0$ coincides with $%
x_{i}^{\varepsilon _{2}+\cdots +\varepsilon _{n}}\psi _{\lambda ^{\prime
}/\mu ^{\prime }}^{\prime }(t)$; see the definition in (\ref{psiprime}).
Conversely, given a conjugate cylindric diagram $\lambda ^{\prime }/d/\mu
^{\prime }$ which is a horizontal strip the lattice row configuration can be
reconstructed in a unique way using the formulae in (\ref{bijectionG'}) with
$\varepsilon _{1}=d$.
\end{proof}

In analogy with Theorem \ref{qboson_Z} we now have the following result:

\begin{theorem}
Let $\lambda ,\mu \in \mathcal{A}_{k,n}^{+}$. The cylindric skew Macdonald
functions $P_{\lambda ^{\prime }/d/\mu ^{\prime }}^{\prime }$, $Q_{\lambda
^{\prime }/d/\mu ^{\prime }}^{\prime }$ are symmetric in the variables $%
x=(x_{1},\ldots ,x_{\ell })$ and one has the expansion
\begin{equation}
\langle \lambda |\boldsymbol{G}^{\prime }(x_{1})\cdots \boldsymbol{G}%
^{\prime }(x_{\ell })|\mu \rangle =\sum_{d\geq 0}z^{d}P_{\lambda ^{\prime
}/d/\mu ^{\prime }}^{\prime }(x_{1},\ldots ,x_{\ell };t)
\end{equation}%
which is equivalent to%
\begin{equation}
Z_{\lambda /d/\mu }^{\prime }(x_{1},\ldots ,x_{\ell })=P_{\lambda ^{\prime
}/d/\mu ^{\prime }}^{\prime }(x_{1},\ldots ,x_{\ell };t)\;.  \label{Z'=P'}
\end{equation}%
Here $1\leq \ell \leq n-1$ as before.
\end{theorem}

\begin{remark}
Setting $z=0$ the above result specialises to the identity
\begin{equation}
Z_{\lambda /0/\mu }^{\prime }(x_{1},\ldots ,x_{\ell })=P_{\lambda ^{\prime
}/\mu ^{\prime }}^{\prime }(x_{1},\ldots ,x_{\ell };t),
\end{equation}%
where $P_{\lambda ^{\prime }/\mu ^{\prime }}^{\prime }$ is the ordinary
(non-cylindric) Macdonald function defined in (\ref{cHLdef}). In particular,
$Z_{\lambda /0/\mu }^{\prime }(x_{1},\ldots ,x_{\ell })=0$ unless $\mu
\subset \lambda $.
\end{remark}

\begin{proof}
Set $\ell =1$ then the statement is an immediate consequence of the previous
proposition. The case $\ell >1$ now trivially follows, since each conjugate
cylindric tableau $\lambda ^{\prime }/d/\mu ^{\prime }$ can be written as a
sequence over (conjugate) cylindric horizontal strips: given an allowed
configuration $\gamma ^{\prime }$ of the entire lattice, denote by $%
m^{(i)}=(m_{1}^{(i)},\ldots ,m_{n}^{(i)})$ the upper vertical, and by $%
\varepsilon ^{(i)}=(\varepsilon _{1}^{(i)},\ldots ,\varepsilon _{n}^{(i)})$
the horizontal edge values in the $i^{\text{th}}$ lattice row. Set $%
m^{(0)}=m(\mu )$, $m^{(\ell )}=m(\lambda )$ and denote by $\lambda ^{(i)}$
the partition whose Young diagram has $m_{j}^{(i)}$ columns of length $j$.
Then $\lambda ^{(i)}[\sum_{l=1}^{i}\varepsilon _{1}^{(i)}]$ is the
associated conjugate cylindric loop, that is in the fundamental region we
add to the partition $\lambda ^{(i)}$ the total of $\sum_{l=1}^{i}%
\varepsilon _{1}^{(i)}$ columns of height $n$. Thus, each row adds a
horizontal strip of length $r_{i}=|\varepsilon ^{(i)}|$ to the Young diagram
of $\mu $ in the fundamental region. The sequence%
\begin{equation*}
(\lambda ^{(0)}[0],\lambda ^{(1)}[\varepsilon _{1}^{(1)}],\lambda
^{(2)}[\varepsilon _{1}^{(1)}+\varepsilon _{1}^{(2)}],\ldots ,\lambda
^{(\ell -1)}[\varepsilon _{1}^{(1)}+\cdots +\varepsilon _{1}^{(\ell
-1)}],\lambda ^{(\ell )}[d])
\end{equation*}%
with $\mu ^{\prime }[0]=\lambda ^{(0)}[0]$, $\lambda ^{(\ell )}[d]=\lambda
^{\prime }[d],$ $d=\sum_{i=1}^{\ell }\varepsilon _{1}^{(i)}$ results in a
unique conjugate cylindric skew tableaux $T$ of shape $\lambda ^{\prime
}/d/\mu ^{\prime }$, because in each row we have the constraint $\varepsilon
_{1}^{(i+1)}\leq m_{n}^{(i)}$ and (\ref{bijectionG'}) must hold true for an
allowed configuration. Moreover, the weight of the tableau is given as
product over the row weights, see (\ref{weight'prod}), and we recall that $%
\Psi _{T}^{\prime }:=\prod_{i=0}^{n-1}\Psi _{\lambda ^{(i+1)}/\varepsilon
_{1}^{(i+1)}/\lambda ^{(i)}}^{\prime }$.
\end{proof}

\begin{figure}[tbp]
\begin{equation*}
\includegraphics[scale=0.72]{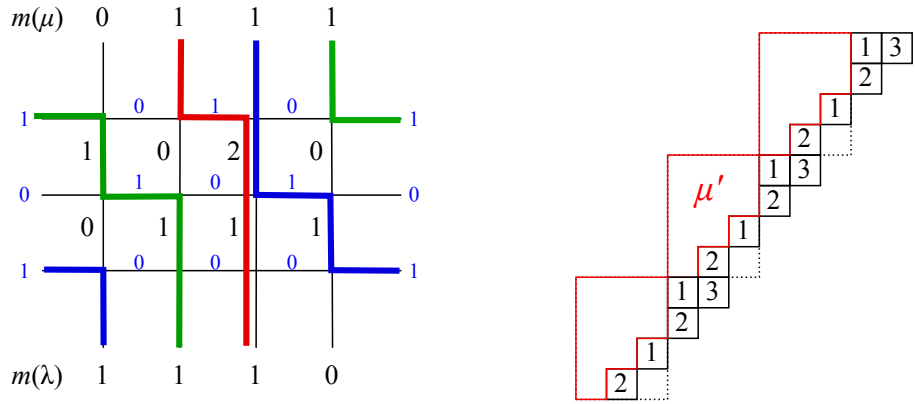}
\end{equation*}%
\caption{Example for the bijection between the lattice configurations of the
vertex model (\protect\ref{L'weight}) and cylindric skew tableaux.}
\label{fig:3rdbijection}
\end{figure}

\begin{corollary}[Expansions of cylindric Macdonald functions]
\label{ccHLexpansion}We have the following expansions of the cylindric
functions (\ref{cylP'})%
\begin{eqnarray}
P_{\lambda ^{\prime }/d/\mu ^{\prime }}^{\prime }(x_1,\ldots,x_{n-1};t) &=&\sum_{\nu \in
\mathcal{\tilde{A}}_{k,n}^{+}}\langle \lambda |\boldsymbol{g}_{\nu ^{\prime
}}^{\prime }|\mu \rangle m_{\nu ^{\prime }}(x_1,\ldots,x_{n-1})  \label{cylP'2m} \\
&=&\sum_{\nu \in \mathcal{\tilde{A}}_{k,n}^{+}}\langle \lambda |\boldsymbol{S%
}_{\nu ^{\prime }}^{\prime }|\mu \rangle s_{\nu ^{\prime }}(x_1,\ldots,x_{n-1}),\quad
\label{cylP'2s} \\
&=&\sum_{\nu \in \mathcal{\tilde{A}}_{k,n}^{+}}\langle \lambda |\boldsymbol{Q%
}_{\nu ^{\prime }}^{\prime }|\mu \rangle P_{\nu ^{\prime }}^{\prime
}(x_1,\ldots,x_{n-1};t),\quad  \label{cylP'2P'}
\end{eqnarray}%
where $m$ are the monomial symmetric functions, $s$ the Schur functions and $%
P^{\prime }$ the Macdonald functions defined in (\ref{cHLdef}). The
expansion coefficients obey the relations%
\begin{equation}
\langle \lambda |\boldsymbol{g}_{\nu ^{\prime }}^{\prime }|\mu \rangle
=\sum_{\sigma \in \mathcal{\tilde{A}}_{k,n}^{+}}\langle \lambda |\boldsymbol{%
S}_{\sigma ^{\prime }}^{\prime }|\mu \rangle K_{\sigma \nu }=\sum_{\substack{
|T|=\lambda ^{\prime }/d/\mu ^{\prime }  \\ \limfunc{wt}(T)=\nu ^{\prime }}}%
\Psi _{T}^{\prime }(t)  \label{g'matrix}
\end{equation}%
with $K=M(s,m)$ being the Kostka-matrix \cite{Macdonald} and%
\begin{equation}
\langle \lambda |\boldsymbol{S}_{\nu ^{\prime }}^{\prime }|\mu \rangle
=\sum_{\sigma \in \mathcal{\tilde{A}}_{k,n}^{+}}K_{\nu \sigma }(t)\langle
\lambda |\boldsymbol{Q}_{\sigma ^{\prime }}^{\prime }|\mu \rangle
=\sum_{w\in \mathfrak{S}_{\ell }}\varepsilon (w)\sum_{\substack{ |T|=\lambda
^{\prime }/d/\mu ^{\prime }  \\ \limfunc{wt}T=\nu ^{\prime }(w)}}\Psi
_{T}^{\prime }(t),  \label{S'matrix}
\end{equation}%
where $K(t)=M(s,P)$, $\ell =\ell (\nu ^{\prime })$ and $\nu ^{\prime
}(w):=(\nu _{1}^{\prime }-1+w_{1},\ldots ,\nu _{\ell }^{\prime }-\ell
+w_{\ell })$.
\end{corollary}

\begin{proof}
The first three equalities, (\ref{cylP'2m}), (\ref{cylP'2s}) and (\ref%
{cylP'2P'}) are direct consequences of the noncommutative Cauchy expansions (%
\ref{ncC2}). Exploiting the known transition matrix from the basis of Schur
functions to the basis of monomial symmetric functions the asserted relation
(\ref{g'matrix}) now follows from (\ref{Z'=P'}). The last relation, Equation
(\ref{S'matrix}), is then a direct consequence of the definition (\ref%
{ncQ'S'}) of $\boldsymbol{S}_{\nu ^{\prime }}^{\prime }$,
\begin{equation*}
\langle \lambda |\boldsymbol{S}_{\nu ^{\prime }}^{\prime }|\mu \rangle
=\sum_{w\in \mathfrak{S}_{\ell }}\varepsilon (w)\langle \lambda |\boldsymbol{%
g}_{\nu _{\ell }^{\prime }-\ell +w_{\ell }}^{\prime }\cdots \boldsymbol{g}%
_{\nu _{1}^{\prime }-1+w_{1}}^{\prime }|\mu \rangle ,
\end{equation*}%
and the fact that $P_{\lambda ^{\prime }}^{\prime }=\sum_{\mu }s_{\mu
^{\prime }}K_{\mu \lambda }(t)$; see \cite{Macdonald}.
\end{proof}

\begin{example}
Let $n=k=5$ and consider $P_{(5,3,2,1,1)/1/(5,3,1,0,0)}^{\prime }$. Then we
find the expansion coefficients in the $m$ and $s$-bases displayed in the
table below:

\begin{equation*}
\begin{tabular}{|c|c|c|}
\hline\hline
$\lambda $ & $m_{\lambda }$ & $s_{\lambda }$ \\ \hline\hline
\multicolumn{1}{|l|}{$4,2,2,0$} & \multicolumn{1}{|l|}{$1+t$} &
\multicolumn{1}{|l|}{$1+t$} \\ \hline
\multicolumn{1}{|l|}{$4,2,1,1$} & \multicolumn{1}{|l|}{$2+3t+t^{2}$} &
\multicolumn{1}{|l|}{$1+2t+t^{2}$} \\ \hline
\multicolumn{1}{|l|}{$3,3,2,0$} & \multicolumn{1}{|l|}{$2+3t+t^{2}$} &
\multicolumn{1}{|l|}{$1+2t+t^{2}$} \\ \hline
\multicolumn{1}{|l|}{$3,3,1,1$} & \multicolumn{1}{|l|}{$4+8t+5t^{2}+t^{3}$}
& \multicolumn{1}{|l|}{$1+3t+3t^{2}+t^{3}$} \\ \hline
\multicolumn{1}{|l|}{$3,2,2,1$} & \multicolumn{1}{|l|}{$%
11+22t+16t^{2}+4t^{3} $} & \multicolumn{1}{|l|}{$3+8t+9t^{2}+3t^{3}$} \\
\hline
\multicolumn{1}{|l|}{$2,2,2,2$} & \multicolumn{1}{|l|}{$%
24+52t+45t^{2}+16t^{3}+t^{4}$} & \multicolumn{1}{|l|}{$%
1+4t+6t^{2}+5t^{3}+t^{4}$} \\ \hline
\end{tabular}%
\end{equation*}%
N.B. in these expansions, as in the definition of $P_{\lambda ^{\prime
}/d/\mu ^{\prime }}^{\prime }$, we have assumed the number of variables to
be at most $n-1$.
\end{example}

\subsection{Kostka-Foulkes polynomials}

As a special case one can also recover the non-skew $P^{\prime },Q^{\prime }$%
-functions from the cylindric functions.

\begin{proposition}[Kostka-Foulkes polynomials]
\label{KFpolys}Let $\lambda \in \mathcal{A}_{k,n}^{+}$ and $\tilde{\lambda}$
the partition with all parts of size $n$ removed. Set $d=k-m_{n}(\lambda )$.
Then
\begin{equation}
P_{\lambda ^{\prime }/d/k^{n}}^{\prime }(x;t)=P_{\tilde{\lambda}^{\prime
}}^{\prime }(x;t)\;,
\end{equation}%
where $k^{n}$ is the partition with $n$ parts of size $k$. In particular,
for any $\nu \in \mathcal{\tilde{A}}_{k,n}^{+}$ we have the following
equalities%
\begin{equation}
\langle \lambda |\boldsymbol{g}_{\nu ^{\prime }}^{\prime }|n^{k}\rangle
=\sum_{\mu }K_{\mu \tilde{\lambda}}(t)K_{\mu ^{\prime }\nu ^{\prime }}=\sum
_{\substack{ |T|=\tilde{\lambda}^{\prime }  \\ \limfunc{wt}(T)=\nu ^{\prime
} }}\psi _{T}^{\prime }(t)
\end{equation}%
and%
\begin{equation}
K_{\nu \tilde{\lambda}}(t)=\langle \lambda |\boldsymbol{S}_{\nu ^{\prime
}}^{\prime }|n^{k}\rangle =\sum_{w}\varepsilon (w)\sum_{\substack{ |T|=%
\tilde{\lambda}^{\prime }  \\ \limfunc{wt}T=\nu ^{\prime }(w)}}\psi
_{T}^{\prime }(t)\;,  \label{KF_S'}
\end{equation}%
where $K_{\nu \tilde{\lambda}}(t)$ is the celebrated Kostka-Foulkes
polynomial and $\nu ^{\prime }(w)=(\nu _{1}^{\prime }-1+w_{1},\ldots ,\nu
_{\ell }^{\prime }-\ell +w_{\ell })$.
\end{proposition}

\begin{proof}
Because of (\ref{Z'=P'}) we have%
\begin{equation*}
P_{\lambda ^{\prime }/d/k^{n}}^{\prime }(x)=\left. \frac{d^{d}}{dz^{d}}%
\langle \lambda |\boldsymbol{G}^{\prime }(x_{1})\cdots \boldsymbol{G}%
^{\prime }(x_{\ell })|n^{k}\rangle \right\vert _{z=0}\;.
\end{equation*}%
We use that each lattice configuration corresponds to a configuration of
non-intersecting paths on the lattice; see Figure \ref{fig:qdeserterweights}%
. All paths start at the rightmost, the $n$th, lattice column and there must
be $d$ paths crossing the boundary. None of these paths ends up in the $n$th
lattice column, since $\sum_{i=1}^{n-1}m_{i}(\lambda )=d$ (and the paths
cannot backtrack on themselves), hence one easily verifies that $\Psi
_{T}^{\prime }=\psi _{\tilde{T}}^{\prime }$. Here $T$ is the cylindric
tableau of shape $\lambda ^{\prime }/d/k^{n}$ and $\tilde{T}$ the tableau of
shape $\tilde{\lambda}^{\prime }$ which is obtained by only considering the
squares of $\lambda ^{\prime }/d/k^{n}$ in the fundamental region. It is now
obvious from (\ref{Psiprime}) and (\ref{psiprime}) that $\limfunc{wt}(T)=%
\limfunc{wt}(\tilde{T})$. This proves the first assertion.

Employing the known expansion $P_{\nu ^{\prime }}^{\prime }=\sum_{\lambda
}s_{\lambda ^{\prime }}K_{\lambda \nu }(t)$ where $K_{\lambda \nu }(t)$ are
the (ordinary) Kostka-Foulkes polynomials \cite{Macdonald}, it follows from
our earlier result that%
\begin{equation*}
P_{\lambda ^{\prime }/d/k^{n}}^{\prime }(x;t)=\sum_{\lambda }s_{\lambda
^{\prime }}(x)K_{\lambda \tilde{\nu}}(t)\;.
\end{equation*}%
Comparing this with (\ref{g'matrix}) and (\ref{S'matrix}) proves the
remaining assertions.
\end{proof}

The expression (\ref{KF_S'}) for Kostka-Foulkes polynomials is of
determinant type and not manifestly positive unlike other expressions \cite{LS} \cite{KR}, see
also \cite[III.6, eqn (6.5) on p242 and Ex. 7 on p245]{Macdonald}. The present formula resembles instead Lusztig's $t$%
-deformed weight multiplicity formula \cite{Lusztig2}%
\begin{equation*}
K_{\lambda \mu }(t)=\sum_{w\in \mathfrak{S}_{k}}(-1)^{|w|}\mathcal{P}%
_{t}(w(\lambda +\rho )-(\mu +\rho )),
\end{equation*}%
where $\rho =(k,k-1,\ldots ,1)$ is the Weyl vector and the weight
multiplicity is given through the $t$-analogue of Konstant's partition
function%
\begin{equation*}
\prod_{\alpha >0}\frac{1}{1-te^{\alpha }}=\sum_{\mu \in \mathbb{Z}^{k}}%
\mathcal{P}_{t}(\mu )e^{\mu }\;.
\end{equation*}%
That is, the coefficient of $t^{m}$ in $\mathcal{P}_{t}(\mu )$ is the number
of ways the weight $\mu $ can be expressed as a sum of $m$ positive roots.

\begin{example}
Set $\lambda =(3,3,2,0)$ and $\mu =(2,2,2,2)$. Then we find for $K_{\lambda
\mu }(t)$ the following non-vanishing summands employing formula (\ref{KF_S'}%
),%
\begin{equation*}
\begin{tabular}{|c|c|c|c|}
\hline\hline
$\varepsilon (w)$ & $w$ & $\lambda ^{\prime }(w)$ & $\psi _{T}^{\prime }$ \\
\hline\hline
$+$ & 123 & $(3,3,2)$ & \multicolumn{1}{|l|}{$%
1+2t+3t^{2}+3t^{3}+2t^{4}+t^{5} $} \\
$-$ & 213 & $(4,2,2)$ & \multicolumn{1}{|l|}{$1+t+2t^{2}+t^{3}+t^{4}$} \\
$-$ & 132 & $(4,3,1)$ & \multicolumn{1}{|l|}{$1+t+t^{2}+t^{3}$} \\
$+$ & 231 & $(4,4,0)$ & \multicolumn{1}{|l|}{$1$} \\ \hline
\end{tabular}%
\end{equation*}%
\noindent and, hence, that $K_{\lambda \mu }(t)=t^{3}+t^{4}+t^{5}$. In
comparison, using the $t$-analogue of Konstant's partition function one
arrives at (compare e.g. with \cite[Example 3.1]{DLT})%
\begin{equation*}
\begin{tabular}{|c|c|c|c|}
\hline\hline
$\varepsilon (w)$ & $w(\lambda +\rho )$ & $w(\lambda +\rho )-(\mu +\rho )$ &
$\mathcal{P}_{t}(w(\lambda +\rho )-(\mu +\rho ))$ \\ \hline\hline
$+$ & $(6,5,3,0)$ & $(1,1,0,-2)$ & \multicolumn{1}{|l|}{$%
t^{2}+3t^{3}+2t^{4}+t^{5}$} \\
$+$ & $(5,6,3,0)$ & $(0,2,0,-2)$ & \multicolumn{1}{|l|}{$t^{2}+t^{3}+t^{4}$}
\\
$-$ & $(6,3,5,0)$ & $(1,-1,2,-2)$ & \multicolumn{1}{|l|}{$t^{3}$} \\ \hline
\end{tabular}%
\end{equation*}%
reproducing the same result for the Kostka-Foulkes polynomial.
\end{example}

\section{A deformation of the Verlinde algebra}

In this section we will identify the expansion coefficients of the cylindric
Macdonald functions in Corollary \ref{ccHLexpansion}, i.e. the matrix
elements $\langle \nu |\boldsymbol{Q}_{\lambda ^{\prime }}^{\prime }|\mu
\rangle $ and $\langle \nu |\boldsymbol{S}_{\lambda ^{\prime }}^{\prime
}|\mu \rangle $, with the \emph{structure constants} of an algebra which is
a quotient of the spherical Hecke algebra. More precisely, this quotient
is the coordinate ring $\Bbbk \lbrack \boldsymbol{V}_{k,n}]$ of a discrete
(0-dimensional) affine variety $\boldsymbol{V}_{k,n}$ where $\Bbbk $ is the
field of Puiseux series in the indeterminate $t=q^{2}$ with complex
coefficients. In particular, $\Bbbk \lbrack \boldsymbol{V}_{k,n}]$ is
finite-dimensional and a commutative Frobenius algebra $\mathfrak{F}_{n,k}$.
The latter has a distinguished basis which we can identify with $%
\{Q_{\lambda ^{\prime }}^{\prime }:\lambda \in \mathcal{A}_{k,n}^{+}\}$ and
the coproduct $\Delta _{n,k}:\mathfrak{F}_{n,k}\rightarrow \mathfrak{F}%
_{n,k}\otimes \mathfrak{F}_{n,k}$ of the Frobenius algebra computed in this
basis yields the cylindric Macdonald function (\ref{cylP'}), $\Delta
_{n,k}P_{\lambda ^{\prime }}^{\prime }=\sum_{d,\mu ^{\prime }}P_{\lambda
^{\prime }/d/\mu ^{\prime }}^{\prime }\otimes P_{\mu ^{\prime }}^{\prime }$.

We arrive at this result by constructing a common eigenbasis for the
noncommutative Macdonald functions (\ref{ncschur}) and (\ref{ncQ'S'}). We
apply a technique known as algebraic Bethe ansatz which we discuss in the
next subsection. Each point in the affine variety $\boldsymbol{V}_{k,n}$
will determine an eigenvector and the latter corresponds to an idempotent of
$\mathfrak{F}_{n,k}$. This correspondence will allow us to identify $%
\mathfrak{F}_{n,k}$ with the subalgebra $\subset \limfunc{End}\mathcal{F}%
_{k}^{\otimes n}$ generated by the matrices $\boldsymbol{Q}_{\nu ^{\prime
}}^{(k)}:=\left( \langle \lambda |\boldsymbol{Q}_{\nu ^{\prime }}^{\prime
}|\mu \rangle \right) _{\lambda ,\mu \in \mathcal{A}_{k,n}^{+}}$. The latter
specialise for $t=0$ to the fusion matrices of the Verlinde algebra and,
hence, we will refer to $\mathfrak{F}_{n,k}$ as \emph{deformed} fusion ring
or Verlinde algebra of $\widehat{\mathfrak{sl}}(n)_{k}$.

\subsection{The Bethe ansatz}

The Bethe ansatz is a well-established technique in the physics literature
on exactly-solvable lattice and quantum integrable models \cite{Baxter}.
Here we will employ the so-called \emph{algebraic} Bethe ansatz also known
as the \emph{quantum inverse scattering method} (see e.g. \cite{KIB} and references therein) which uses the
commutation relations of the Yang-Baxter algebra (\ref{ybe}) to construct
eigenvectors, so-called \textquotedblleft Bethe vectors\textquotedblright ,
of the $q$-boson model transfer matrix (\ref{ncE}); compare with the
discussion in \cite{Bogoliubovetal}. Using Proposition \ref{TQ} one deduces
that these are also eigenvectors of the formal power series (\ref{ncG'}).
This allows one to recover the results in \cite{vanDiejen} where the
eigenvectors of certain \textquotedblleft discrete
Laplacians\textquotedblright\ for the quantum nonlinear Schr\"{o}dinger have
been computed using the so-called \emph{coordinate} Bethe ansatz.\medskip

\noindent Our starting point is the \emph{ansatz }that the eigenvectors of (\ref%
{ncE}) are of the algebraic form%
\begin{equation}
\mathfrak{b}(y):=B(y_{1}^{-1})\cdots B(y_{k}^{-1})|\emptyset \rangle
=\sum_{\lambda \in \mathcal{A}_{k,n}^{+}}Q_{\lambda }(y^{-1};t)|\lambda
\rangle ,  \label{Bethevec}
\end{equation}%
where $y=(y_{1},\ldots ,y_{k})$ are some \emph{dependent} invertible
variables, called \emph{Bethe roots}, which need to be determined. We
specify them below as points in a family of affine varieties depending on $t$%
, for the moment we treat them as formal variables. Using the commutation
relations (\ref{aba1}) and (\ref{aba2}) of the Yang-Baxter algebra, one
arrives at a set of nonlinear equations, called the \emph{Bethe ansatz
equations} and an expression of the eigenvalue for the transfer matrix (\ref%
{ncE}).

\begin{proposition}[Bogoliubov-Izergin-Kitanine]
(i) The Bethe vector $\mathfrak{b}(y)$ is an eigenvector of the transfer
matrix $\boldsymbol{E}$ provided the Bethe roots $y=(y_{1},\ldots ,y_{k})$
satisfy the following set of coupled nonlinear equations,%
\begin{equation}
y_{i}^{n}\prod_{j\not=i}\frac{y_{i}-y_{j}t}{y_{i}-y_{j}}=z\prod_{j\not=i}%
\frac{y_{i}t-y_{j}}{y_{i}-y_{j}}\ ,\qquad i=1,\ldots ,k~.  \label{bae}
\end{equation}%
(ii) Let $y$ be a solution to the Bethe ansatz equations (\ref{bae}). Then
one has the eigenvalue equation%
\begin{equation}
\boldsymbol{E}(u)\mathfrak{b}(y)=\left( \prod_{i=1}^{k}\frac{1-u~y_{i}t}{%
1-u~y_{i}}+zt^{k}u^{n}\prod_{i=1}^{k}\frac{1-ut^{-1}y_{i}}{1-u~y_{i}}\right)
\mathfrak{b}(y)\;.  \label{ncEspec}
\end{equation}
\end{proposition}

\begin{proof}
The asserted identities are derived via induction in $k$ exploiting the
commutation relations (\ref{aba0}), (\ref{aba1}), (\ref{aba2}) and (\ref%
{aba3}) of the Yang-Baxter algebra. This is a standard computation (see \cite%
{KIB} for a textbook reference) and we therefore omit it.
\end{proof}

The following result is not contained in \cite{Bogoliubovetal}.
\begin{corollary}
Under the same set of conditions (\ref{bae}) as above the Bethe vector $%
\mathfrak{b}(y)$ is also an eigenvector of the transfer matrix $\boldsymbol{G%
}^{\prime }$,%
\begin{equation}
\boldsymbol{G}^{\prime }(u)\mathfrak{b}(y)=\prod_{i=1}^{k}(1+u~y_{i})~%
\mathfrak{b}(y)\;.  \label{ncG'spec}
\end{equation}%
This, in particular implies that%
\begin{equation}
\boldsymbol{Q}_{\lambda ^{\prime }}^{\prime }\mathfrak{b}(y)=P_{\lambda
}(y;t)\mathfrak{b}(y),  \label{ncQ'spec}
\end{equation}%
where $\boldsymbol{Q}_{\lambda ^{\prime }}^{\prime }$ is the noncommutative
HL polynomial defined in (\ref{ncQ'S'}).
\end{corollary}
\begin{proof}
The assertion is a direct consequence of the identity (\ref{TQ}) we proved earlier. Employing
the expansion (\ref{TQ2}) one verifies that $\boldsymbol{g}'_r$ can be written as a polynomial
in the $\boldsymbol{e}_{i}$'s. For instance, we have $(1-t)\boldsymbol{g}'_1=\boldsymbol{e}_1$,
$(1-t^2)\boldsymbol{g}'_2=(1-t)^{-1}\boldsymbol{e}_1-\boldsymbol{e}_2$ etc. Hence, $\mathfrak{b}(y)$
must be an eigenvector of $\boldsymbol{G}'(u)$. The corresponding eigenvalue is determined by (\ref{TQ}), (\ref{ncEspec})
and $\boldsymbol{G}'(0)=1$. The eigenvalue of the noncommutative Macdonald polynomial $\boldsymbol{Q}'_\lambda$
follows from the definition (\ref{ncQ'S'}) employing the automorphism $\omega_t$ in (\ref{omega}).
\end{proof}

Having derived the formal conditions (\ref{bae}) on the Bethe roots for the
Bethe vector (\ref{Bethevec}) to be an eigenvector, we need to investigate
if solutions to the set of equations (\ref{bae}) do exist and if any, how
many and their properties.

\subsection{Completeness of the Bethe ansatz}

In this section we are going to prove for $t$ an indeterminate that the
Bethe vectors (\ref{Bethevec}) form eigenbasis; see Theorem \ref{complete}
below. The strategy of the proof is entirely different from the one used in
\cite{vanDiejen} where $t$ is assumed to be a real parameter in the interval
$[-1,1]$. Here we will introduce an algebraic variety and its coordinate
ring showing that the associated Hilbert polynomial is a constant, the
dimension of the $k$ particle subspace $\mathcal{F}_{k}^{\otimes n}$. In
other words, the solutions to (\ref{bae}) are a 0-dimensional variety, a set
of discrete points, with cardinality $|\mathcal{A}_{k,n}^{+}|=\dim \mathcal{F%
}_{k}^{\otimes n}$. In the next section we will then identify this
coordinate ring as a deformation of the Verlinde ring. We start with some
general observations concerning the possible solutions to (\ref{bae}).

\begin{description}
\item[\emph{Permutation invariance}.] First we note that the Bethe vectors (%
\ref{Bethevec}) are symmetric in the Bethe roots, $\mathfrak{b}(y)=\mathfrak{%
b}(wy)$ for any $w\in \mathfrak{S}_{k}$ because of (\ref{aba0}). We have
already exploited this earlier; see Corollary \ref{YBHL}. Thus, we can
identify each solution $y=(y_{1},\ldots ,y_{k})$ to the Bethe ansatz
equations (\ref{bae}) with the set $\{y_{1},\ldots ,y_{k}\}$ or its orbit
under the natural action of the symmetric group.

\item[\emph{Scale invariance}.] Henceforth we will assume $z^{\pm 1/2n}$ to
exist. The $z$-dependence of the Bethe equations (\ref{bae}) can be
eliminated by rescaling the Bethe roots as $y\rightarrow z^{-1/n}y$, hence,
without loss of generality we will repeatedly set $z=1$ as it simplifies the
computations. By the same token we note that there is a $\mathbb{Z}_{n}$%
-action which leaves the set of solutions invariant: for any $\ell \in
\mathbb{Z}_{n}$ and any solution $y=(y_{1},\ldots ,y_{k})$ the rescaled
solution $e^{2\pi i\ell /n}y=(e^{2\pi i\ell /n}y_{1},\ldots ,e^{2\pi i\ell
/n}y_{k})$ also satisfies the Bethe ansatz equations.

\item[\emph{Inversion and complex conjugation}.] Finally, if $%
y=(y_{1},\ldots ,y_{k})$ is a solution, then one easily checks that $%
y^{-1}=(y_{1}^{-1},\ldots ,y_{k}^{-1})$ is a solution for $z\rightarrow
z^{-1}$. If we make the following further assumptions on the indeterminates $%
t$ and $z$,%
\begin{equation}
\bar{t}=t\qquad \text{and\qquad }\bar{z}^{\pm \frac{1}{2n}}=z^{\mp \frac{1}{%
2n}}  \label{tzassumption}
\end{equation}%
then we find that also $\bar{y}=(\bar{y}_{1},\ldots ,\bar{y}_{k})$ is a
solution with $z\rightarrow z^{-1}$. This latter assumption is physically
motivated: it makes the quantum Hamiltonians (\ref{IoM}) Hermitian; see (\ref%
{IoMspec}) below.
\end{description}

We will now exploit the permutation invariance of the Bethe ansatz equations
(\ref{bae}) by reformulating them in terms of symmetric functions in the
Bethe roots. This will lead us to the definition of an algebraic variety and
its vanishing ideal in the ring of symmetric functions. Recall from the
definition (\ref{G}) that%
\begin{equation*}
g_{r}(y;0,t)=Q_{(r)}(y;0,t)=(1-t)\sum_{i=1}^{k}y_{i}^{r}\prod_{j\not=i}\frac{%
y_{i}-y_{j}t}{y_{i}-y_{j}}
\end{equation*}%
is a symmetric polynomial of degree $r$. For ease of notation, we will
henceforth denote $g_{r}(y;0,t)$ simply by $g_{r}(y;t)$.

\begin{lemma}
(i) The Bethe ansatz equations (\ref{bae}) are equivalent to the equations%
\begin{equation}
g_{n}(y;t)=z(1-t^{k})\quad \text{and\quad }%
g_{n+r}(y;t)+zt^{k}g_{r}(y;t^{-1})=0,\quad 0<r<k~,  \label{gbae}
\end{equation}%
(ii) Let $y$ be a solution to the Bethe ansatz equations (\ref{bae}). Then
one has the identities%
\begin{equation}
g_{n-r}(y_{1},\ldots ,y_{k};t)=zg_{r}(y_{1}^{-1},\ldots
,y_{k}^{-1};t),\qquad 0<r<n\;.  \label{baecc}
\end{equation}
\end{lemma}

\begin{proof}
To prove the first assertion (i) assume that $y=(y_{1},\ldots ,y_{k})$
satisfy the Bethe ansatz equations (\ref{bae}). Recall $\prod_{i=1}^{k}\frac{%
1-u~y_{i}t}{1-u~y_{i}}=\sum_{r\geq 0}g_{r}(y;t)u^{r}$ with $g_{0}(y;t)=1$.
Using the coproduct of the Yang-Baxter algebra one shows easily by induction
that $\boldsymbol{e}_{n}=z~1$ and that $\boldsymbol{e}_{r}=0$ for $r>n$. The
assertion now follows from (\ref{ncEspec}).

Assume now that $y=(y_{1},\ldots ,y_{k})$ are invertible and such that (\ref%
{gbae}) hold. Let $E(u)=\prod_{i=1}^{k}(1+u~y_{i})$ be the generating
function of the elementary symmetric functions $e_{r}(y)$ and $%
G(u)=E(-ut)/E(-u)$ be the generating function of the $g_{r}(y;t)$'s. Then $%
E(-u)G(u)=E(-ut)$ and, hence,%
\begin{equation*}
\sum_{a+b=r}(-1)^{a}e_{a}(y)g_{b}(y;t)=(-t)^{r}e_{r}(y)\;.
\end{equation*}%
Employing this relation one calculates (recall that $n>2$ and that $%
e_{r}(y)=0$ for $r>k$)%
\begin{multline*}
0=e_{1}(g_{n+k-1}(t)+zt^{k}g_{k-1}(t^{-1}))=g_{n+k}(t)+zt^{k}g_{k}(t^{-1})+z(-1)^{k-1}e_{k}
\\
+e_{2}(g_{n+k-2}(t)+zt^{k}g_{k-2}(t^{-1}))-e_{3}\left(
g_{n+k-3}(t)+zt^{k}g_{k-3}(t^{-1})\right) +\cdots \\
\cdots +(-1)^{k}e_{k}\left( g_{n}(t)+zt^{k}\right)
=g_{n+k}(t)+zt^{k}g_{k}(t^{-1})\;,
\end{multline*}%
where we have used in the last line (\ref{gbae}). Similarly one shows by
induction that (\ref{gbae}) imply $g_{n+r}(t)+zt^{k}g_{r}(t^{-1})=0$ for any
$r\geq k$. Set $f=f(u)$ to be the eigenvalue appearing in (\ref{ncEspec}).
It now follows that $f$ is polynomial in $u,$ since $%
g_{n+r}(t)+zt^{k}g_{r}(t^{-1})=0$ for all $r\geq k$ ensures that the formal
power series expansion of $f$ with respect to $u$ terminates after finitely
many terms. Thus, the residues of $f$ at $u=y_{i}^{-1}$ for $i=1,\ldots ,k$
have to vanish. These conditions are equivalent to (\ref{bae}).

The second assertion (ii) follows from a simple computation,
\begin{eqnarray*}
g_{n-r}(y;t) &=&(1-t)\sum_{i=1}^{k}y_{i}^{n-r}\prod_{j\not=i}\frac{%
y_{i}-y_{j}t}{y_{i}-y_{j}}=z(1-t)\sum_{i=1}^{k}y_{i}^{-r}\prod_{j\not=i}%
\frac{y_{i}t-y_{j}}{y_{i}-y_{j}} \\
&=&z(1-t)\sum_{i=1}^{k}y_{i}^{-r}\prod_{j\not=i}\frac{y_{i}^{-1}-y_{j}^{-1}t%
}{y_{i}^{-1}-y_{j}^{-1}}=zg_{r}(y_{1}^{-1},\ldots ,y_{k}^{-1};t)\;.
\end{eqnarray*}
\end{proof}

\begin{remark}
Using the relation (ii) the spectrum of the commuting quantum Hamiltonians (%
\ref{IoM}) of the $q$-boson model as%
\begin{equation}
H_{r}^{\pm }\mathfrak{b}(y)=-\frac{g_{r}(y_{1},\ldots ,y_{k};t)\pm
zg_{r}(y_{1}^{-1},\ldots ,y_{k}^{-1};t)}{2}\mathfrak{b}(y)\;.
\label{IoMspec}
\end{equation}%
Assuming (\ref{tzassumption}) implies that $z^{-1/2}H_{r}^{\pm }$ are
(anti-) Hermitian operators.
\end{remark}

Recall that the functions $g_{r}(t)=g_{r}(0,t)$ are elements in $\mathbb{Z}%
[t][e_{1},\ldots ,e_{k}]$ where the $e_{r}$ are the elementary symmetric
functions in $k$ variables; see (\ref{E}). Thus, we can interpret (\ref{gbae}%
) as a set of polynomial equations in the variables $\{e_{1},\ldots ,e_{k}\}$
with coefficients in $\mathbb{Z}[t]$.

\begin{example}
Employing that%
\begin{equation*}
g_{r}(t)=\sum_{s=0}^{r}h_{r-s}e_{s}(-t)^{s},\qquad h_{r}=\det (e_{1-i+j})
\end{equation*}%
the equations (\ref{gbae}) can be reformulated in terms of elementary
symmetric polynomials leading to a coupled set of algebraic equations. For
instance, set $n=3$ and $k=2$. From the Bethe ansatz equations (\ref{bae})
one easily derives that $e_{k}(y)^{n}=1$; this latter relation is true for general $n,k$. Thus,
in the present example with $k=2$ we only need an additional relation to
express $e_{1}$ in terms of $e_{k=2}$. Using the expansion of the equation $%
zg_{1}(y^{-1};t)=g_{2}(y;t)$ into elementary symmetric polynomials one
arrives for $n=3,\;k=2$ and $z=1$ at the quadratic equation%
\begin{equation*}
e_{1}^{2}-e_{1}e_{2}^{2}-(1+t)e_{2}=0
\end{equation*}%
Here we have used that $e_{r}(y^{-1})=e_{k}^{n-1}(y)e_{k-r}(y)$.
\end{example}

As this example shows we are led to solving polynomial equations with
coefficients in $\mathbb{Z}[t]$. In order to guarantee the existence of
solutions we must work over an algebraically closed field. This motivates us
to consider the algebraically closed field of Puiseux series,%
\begin{equation}
\mathbb{\Bbbk }=\mathbb{C}\{\!\{t\}\!\}:=\tbigcup_{m=1}^{\infty }\mathbb{C}%
((t^{1/m})),  \label{Puiseux}
\end{equation}%
which is the formal union of the fields of Laurent series in $t^{1/m}$. That
is, an element $f\in \mathbb{\Bbbk }$ is a formal expression of the form $%
f=\sum_{\ell \geq \ell _{0}}c_{\ell }t^{\ell /m}$ with $c_{\ell }\in \mathbb{%
C}$.

Fix $k>0$ and set $\tilde{g}_{n+r}(t)=g_{n+r}(t)+zt^{k}g_{r}(t^{-1})$ for $%
r=1,\ldots ,k-1$, $\tilde{g}_{n}(t)=g_{n}(t)-z(1-t^{k})$ and $\tilde{g}%
_{r}(t)=g_{r}(t)$ for $r\notin \{n,n+1,\ldots ,n+k-1\}$. In what follows we
suppress for simplicity the explicit dependence on $t$ in the notation and
set once more $z=1$. Denote by%
\begin{equation}
\boldsymbol{V}_{k,n}:=\left\{ \epsilon =(\epsilon _{1},\ldots ,\epsilon
_{k})\in \mathbb{\Bbbk }^{k}:\tilde{g}_{n}|_{e=\epsilon }=\cdots =\tilde{g}%
_{n+k-1}|_{e=\epsilon }=0\right\}  \label{BetheV}
\end{equation}%
the solutions to (\ref{gbae}) for $z=1$ in the affine space $\mathbb{\Bbbk }%
^{k}$, where $\tilde{g}_{r}|_{e=\epsilon }$ is obtained by replacing $%
e_{r}\rightarrow \epsilon _{r}$ in its polynomial expression $\tilde{g}%
_{n}=\sum_{\lambda }c_{\lambda }e_{\lambda _{1}}\cdots e_{\lambda _{\ell }}$.

\begin{lemma}
\label{nullstellen}Let $\boldsymbol{I}(\boldsymbol{V}_{k,n})\subset \mathbb{%
\Bbbk }[e_{1},\ldots ,e_{k}]$ be the vanishing ideal of the affine variety (%
\ref{BetheV}) and define the two-sided ideal $\boldsymbol{I}_{k,n}:=\langle
\tilde{g}_{n},\ldots ,\tilde{g}_{n+k-1}\rangle $. Then%
\begin{equation}
\boldsymbol{I}(\boldsymbol{V}_{k,n})=\boldsymbol{I}_{k,n}\;.
\end{equation}
\end{lemma}

In other words, given two symmetric functions $f,g\in \Bbbk \lbrack
e_{1},\ldots ,e_{k}]$ their difference $f-g$ lies in the ideal $\boldsymbol{I%
}_{k,n}$ if and only if $f(y)=g(y)$ for all solutions $y$ of the Bethe
ansatz equations (\ref{bae}). The proof of this statement follows a similar
strategy as the one given in \cite[Proof of Theorem 6.20, Claim 1]{KS} for
the Verlinde algebra.

\begin{proof}
We show that the ideal $\boldsymbol{I}_{k,n}$ is radical, i.e. $\boldsymbol{I%
}_{k,n}=\sqrt{\boldsymbol{I}_{k,n}}$, and the assertion then follows from
the strong Nullstellensatz. The results in \cite[(2.16), p. 213]{Macdonald}
imply that $\Lambda _{k}(t):=\mathbb{\Bbbk }[e_{1},\ldots ,e_{k}]\cong
\mathbb{\Bbbk }[g_{1},\ldots ,g_{k}]$ and in the projective limit $\Lambda
(t)=\lim\limits_{\longleftarrow }\Lambda _{k}(t)$ we have $\mathbb{\Bbbk }%
[g_{1},g_{2}\ldots ]\cong \mathbb{\Bbbk }[e_{1},e_{2},\ldots ]$. In
particular the $g_{r}$'s are all algebraically independent and, thus, the $%
\tilde{g}_{r}$ also form an algebraically independent set (note that $%
g_{n+r} $ and $g_{r}$ have different degree). Hence, the elements in $\{%
\tilde{g}_{\lambda }:=\tilde{g}_{\lambda _{1}}\tilde{g}_{\lambda _{2}}\cdots
\}_{\lambda },$ where $\lambda $ ranges over all partitions, are linearly
independent. Suppose $f=\sum_{\lambda }c_{\lambda }\tilde{g}_{\lambda }$
with $f\in \mathbb{\Bbbk }[e_{1},\ldots ,e_{k}]\subset \mathbb{\Bbbk }%
[e_{1},e_{2},\ldots ,]$ is not in $\boldsymbol{I}_{k,n}$, then there must
exist at least one partition $\mu $ such that $\mu _{j}\notin \{n,n+1,\ldots
,n+k-1\}$ for all $j$ and $c_{\mu }\neq 0$. We can therefore conclude that
the expansion of $f^{m}$, $m\geq 1$ contains $\tilde{g}_{\mu ^{m}}$ where $%
\mu ^{m}$ is the partition containing each part $\mu _{j}>0$ exactly $m$
times. Thus, $f^{m}\notin \boldsymbol{I}_{k,n}$ and projecting onto $\mathbb{%
\Bbbk }[e_{1},\ldots ,e_{k}]$ now yields the desired result.
\end{proof}

It will be important to compute which Hall-Littlewood functions lie in the
ideal $\boldsymbol{I}_{k,n}$. Recall the definition of the raising ($i<j$)
and lowering ($i>j$) operators $R_{ij}\lambda :=(\ldots ,\lambda
_{i}+1,\ldots ,\lambda _{j}-1,\ldots )$ and $R_{ij}f_{\lambda
}:=f_{R_{ij}\lambda }$. Then, in principle, one can calculate whether a
Hall-Littlewood $Q$-function lies in $\boldsymbol{I}_{k,n}$ by employing the
identity (\ref{Q2g}) \cite{Macdonald}. However, a simpler approach is given
by making use of the alternative definition of Hall-Littlewood functions via
the symmetric group. Namely, introduce another function $R_{\lambda
}=R_{\lambda }(x_{1},\ldots ,x_{k};t)$ via the equalities
\begin{equation}
(1-t)^{k}R_{\lambda }=P_{\lambda }\prod_{i=0}^{\lambda
_{1}}(t)_{m_{i}(\lambda )}\quad \quad \text{and}\quad \quad
(1-t)^{k}R_{\lambda }=(t)_{k-\ell (\lambda )}Q_{\lambda }\;.  \label{RPQ}
\end{equation}%
The function $R_{\lambda }$ for $\lambda $ a \emph{partition} can be
generalised to \emph{compositions }$\mu $ via \cite{Macdonald}%
\begin{equation}
R_{\mu }(x_{1},\ldots ,x_{k};t):=\sum_{w\in \mathfrak{S}_{k}}w\left( x^{\mu
}\tprod_{i<j}\frac{x_{i}-tx_{j}}{x_{i}-x_{j}}\right)  \label{HLR}
\end{equation}%
The latter are a linear combination of the functions $R_{\lambda }$ indexed
by \emph{partitions }$\lambda $ which is obtained by repeatedly applying the
following \textquotedblleft straightening rules\textquotedblright\ \cite%
{Macdonald}%
\begin{equation}
R_{\lambda .\sigma _{i}}=tR_{\lambda }-R_{(\ldots ,\lambda _{i}-1,\lambda
_{i+1}+1,\ldots )}+tR_{(\ldots ,\lambda _{i+1}+1,\lambda _{i}-1,\ldots
)},\qquad i=1,\ldots ,k-1  \label{straight1}
\end{equation}%
and%
\begin{equation}
R_{(\lambda _{1},\ldots ,\lambda _{k})}=0,\qquad \lambda _{k}<0\;.
\label{straight2}
\end{equation}%
The following lemma now shows that the Bethe ansatz equations (\ref{bae})
extend these straightening rules to the extended affine symmetric group $%
\mathfrak{\hat{S}}_{k}$, that is any $R_{\mu }$ with $\mu \in \mathcal{P}%
_{k} $ can be written as a linear combination of $R_{\lambda }$ with $%
\lambda \in \mathcal{A}_{k,n}^{+}$. We therefore can say that this
particular model exhibits an extended affine Weyl group invariance or
`symmetry'.

\begin{lemma}[extended affine straightening rules]
Let $\lambda \in \mathcal{A}_{k,n}^{+}$. Then the following polynomials are
in the ideal $\boldsymbol{I}_{k,n}$%
\begin{equation}
R_{\lambda .\sigma _{0}}-tR_{\lambda }+R_{(\lambda _{1}+1,\ldots ,\lambda
_{k}-1)}-tR_{(\lambda _{k}-1+n,\ldots ,\lambda _{1}+1-n)}  \label{straight3}
\end{equation}%
and%
\begin{equation}
R_{\lambda }-zR_{\lambda .\tau },\qquad \lambda _{1}\geq n,
\label{straight4}
\end{equation}%
where $\sigma _{0}$, $\tau $ are the additional generators of the extended
affine symmetric group $\mathfrak{\hat{S}}_{k}$\ whose right action on $%
\lambda $ is given by (\ref{affaction}).
\end{lemma}

\begin{proof}
Let $\theta (t):=\prod_{1\leq i<j\leq k}\frac{y_{i}-ty_{j}}{y_{i}-y_{j}}$
and $\sigma _{ij}:=\sigma _{j-1}\sigma _{j-2}\cdots \sigma _{i+1}\sigma _{i}$%
. From the Bethe ansatz equations
\begin{equation*}
1=zy_{1}^{-n}\prod_{j\neq 1}\frac{y_{1}t-y_{j}}{y_{1}-ty_{j}}%
=z^{-1}y_{k}^{n}\prod_{j\neq 1}\frac{y_{j}t-y_{k}}{y_{j}-ty_{k}}
\end{equation*}%
one derives the identity%
\begin{equation*}
(y_{1}-ty_{k})y^{\lambda .\sigma _{0}}\theta (t)=-\sigma _{1k}\left[
(y_{1}-ty_{k})y^{\lambda }\theta (t)\right]
\end{equation*}%
and the first affine straightening rule now follows from Lemma \ref{nullstellen}.

To prove the second rule we first note the relations%
\begin{eqnarray*}
\sigma _{i}\theta (t) &=&\frac{y_{i}t-y_{i+1}}{y_{i}-ty_{i+1}}\theta (t), \\
\sigma _{ij}\theta (t) &=&\frac{y_{i}t-y_{j}}{y_{i}-ty_{j}}~\frac{%
y_{i+1}t-y_{j}}{y_{i+1}-ty_{j}}~\cdots ~\frac{y_{j-1}t-y_{j}}{y_{j-1}-ty_{j}}%
\theta (t), \\
\sigma _{ij}^{-1}\theta (t) &=&\frac{y_{i}t-y_{j}}{y_{i}-ty_{j}}~\frac{%
y_{i}t-y_{j-1}}{y_{i}-ty_{j-1}}~\cdots ~\frac{y_{i}t-y_{i+1}}{y_{i}-ty_{i+1}}%
\theta (t)\;.
\end{eqnarray*}%
Hence, it follows from the Bethe ansatz equations that%
\begin{equation*}
\sigma _{ik}^{-1}\theta (t)=\left( \tprod_{i<j}\frac{ty_{i}-y_{j}}{%
y_{i}-ty_{j}}\right) \theta (t)\overset{\text{BAE}}{=}z^{-1}y_{i}^{n}\left(
\tprod_{j<i}\frac{y_{i}-ty_{j}}{ty_{i}-y_{j}}\right) \theta (t)=z^{-1}\sigma
_{1i}[y_{1}^{n}\theta (t)]\;.
\end{equation*}%
In particular, choosing $i=1$ we obtain the identity%
\begin{eqnarray*}
y_{1}^{n}\theta (t) &=&z\sigma _{1k}^{-1}\theta (t)=z\sigma _{1}\sigma
_{2}\cdots \sigma _{k-1}\theta (t), \\
&\Rightarrow &\;y^{\lambda }\theta (t)=z\sigma _{1}\sigma _{2}\cdots \sigma
_{k-1}\left( y_{1}^{\lambda _{2}}\cdots y_{k-1}^{\lambda _{k}}y_{k}^{\lambda
_{1}-n}\theta (t)\right)
\end{eqnarray*}%
which proves the second straightening rule applying once more Lemma \ref{nullstellen}.
\end{proof}

\begin{remark}
Note in particular that for $\lambda =(\lambda _{1},\ldots ,\lambda _{k})$
with $\lambda _{1}=n$ and $\lambda _{k}>0$ the polynomials%
\begin{equation}
P_{\lambda }-z^{m_{n}(\lambda )}P_{\tilde{\lambda}}\quad \quad \text{and}%
\quad \quad Q_{\lambda }-z^{m_{n}(\lambda )}(t)_{m_{n}(\lambda )}Q_{\tilde{%
\lambda}},  \label{PQreduction}
\end{equation}%
where $\tilde{\lambda}$ is the reduced partition with all parts $n$ removed
lie in the ideal $\boldsymbol{I}_{k,n}$.
\end{remark}

The relation between the extended affine straightening rules and the
definition of $\boldsymbol{I}_{k,n}$ is given by the following result.

\begin{lemma}
We have for $r>0$ that%
\begin{equation}
\frac{(1-t)^{k}}{(t)_{k-1}}~R_{(0,\ldots ,0,r)}=-t^{k}g_{r}(t^{-1})\;.
\label{greduction}
\end{equation}
\end{lemma}

\begin{proof}
Inserting the definition of the polynomial $R_{\mu }$ we obtain%
\begin{eqnarray*}
R_{(0,\ldots ,0,r)} &=&\sum_{w\in \mathfrak{S}_{k}}w\left(
x_{k}^{r}\tprod_{i<j}\frac{x_{i}-tx_{j}}{x_{i}-x_{j}}\right) \\
&=&\sum_{i=1}^{k}x_{i}^{r}\tprod_{j\neq i}\frac{x_{i}t-x_{j}}{x_{i}-x_{j}}%
\sum_{w\in \mathfrak{S}_{k-1}}w\left( \tprod_{i<j<k}\frac{x_{i}-tx_{j}}{%
x_{i}-x_{j}}\right) \\
&=&\frac{(t)_{k-1}}{(1-t)^{k-1}}\sum_{i=1}^{k}x_{i}^{r}\tprod_{j\neq i}\frac{%
x_{i}t-x_{j}}{x_{i}-x_{j}}=-\frac{(t)_{k-1}t^{k}}{(1-t)^{k}}g_{r}(t^{-1})\;,
\end{eqnarray*}%
where in the second equality in the first line the sum only runs over
permutations $w\in \mathfrak{S}_{k-1}$ which only permute the first $k-1$
indices.
\end{proof}

\begin{proposition}[Basis of the coordinate ring]
The quotient $\Bbbk \lbrack \boldsymbol{V}_{k,n}]:=\mathbb{\Bbbk }%
[e_{1},\ldots ,e_{k}]/\boldsymbol{I}_{k,n}$ viewed as a vector space has
dimension $|\mathcal{A}_{k,n}^{+}|=\binom{n-1}{k}$. A basis is given by the
equivalence classes of the Hall-Littlewood functions $P_{\tilde{\lambda}}$
with $\tilde{\lambda}\in \mathcal{\tilde{A}}_{k,n}^{+}$.
\end{proposition}

\begin{proof}
The proof is a generalisation of the one in \cite[Theorem 6.20, Claim 4]{KS}.
Recall that the Hall-Littlewood functions $\{P_{\lambda }:\lambda $
partition with $\ell (\lambda )\leq k\}$ form a $\mathbb{Z}[t]$-basis of $%
\mathbb{Z}[t][e_{1},\ldots ,e_{k}]$; this follows from \cite[(2.7), p. 209]%
{Macdonald} and projecting onto the ring of symmetric function with $k$%
-variables by setting $e_{r}=0$ for $r>k$. Denote by $[P_{\lambda
}]:=P_{\lambda }+\boldsymbol{I}_{k,n}$ the equivalence class of $P_{\lambda
} $ in $\Bbbk \lbrack \boldsymbol{V}_{k,n}]$. Then it follows from our
previous lemma that for any partition $\lambda \in \mathcal{P}^+_{k}$ there
exists a $\mu \in \mathfrak{S}_{k}\mathcal{\tilde{A}}_{k,n}^{+}$ such that $%
[P_{\lambda }]=[P_{\mu }]$. Using the (non-affine) straightening rules (\ref%
{straight1}) and (\ref{straight2}) $P_{\mu }$ can be written as $\mathbb{Z}%
[t]$-linear combination of $P_{\tilde{\nu}}$'s with $\tilde{\nu}\in \mathcal{%
\tilde{A}}_{k,n}^{+}$. Thus, we have that the vector space dimension of $%
\Bbbk \lbrack \boldsymbol{V}_{k,n}]$ must be smaller or equal than $|%
\mathcal{\tilde{A}}_{k,n}^{+}|=|\mathcal{A}_{k,n}^{+}|$.

Now assume that $0=\sum_{\tilde{\lambda}\in \mathcal{\tilde{A}}_{k,n}^{+}}c_{%
\tilde{\lambda}}[P_{\tilde{\lambda}}]$ which is equivalent to $\sum_{\tilde{%
\lambda}\in \mathcal{\tilde{A}}_{k,n}^{+}}c_{\tilde{\lambda}}P_{\tilde{%
\lambda}}=\sum_{r=n}^{n+k-1}f_{r}\tilde{g}_{r}$ for some $f_{r}\in \mathbb{%
\Bbbk }[e_{1},\ldots ,e_{k}]$. The transition matrix between the basis $%
\{P_{\lambda }\}$ and the basis $\{g_{\lambda }\}$ in the ring of symmetric
functions is strictly lower unitriangular with respect to the natural or
dominance partial ordering \cite[(2.16), p. 213]{Macdonald}, whence we can
conclude that the last equality can be written as $\sum_{\tilde{\lambda}\in
\mathcal{\tilde{A}}_{k,n}^{+}}c_{\tilde{\lambda}}g_{\tilde{\lambda}%
}=\sum_{\mu \notin \mathcal{\tilde{A}}_{k,n}^{+}}d_{\mu }\tilde{g}_{\mu }$
for some $d_{\mu }\in \mathbb{\Bbbk }$. Recall that the $\tilde{g}_{\lambda
} $'s are also linearly independent and that $\tilde{g}_{\tilde{\lambda}}=g_{%
\tilde{\lambda}}$ for $\tilde{\lambda}\in \mathcal{\tilde{A}}_{k,n}^{+}$.
Thus, $c_{\tilde{\lambda}}=0$ for all $\tilde{\lambda}\in \mathcal{\tilde{A}}%
_{k,n}^{+}$ which establishes linear independence of the $P_{\tilde{\lambda}%
} $ with $\tilde{\lambda}\in \mathcal{\tilde{A}}_{k,n}^{+}$ in $\Bbbk
\lbrack \boldsymbol{V}_{k,n}]$.
\end{proof}

We are now ready to prove the main statement of this section.\

\begin{theorem}[completeness of the Bethe ansatz]
\label{complete}Let $t$ be an indeterminate and set $z=1$. \textbf{(1)}
There is a bijection $\sigma \mapsto y_{\sigma }=(y_{1},...,y_{k})\in
\mathbb{\Bbbk }^{k}$ between partitions in $\mathcal{A}_{k,n}^{+}$ and
solutions\ to the Bethe ansatz equations (\ref{bae}). We denote the Bethe
vector (\ref{Bethevec}) corresponding to $\sigma \in \mathcal{A}_{k,n}^{+}$
by $\mathfrak{b}_{\sigma }:=\mathfrak{b}(y_{\sigma })$. \textbf{(2)} The
Bethe states $\{\mathfrak{b}_{\sigma }~|~\sigma \in \mathcal{A}_{k,n}^{+}\}$
provide an orthogonal basis of $\mathcal{F}_{k}^{\otimes n}\otimes _{\mathbb{%
C}(t)}\mathbb{\Bbbk }$, i.e. one has the identity%
\begin{equation*}
\langle \mathfrak{b}_{\rho }|\mathfrak{b}_{\sigma }\rangle =\sum_{\lambda
\in \mathcal{A}_{k,n}^{+}}Q_{\lambda }(y_{\rho };t)P_{\lambda }(\bar{y}%
_{\sigma };t)=0,
\end{equation*}%
where $\rho \neq \sigma $ are two distinct partitions in $\mathcal{A}%
_{k,n}^{+}$.
\end{theorem}

\begin{proof}
Without loss of generality we can set again $z=1$. It follows from our
previous lemma, $\boldsymbol{I}(\boldsymbol{V}_{k,n})=\boldsymbol{I}_{k,n}$,
that the dimension of the algebraic variety $\boldsymbol{V}_{k,n}$ equals
the degree of the (affine) Hilbert polynomial of $\boldsymbol{I}_{k,n}$; see
e.g. \cite[Definitions 5,7 and Theorem 8, Chapter 9, \S 3, pp. 459-461]{CLOS}%
. But since $\mathbb{\Bbbk }[e_{1},\ldots ,e_{k}]/\boldsymbol{I}_{k,n}=\Bbbk
\lbrack \boldsymbol{V}_{k,n}]$ has finite dimension as a vector space, the
affine Hilbert polynomial of $\boldsymbol{I}_{k,n}$ is of degree zero, i.e.
a constant and this constant equals the vector space dimension \cite[Chapter
9, \S 4, Ex 10, p.475]{CLOS} which we showed to be $|\mathcal{A}_{k,n}^{+}|$%
. Furthermore, $\boldsymbol{V}_{k,n}$ is non-empty. To see this assume first
that $\boldsymbol{V}_{k,n}=\varnothing $. Then the Hilbert polynomial of $%
\boldsymbol{I}(\boldsymbol{V}_{k,n})$ would be the zero polynomial \cite[p.
461]{CLOS}, but we have just seen that it is a nonzero constant. Hence, we
can conclude that $\boldsymbol{V}_{k,n}$ consists of finitely many points
and, furthermore, $|\boldsymbol{V}_{k,n}|$ equals the (constant) Hilbert
polynomial of $\Bbbk \lbrack \boldsymbol{V}_{k,n}]$; see for e.g. \cite[%
Chapter 9, \S 4, Prop 6, p. 471 and Ex 11, p. 475]{CLOS}. Hence, we arrive
at $|\boldsymbol{V}_{k,n}|=|\mathcal{A}_{k,n}^{+}|$ meaning that there are
as many distinct solutions to the Bethe ansatz equations as the dimension of
the subspace $\subset \mathcal{F}^{\otimes n}$ spanned by $\{|\lambda
\rangle :\lambda \in \mathcal{A}_{k,n}^{+}\}$.

The second claim, that the Bethe vectors are orthogonal, now follows from
observing that the eigenvalues of the transfer matrices (\ref{ncE}), (\ref%
{ncG'}) separate points, which means that for $\rho ,\sigma \in \mathcal{A}%
_{k,n}^{+}$ with $\rho \neq \sigma $ the corresponding eigenvalues have to
be different. From this fact one now easily deduces that the corresponding
scalar product between the eigenvectors has to vanish.
\end{proof}

We conclude this section by stating two more technical results which are
related to the behaviour of the Bethe roots under taking the inverse and
complex conjugation. They will be needed when introducing a Frobenius
structure on the coordinate ring $\Bbbk \lbrack \boldsymbol{V}_{k,n}]$ in
the next section.

\begin{lemma}[Inversion property]
Let $\lambda \in \mathcal{A}_{k,n}^{+}$ and $y=(y_{1},\ldots ,y_{k})$ be a
solution to the Bethe ansatz equations. Then
\begin{equation}
R_{\lambda }(y;t)=z^{k}R_{(n-\lambda _{k},\ldots ,n-\lambda _{2},n-\lambda
_{1})}(y^{-1};t)  \label{dualR}
\end{equation}%
and, thus, setting $z=1$ we have $P_{\lambda }(y;t)=P_{\lambda ^{\ast
}}(y^{-1};t)$ as well as $Q_{\lambda }(y;t)=Q_{\lambda ^{\ast }}(y^{-1};t)$
with $\lambda ^{\ast }$ being the inverse image of $(n-\lambda _{k},\ldots
,n-\lambda _{2},n-\lambda _{1})\in \mathcal{\tilde{A}}_{k,n}^{+}$ under the
bijection $\symbol{126}:\mathcal{A}_{k,n}^{+}\rightarrow \mathcal{\tilde{A}}%
_{k,n}^{+}$.
\end{lemma}

\begin{proof}
From the Bethe ansatz equations we have that $y_{1}^{n}\cdots
y_{k}^{n}=z^{k} $. Furthermore, one easily verifies that%
\begin{equation*}
w_{k}\left( \prod_{1\leq i<j\leq k}\frac{y_{i}-ty_{j}}{y_{i}-y_{j}}\right)
=\prod_{1\leq i<j\leq k}\frac{y_{i}t-y_{j}}{y_{i}-y_{j}}=\prod_{1\leq
i<j\leq k}\frac{y_{i}^{-1}-ty_{j}^{-1}}{y_{i}^{-1}-y_{j}^{-1}}\;.
\end{equation*}%
Hence, we arrive at the identity%
\begin{equation*}
y^{\lambda }\prod_{1\leq i<j\leq k}\frac{y_{i}-ty_{j}}{y_{i}-y_{j}}%
=w_{k}\left( y_{1}^{\lambda _{k}-n}\cdots y_{k-1}^{\lambda
_{2}-n}y_{k}^{\lambda _{1}-n}\prod_{1\leq i<j\leq k}\frac{y_{i}t-y_{j}}{%
y_{i}-y_{j}}\right)
\end{equation*}%
which gives the desired equation for $R_{\lambda }$. Exploiting the
definition (\ref{RPQ}) of $P_{\lambda }$ in terms of $R_{\lambda }$, the
remaining identities follow from $P_{\lambda }(y;t)=z^{m_{n}(\lambda )}P_{%
\tilde{\lambda}}(y;t)$ and the obvious fact that $b_{\lambda }(t)=b_{\lambda
^{\ast }}(t)$.
\end{proof}

\begin{lemma}[dual Bethe vectors]
Assume (\ref{tzassumption}) holds. For any solution $y=(y_{1},\ldots ,y_{k})$
of the Bethe ansatz equations we have that $P_{\tilde{\lambda}}(\bar{y}%
^{-1};t)=P_{\tilde{\lambda}}(y;t)$. This in particular, implies that the
dual Bethe vectors are given by%
\begin{equation}
\mathfrak{b}_{\sigma }^{\ast }=\frac{1}{||\mathfrak{b}_{\sigma }||^{2}}%
\sum_{\lambda \in \mathcal{A}_{k,n}^{+}}P_{\tilde{\lambda}}(y;t)\langle
\lambda |,\qquad ||\mathfrak{b}_{\sigma }||^{2}=\sum_{\lambda \in \mathcal{A}%
_{k,n}^{+}}Q_{\lambda }(y_{\sigma };t)P_{\lambda }(\bar{y}_{\sigma };t),
\label{dualBethevec}
\end{equation}%
that is, we have the identity $\left\langle \mathfrak{b}_{\rho }^{\ast }|%
\mathfrak{b}_{\sigma }\right\rangle =\delta _{\rho \sigma }$.
\end{lemma}

\begin{proof}
Assume (\ref{tzassumption}) holds. Then we can assume without loss of
generality that $z=1$. It then follows from (\ref{Epieri}) that the quantum
Hamiltonians $H_{r}^{\pm }$ defined in (\ref{IoM}) are (anti-)Hermitian and,
hence, we infer from (\ref{IoMspec}) that $g_{r}(y;t)=g_{r}(\bar{y}^{-1};t)$
for $r=1,\ldots ,n-1$. Observing that $t^{k}g_{r}(y;t^{-1})$ is a polynomial
of the $g_{r}(y;t)$'s for $r=1,\ldots ,n-1$ (this follows from (\ref%
{greduction})) we can employ (\ref{gbae}) to conclude that $g_{r}(y;t)=g_{r}(%
\bar{y}^{-1};t)$\ for $r\geq n$. Hence, we find $Q_{\tilde{\lambda}}(y;t)=Q_{%
\tilde{\lambda}}(\bar{y}^{-1};t)$ and $P_{\tilde{\lambda}}(\bar{y}%
^{-1};t)=P_{\tilde{\lambda}}(y;t)$ since the latter are polynomials in the $%
g_{r}$'s. Noting that $Q_{\tilde{\lambda}},P_{\tilde{\lambda}}$ are
homogeneous functions of degree $|\tilde{\lambda}|$ the $z$-dependence is
easily re-introduced.
\end{proof}

\subsection{Deformed fusion coefficients and Frobenius structures}

This section will see the formulation of the main result of this article:
there exists a natural Frobenius algebra structure on the
Kirillov-Reshetikhin module $W^{1,k}$ which we have previously identified
with the $k$-particle Fock space $\mathcal{F}_{k}^{\otimes n}$; see
Proposition \ref{KRmodule}. We are going to extend the base field from $%
\mathbb{C}(t)$ to $\mathbb{\Bbbk }$ as this will allow us to include the
idempotents.

\begin{theorem}[deformed Verlinde algebra]
\label{Frobenius}Let $\mathfrak{F}_{n,k}:=\mathcal{F}_{k}^{\otimes n}\otimes
_{\mathbb{C}(t)}\mathbb{\Bbbk }$ and set $z=1$. Define for $\lambda ,\mu \in
\mathcal{A}_{k,n}^{+}$ the product%
\begin{equation}
|\lambda \rangle \circledast |\mu \rangle :=\boldsymbol{Q}_{\lambda ^{\prime
}}^{\prime }|\mu \rangle  \label{Fproduct}
\end{equation}%
and the bilinear form $\eta :\mathfrak{F}_{n,k}\otimes \mathfrak{F}%
_{n,k}\rightarrow \mathbb{\Bbbk }$%
\begin{equation}
\eta (|\lambda \rangle \otimes |\mu \rangle ):=\delta _{\lambda \mu ^{\ast
}}/b_{\lambda }(t)\;.  \label{eta}
\end{equation}%
Then $(\mathfrak{F}_{n,k},\circledast ,\eta )$ is a commutative Frobenius
algebra with unit $|n^{k}\rangle $.
\end{theorem}

\begin{remark}
Implicit in the last theorem is the statement that $\boldsymbol{Q}_{\lambda
^{\prime }}^{\prime }|n^{k}\rangle =|\lambda \rangle $. Moreover, one checks
from the definition (\ref{ncG'}), (\ref{ncm}) that $\boldsymbol{g}%
_{r}^{\prime }|n^{k}\rangle =0$ for $r>k$ and, hence, $\boldsymbol{Q}%
_{\lambda ^{\prime }}^{\prime }|n^{k}\rangle =0$ for $\lambda \notin
\mathcal{A}_{k,n}^{+}$. That is, the family $\boldsymbol{B}_{n}:=\{%
\boldsymbol{Q}_{\lambda }^{\prime }:\lambda \in \mathcal{P}_{n}^{+}\}\subset
\mathcal{H}_{q}^{\otimes n}$ of noncommutative analogues of Macdonald
polynomials generates the canonical basis in the $U_{n}$-module $%
S^{k}(V)\subset V^{\otimes k}$ when acting on the highest weight vector $%
|n^{k}\rangle =v_{n}\otimes \cdots \otimes v_{n}$.
\end{remark}

The proof of these statements will employ the expression of the matrix
elements $\langle \lambda |\boldsymbol{Q}_{\nu ^{\prime }}^{\prime }|\mu
\rangle $ in terms of the Bethe vectors (\ref{Bethevec}) which we state as a
separate lemma. To ease the notation we introduce the transition matrix
\begin{equation}
\mathcal{S}_{\lambda \mu }(t):=||\mathfrak{b}_{\mu }||\langle \mathfrak{b}%
_{\mu }^{\ast }|\lambda \rangle =\frac{P_{\lambda }(y_{\mu };t)}{||\mathfrak{%
b}_{\mu }||}  \label{modS}
\end{equation}%
from the basis of \emph{normalised} Bethe vectors to the vectors $\{|\lambda
\rangle :\lambda \in \mathcal{A}_{k,n}^{+}\}$ which we identified with the
canonical basis in $S^{k}(V)$. The matrix elements of the inverse matrix $%
\mathcal{S}^{-1}(t)$ are given by%
\begin{equation}
\mathcal{S}_{\mu \lambda }^{-1}(t)=\frac{\langle \lambda |\mathfrak{b}_{\mu
}\rangle }{||\mathfrak{b}_{\mu }||}=\frac{Q_{\lambda }(y_{\mu }^{-1};t)}{||%
\mathfrak{b}_{\mu }||}=b_{\lambda }(t)\overline{\mathcal{S}_{\lambda \mu }(t)%
}=b_{\lambda }(t)z^{\frac{|\lambda |-|\lambda ^{\ast }|}{n}}\mathcal{S}%
_{\lambda ^{\ast }\mu }(t)\;,  \label{modSinv}
\end{equation}%
where we have used (\ref{dualR}), (\ref{dualBethevec}) in the last two
equalities under the assumption (\ref{tzassumption}). Labelling the
partition $n^{k}$ with `$0$' note that we have in particular $\mathcal{S}%
_{0\mu }(t)=||\mathfrak{b}_{\mu }||^{-1}$ because of (\ref{PQreduction}).

\begin{lemma}[deformed Verlinde formulae]
\label{Verlinde}Let $\lambda ,\mu ,\nu \in \mathcal{A}_{k,n}^{+}$. Then we
have%
\begin{equation}
\langle \lambda |\boldsymbol{Q}_{\nu ^{\prime }}^{\prime }|\mu \rangle =z^{%
\frac{|\mu |+|\nu |-|\lambda |}{n}}\sum_{\sigma \in \mathcal{A}_{k,n}^{+}}%
\frac{\mathcal{S}_{\mu \sigma }(t)\mathcal{S}_{\nu \sigma }(t)\mathcal{S}%
_{\sigma \lambda }^{-1}(t)}{\mathcal{S}_{0\sigma }(t)}  \label{Verlinde1}
\end{equation}%
and%
\begin{equation}
\langle \lambda |\boldsymbol{S}_{\tilde{\nu}^{\prime }}^{\prime }|\mu
\rangle =z^{\frac{|\mu |+|\tilde{\nu}|-|\lambda |}{n}}\sum_{\sigma \in
\mathcal{A}_{k,n}^{+}}s_{\tilde{\nu}}(y_{\sigma })\mathcal{S}_{\nu \sigma
}(t)\mathcal{S}_{\sigma \lambda }^{-1}(t)\;.  \label{Verlinde2}
\end{equation}%
Both coefficients vanish identically unless $d=\frac{|\mu |+|\nu |-|\lambda |%
}{n}\in \mathbb{Z}_{\geq 0}$.
\end{lemma}

\begin{proof}
Denote by $y_{\sigma }=z^{1/n}y_{\sigma }^{\prime }$ the solution to (\ref%
{bae}) under the bijection of Theorem \ref{complete}. Then we have the
following simple calculation,%
\begin{eqnarray}
\langle \lambda |\boldsymbol{Q}_{\nu ^{\prime }}^{\prime }|\mu \rangle
&=&\sum_{\sigma \in \mathcal{A}_{k,n}^{+}}\langle \lambda |\boldsymbol{Q}%
_{\nu ^{\prime }}^{\prime }|\mathfrak{b}_{\sigma }\rangle \langle \mathfrak{b%
}_{\sigma }^{\ast }|\mu \rangle =\sum_{\sigma \in \mathcal{A}_{k,n}^{+}}z^{%
\frac{|\nu |}{n}}P_{\tilde{\nu}}(y_{\sigma }^{\prime };t)\langle \lambda |%
\mathfrak{b}_{\sigma }\rangle \langle \mathfrak{b}_{\sigma }^{\ast }|\mu
\rangle  \notag  \label{Vcalc} \\
&=&z^{\frac{|\mu |+|\nu |-|\lambda |}{n}}\sum_{\sigma \in \mathcal{A}%
_{k,n}^{+}}\frac{P_{\tilde{\mu}}(y_{\sigma }^{\prime };t)P_{\tilde{\nu}%
}(y_{\sigma }^{\prime };t)Q_{\lambda ^{\ast }}(y_{\sigma }^{\prime };t)}{||%
\mathfrak{b}_{\sigma }||^{2}}
\end{eqnarray}%
and the first assertion now follows from the definition (\ref{modS}). Here
we have used in the first line that the Bethe vectors form an eigenbasis in $%
\mathcal{F}_{k}^{\otimes n}$ and in the second line the explicit expansions (%
\ref{Bethevec}), (\ref{dualBethevec}) of them and their dual vectors with
respect to the basis $\{|\lambda \rangle \}_{\lambda \in \mathcal{A}%
_{k,n}^{+}}$ as well as (\ref{PQreduction}). The second identity for $%
\langle \lambda |\boldsymbol{S}_{\tilde{\nu}^{\prime }}^{\prime }|\mu
\rangle $ follows from a computation along the same lines using (\ref%
{ncG'spec}) and the definition (\ref{ncQ'S'}).

To see that both coefficients vanish identically unless $d=\frac{|\mu |+|\nu
|-|\lambda |}{n}\in \mathbb{Z}_{\geq 0}$, recall from the definition (\ref%
{NQ'KS'}) and (\ref{ncQ'S'}) that $\boldsymbol{Q}_{\nu ^{\prime }}^{\prime }$
and $\boldsymbol{S}_{\nu ^{\prime }}^{\prime }$ are polynomial in the $a_{i}$%
's and the $z$-dependence of the latter is given in (\ref{qplacticrep}).
Thus, the overall power $\frac{|\mu |+|\nu |-|\lambda |}{n}$ of $z$ in $%
\langle \lambda |\boldsymbol{Q}_{\nu ^{\prime }}^{\prime }|\mu \rangle
,\langle \lambda |\boldsymbol{S}_{\nu ^{\prime }}^{\prime }|\mu \rangle $
must be a non-negative integer by definition.
\end{proof}

\begin{remark}
The notation chosen for the transition matrix is not coincidental. It is a
generalisation of the modular $\mathcal{S}$-matrix of the Verlinde ring
which is recovered when formally setting $t=0$ in (\ref{modS}). In \cite[%
Props 6.11 and 6.15, Def 6.13]{KS} it has been shown that the $\mathcal{S}%
(0) $-matrix is the transition matrix from the Bethe vectors to the basis $%
\{|\lambda \rangle :\lambda \in \mathcal{A}_{k,n}^{+}\}$ and coincides with
the famous Kac-Peterson formula. Thus, we can interpret (\ref{Verlinde1}) as
`deformed Verlinde formula'.
\end{remark}

Formula (\ref{Verlinde1}) implies several obvious `symmetries' of the matrix
elements $\langle \lambda |\boldsymbol{Q}_{\nu ^{\prime }}^{\prime }|\mu
\rangle $ which we summarise in the following corollary.

\begin{corollary}[symmetries]
\label{symmetries}Let $N_{\mu \nu }^{\lambda }(t)=\langle \lambda |%
\boldsymbol{Q}_{\nu ^{\prime }}^{\prime }|\mu \rangle $ and set $z=1$.
Assume further that $\bar{t}=t$. Then we have the following identities:

\begin{enumerate}
\item $N_{\mu \nu }^{\lambda }(t)=N_{\nu \mu }^{\lambda }(t)$ and $\overline{%
N_{\mu \nu }^{\lambda }(t)}=N_{\mu \nu }^{\lambda }(t)=N_{\mu ^{\ast }\nu
^{\ast }}^{\lambda ^{\ast }}(t)$

\item Charge conjugation: $b_{\lambda }(t)N_{\mu \nu }^{\lambda }(t)=b_{\nu
}(t)N_{\mu \lambda ^{\ast }}^{\nu ^{\ast }}(t)$

\item $N_{\mu n^{k}}^{\lambda }(t)=\delta _{\lambda \mu }$ and $N_{\lambda
\mu }^{n^{k}}(t)=\delta _{\lambda \mu ^{\ast }}/b_{\lambda }(t)$

\item Rotation invariance: $N_{\mu \nu }^{\lambda }(t)=N_{\limfunc{rot}(\mu
)\nu }^{\limfunc{rot}(\lambda )}(t)$.
\end{enumerate}
\end{corollary}

\begin{proof}
Note that invariance under complex conjugation, $\overline{N_{\mu \nu
}^{\lambda }(t)}=N_{\mu \nu }^{\lambda }(t)$, follows simply from the
definition (\ref{ncQ'S'}) and observing that $\langle \lambda |a_{i}|\mu
\rangle \in \mathbb{Z}[t]$ for all $i=1,\ldots ,n$. The remaining identities
in 1-2 are now trivial consequences of (\ref{dualR}), $Q_{\lambda
}=b_{\lambda }P_{\lambda }$ and $b_{\lambda }=b_{\lambda ^{\ast }}$. For
statement 3 simply use (\ref{PQreduction}) and observe that $P_{\emptyset
}=1 $, whence $N_{\mu n^{k}}^{\lambda }(t)=\sum_{\sigma }\mathcal{S}_{\mu
\sigma }(t)\mathcal{S}_{\sigma \lambda }^{-1}(t)=\delta _{\mu \lambda }$.
Now use statement 2 to obtain the second identity in 3. Finally, the last
statement is a simply consequence of $P_{\limfunc{rot}(\lambda
)}(y;t)=e_{k}(y)P_{\limfunc{rot}(\lambda )}(y;t),$ $b_{\lambda }=b_{\limfunc{%
rot}(\lambda )}$ and that $e_{k}(\bar{y})=e_{k}(y^{-1})=e_{k}(y)^{-1}$.
\end{proof}

We are now ready to prove the main result.

\begin{proof}[Theorem \protect\ref{Frobenius}]
From the first and second property in the last Corollary it is now clear
that the product (\ref{Fproduct}) is commutative, $N_{\nu \mu }^{\lambda
}(t)=\langle \lambda |\boldsymbol{Q}_{\nu ^{\prime }}^{\prime }|\mu \rangle
=\langle \lambda |\boldsymbol{Q}_{\mu ^{\prime }}^{\prime }|\nu \rangle
=N_{\mu \nu }^{\lambda }(t),$ and that $|n^{k}\rangle \ast |\lambda \rangle
=|\lambda \rangle $. Associativity then easily follows from (\ref{integrable}%
),%
\begin{multline*}
\lambda \circledast (\mu \circledast \nu )=\sum_{\sigma \in \mathcal{A}%
_{k,n}^{+}}N_{\mu \nu }^{\rho }(t)\lambda \circledast \rho =\sum_{\rho
,\sigma \in \mathcal{A}_{k,n}^{+}}N_{\lambda \rho }^{\sigma }(t)N_{\nu \mu
}^{\rho }(t)~\sigma \\
=\sum_{\rho ,\sigma \in \mathcal{A}_{k,n}^{+}}\langle \sigma |\boldsymbol{Q}%
_{\nu ^{\prime }}^{\prime }\boldsymbol{Q}_{\lambda ^{\prime }}^{\prime }|\mu
\rangle ~\sigma =\sum_{\rho ,\sigma \in \mathcal{A}_{k,n}^{+}}N_{\nu \rho
}^{\sigma }(t)N_{\lambda \mu }^{\rho }(t)~\sigma =(\lambda \circledast \mu
)\circledast \nu \;,
\end{multline*}%
where for ease of notation we have denoted vectors simply by partitions.
Thus, $(\mathfrak{F}_{n,k},\circledast )$ is a unital, associative and
commutative algebra. Since $\{\lambda :\lambda \in \mathcal{A}_{k,n}^{+}\}$
is a basis and the map $\lambda \mapsto \lambda ^{\ast }$ simply amounts to
a reordering of this basis, it is obvious that $\eta $ is nondegenerate. One
now easily checks with the help of (\ref{Verlinde1}) and $b_{\lambda ^{\ast
}}(t)=b_{\lambda }(t)$ that $\eta (\lambda ,\mu )=\eta (\mu ,\lambda )$ and%
\begin{equation*}
\eta (\lambda \circledast \mu ,\nu )=N_{\lambda \mu }^{\nu ^{\ast
}}(t)/b_{\nu }(t)=N_{\mu \nu }^{\lambda ^{\ast }}(t)/b_{\lambda }(t)=\eta
(\lambda ,\mu \circledast \nu )\;
\end{equation*}%
according to the second property in Corollary \ref{symmetries}.
\end{proof}

Denote by $\boldsymbol{Q}_{\lambda }^{(k)}$ the restriction of $\boldsymbol{Q%
}_{\lambda }^{\prime }$ to the subspace $\mathcal{F}_{k}^{\otimes n}$
spanned by $\{|\lambda \rangle ~|~\lambda \in \mathcal{A}_{k,n}^{+}\}$.

\begin{corollary}[deformed fusion matrices]
Set $z=1$ and consider the subalgebra $\subset \limfunc{End}\mathfrak{F}%
_{n,k}$ generated by $\{\boldsymbol{Q}_{\lambda ^{\prime }}^{(k)}~|~\lambda
\in \mathcal{A}_{k,n}^{+}\}$. Then the map $|\lambda \rangle \mapsto
\boldsymbol{Q}_{\lambda ^{\prime }}^{(k)}$ is an algebra isomorphism. That
is, we have for all $\mu ,\nu \in \mathcal{A}_{k,n}^{+}$\ the product
expansion%
\begin{equation}
\boldsymbol{Q}_{\mu ^{\prime }}^{(k)}\boldsymbol{Q}_{\nu ^{\prime
}}^{(k)}=\sum_{\lambda \in \mathcal{A}_{k,n}^{+}}N_{\mu \nu }^{\lambda }(t)%
\boldsymbol{Q}_{\lambda ^{\prime }}^{(k)}\;.
\end{equation}
\end{corollary}

\begin{proof}
By definition of (\ref{Fproduct}) and exploiting associativity we compute
\begin{eqnarray*}
\boldsymbol{Q}_{\mu ^{\prime }}^{\prime }\boldsymbol{Q}_{\nu ^{\prime
}}^{\prime }|\rho \rangle &=&|\mu \rangle \circledast (|\nu \rangle
\circledast |\rho \rangle )=(|\mu \rangle \circledast |\nu \rangle
)\circledast |\rho \rangle \\
&=&\sum_{\lambda \in \mathcal{A}_{k,n}^{+}}N_{\mu \nu }^{\lambda
}(t)|\lambda \rangle \circledast |\rho \rangle =\sum_{\lambda \in \mathcal{A}%
_{k,n}^{+}}N_{\mu \nu }^{\lambda }(t)\boldsymbol{Q}_{\lambda ^{\prime
}}^{\prime }|\rho \rangle
\end{eqnarray*}%
for any $\rho \in \mathcal{A}_{k,n}^{+}$. Hence, the assertion follows.
\end{proof}

\begin{remark}
The last presentation of the Frobenius algebra in $\limfunc{End}\mathfrak{F}%
_{n,k}$ is particularly convenient to show that it specialises to the
Verlinde ring when restricting the coefficients to $\mathbb{Z}[t]$. Namely,
setting formally $t=0$ in (\ref{qKnuth}) by making the replacement $\mathcal{%
\hat{U}}_{n}^{-}\rightarrow \mathcal{\hat{U}}_{n}^{-}/\langle t\rangle $,
one obtains a representation of the local affine plactic algebra \cite[Def
5.4 and Prop 5.8]{KS} as mentioned earlier. Thus, the deformed fusion
matrices specialise for $t=0$ to%
\begin{equation*}
\boldsymbol{Q}_{\lambda }^{\prime }=\boldsymbol{S}_{\lambda }^{\prime }=\det
(\boldsymbol{h}_{\lambda _{i}-i+j})=\boldsymbol{s}_{\lambda },
\end{equation*}%
where $\boldsymbol{h}_{r}=\sum_{\lambda \vdash r}\boldsymbol{m}_{\lambda }$
is now the affine plactic homogeneous symmetric polynomial and $\boldsymbol{s%
}_{\lambda }$ is the affine plactic Schur polynomial defined in \cite[Def
5.15-6]{KS} and \cite[Prop 4.1 and Prop 5.1]{Korff}. The latter have been
identified with the fusion matrices of the Verlinde ring \cite[Theorem 6.18]%
{KS} and we can therefore conclude that the constant terms $N_{\mu \nu
}^{\lambda }(0)$ and $K_{\nu ,\lambda /d/\mu }(0)$ are the
Wess-Zumino-Novikov-Witten fusion coefficients $\mathcal{N}_{\hat{\mu}\hat{%
\nu}}^{(k)\hat{\lambda}}$. The indices of the fusion coefficients are given
via the following bijection $\lambda \mapsto \hat{\lambda}%
:=\tsum_{i=1}^{n}m_{i}(\lambda )\hat{\omega}_{i}$ between $\mathcal{A}%
_{k,n}^{+}$ and the set of \emph{integral dominant }$\widehat{\mathfrak{sl}}%
(n)$ \emph{weights at level }$k$,%
\begin{equation}
P_{n,k}^{+}=\left\{ \left. \hat{\lambda}=\tsum_{i=1}^{n}m_{i}\hat{\omega}%
_{i}~\right\vert ~\tsum_{i=1}^{n}m_{i}=k\right\} ,  \label{levelkweights}
\end{equation}%
where the $\hat{\omega}_{i}$'s denote the affine fundamental weights.
\end{remark}

The following proposition shows that our discussion of the Bethe ansatz in
the previous section has the algebraic interpretation of computing its
Peirce decomposition.

\begin{proposition}[idempotents]
The Bethe vectors (\ref{Bethevec}) are the idempotents of the Frobenius
algebra $(\mathfrak{F}_{k,n},\circledast ,\eta )$,%
\begin{equation}
\mathfrak{e}_{\lambda }\circledast \mathfrak{e}_{\mu }=\delta _{\lambda \mu }%
\mathfrak{e}_{\lambda },\qquad \mathfrak{e}_{\lambda }:=|\mathcal{S}%
_{0\lambda }|^{2}\mathfrak{b}_{\lambda }=\mathcal{S}_{0\lambda }\sum_{\mu
\in \mathcal{A}_{k,n}^{+}}\mathcal{S}_{\lambda \mu }^{-1}|\mu \rangle \;.
\end{equation}%
Moreover, we have the following decomposition of the unit, $|n^{k}\rangle
=\sum_{\lambda \in \mathcal{A}_{k,n}^{+}}\mathfrak{e}_{\lambda }$. The dual
Bethe vectors on the other hand obey%
\begin{equation}
\Delta _{n,k}\mathfrak{e}_{\lambda }^{\ast }=\mathfrak{e}_{\lambda }^{\ast
}\otimes \mathfrak{e}_{\lambda }^{\ast },\qquad \mathfrak{e}_{\lambda
}^{\ast }:=\mathfrak{b}_{\lambda }^{\ast }/|\mathcal{S}_{0\lambda
}|^{2}=\sum_{\mu \in \mathcal{A}_{k,n}^{+}}\langle \mu |~\frac{\mathcal{S}%
_{\mu \lambda }}{\mathcal{S}_{0\lambda }}
\end{equation}%
where $\Delta _{n,k}:\mathfrak{F}_{n,k}\rightarrow \mathfrak{F}_{n,k}\otimes
\mathfrak{F}_{n,k}$ is the coproduct induced by $\eta $.
\end{proposition}

\begin{proof}
The first statement is a trivial consequence of (\ref{Verlinde1}),%
\begin{equation*}
\mathfrak{e}_{\lambda }\circledast \mathfrak{e}_{\mu }=\mathcal{S}_{0\lambda
}\sum_{\nu \in \mathcal{A}_{k,n}^{+}}\mathcal{S}_{\lambda \nu }^{-1}%
\boldsymbol{Q}_{\nu ^{\prime }}^{\prime }\mathfrak{e}_{\mu }=\mathcal{S}%
_{0\lambda }\sum_{\nu \in \mathcal{A}_{k,n}^{+}}\frac{\mathcal{S}_{\lambda
\nu }^{-1}\mathcal{S}_{\nu \mu }}{\mathcal{S}_{0\mu }}\mathfrak{e}_{\mu
}=\delta _{\lambda \mu }\mathfrak{e}_{\lambda }\;.
\end{equation*}%
Thus, we have for any $\mu $ that $\mathfrak{e}_{\mu }\circledast
(\sum_{\lambda }\mathfrak{e}_{\lambda })=\mathfrak{e}_{\mu }$. Since the $\{%
\mathfrak{b}_{\lambda }\}$ and, hence, the $\{\mathfrak{e}_{\lambda }\}$ are
a basis it follows that for any $\mathfrak{f}\in \mathfrak{F}_{n,k}$ we have
$\mathfrak{f}\circledast (\sum_{\lambda }\mathfrak{e}_{\lambda })=\mathfrak{f%
}$. Setting $\boldsymbol{1}:=|n^{k}\rangle $ and $\boldsymbol{1}^{\prime
}:=\sum_{\lambda }\mathfrak{e}_{\lambda }$ it follows that $\eta (%
\boldsymbol{1}-\boldsymbol{1}^{\prime },\mathfrak{f})=\eta (\boldsymbol{1},%
\mathfrak{f}\circledast (\boldsymbol{1}-\boldsymbol{1}^{\prime }))=0$ for
all $\mathfrak{f}\in \mathfrak{F}_{n,k}$ and, therefore, $\boldsymbol{1}=%
\boldsymbol{1}^{\prime }$ because $\eta $ is non-degenerate.

To prove the second statement we explicitly compute the coproduct using the
known facts about the structure of Frobenius algebras; see e.g. \cite{Kock}.
Let $\mathfrak{m}:\mathfrak{F}_{n,k}\otimes \mathfrak{F}_{n,k}\rightarrow
\mathfrak{F}_{n,k}$ be the regular representation or multiplication map, $%
\mathfrak{m}(|\mu \rangle \otimes |\nu \rangle )=|\mu \rangle \circledast
|\nu \rangle $ and $\mathfrak{m}^{\ast }:\mathfrak{F}_{n,k}^{\ast
}\rightarrow \mathfrak{F}_{n,k}^{\ast }\otimes \mathfrak{F}_{n,k}^{\ast }$
its dual map. Then the co-product $\Delta _{n,k}$ is obtained via the
following commutative diagram
\begin{equation}
\begin{CD} \mathfrak{F}_{n,k} @>\Delta_{n,k}>>
\mathfrak{F}_{n,k}\otimes\mathfrak{F}_{n,k}\\ @VV{\Phi}V
@VV{\Phi\otimes\Phi}V\\ \mathfrak{F}^\ast_{n,k} @>\mathfrak{m}^{\ast}>>
\mathfrak{F}^\ast_{n,k}\otimes\mathfrak{F}^\ast_{n,k} \end{CD}\quad ,
\label{diagram}
\end{equation}%
where the Frobenius isomorphism $\Phi :\mathfrak{F}_{n,k}\rightarrow
\mathfrak{F}_{n,k}^{\ast }$ is given by%
\begin{equation}
\Phi :|\lambda \rangle \mapsto b_{\lambda }^{-1}(t)\langle \lambda ^{\ast
}|\;.  \label{Frobiso}
\end{equation}%
We claim that the coproduct in the basis $\{|\lambda \rangle \}$ is computed
to
\begin{equation}
\Delta _{n,k}|\lambda \rangle =\sum_{\mu ,\nu \in \mathcal{A}_{k,n}^{+}}%
\frac{b_{\mu }(t)b_{\nu }(t)}{b_{\lambda }(t)}N_{\mu \nu }^{\lambda }(t)|\nu
\rangle \otimes |\mu \rangle \;.  \label{Fcoprod}
\end{equation}%
Thus, we have in particular $(\Phi \otimes \Phi )\Delta _{n,k}(|\lambda
\rangle )(|\nu \rangle \otimes |\mu \rangle )=b_{\lambda }^{-1}N_{\mu ^{\ast
}\nu ^{\ast }}^{\lambda }$, where we have used once more that $b_{\mu ^{\ast
}}=b_{\mu }$. According to (\ref{diagram}) this result has to match%
\begin{equation*}
\mathfrak{m}^{\ast }\circ \Phi (|\lambda \rangle )(|\nu \rangle \otimes |\mu
\rangle )=b_{\lambda }^{-1}\langle \lambda ^{\ast }|\mu \circledast \nu
\rangle =b_{\lambda }^{-1}N_{\nu \mu }^{\lambda ^{\ast }}~.
\end{equation*}%
That both results are indeed equal now follows from the properties in
Corollary \ref{symmetries}.

Identify the bra-vector $\langle \lambda |$ in $\mathfrak{F}_{n,k}^{\ast }$
with the ket-vector $b_{\lambda }|\lambda \rangle $ in $\mathfrak{F}_{n,k}$.
Then a straightforward computation yields the last assertion,%
\begin{eqnarray*}
\Delta _{n,k}\mathfrak{e}_{\lambda }^{\ast } &=&\sum_{\mu ,\nu ,\rho }\frac{%
\mathcal{S}_{\mu \lambda }}{\mathcal{S}_{0\lambda }}N_{\nu \rho }^{\mu
}b_{\rho }b_{\nu }|\rho \rangle \otimes |\nu \rangle \\
&=&\sum_{\mu ,\nu ,\rho ,\sigma }\frac{\mathcal{S}_{\mu \lambda }}{\mathcal{S%
}_{0\lambda }}\frac{\mathcal{S}_{\nu \sigma }\mathcal{S}_{\rho \sigma }%
\mathcal{S}_{\sigma \mu }^{-1}}{\mathcal{S}_{0\sigma }}b_{\rho }b_{\nu
}|\rho \rangle \otimes |\nu \rangle \\
&=&\sum_{\mu ,\nu ,\rho ,\sigma }\frac{\mathcal{S}_{\nu \lambda }\mathcal{S}%
_{\rho \lambda }}{\mathcal{S}_{0\lambda }^{2}}b_{\rho }b_{\nu }|\rho \rangle
\otimes |\nu \rangle =\mathfrak{e}_{\lambda }^{\ast }\otimes \mathfrak{e}%
_{\lambda }^{\ast }\;.
\end{eqnarray*}%
Here we have used (\ref{Verlinde1}) in the second line.
\end{proof}

\begin{remark}
The explicit computation of the coproduct (\ref{Fcoprod}) ties the Frobenius
algebra $\mathfrak{F}_{n,k}$ to our earlier discussion of cylindric skew
Macdonald functions. Make the formal identification $|\lambda \rangle
\mapsto Q_{\tilde{\lambda}^{\prime }}^{\prime }$ and $\langle \lambda
|\mapsto P_{\tilde{\lambda}^{\prime }}^{\prime }$ but instead of taking the
usual product (\ref{Demazureprod}) in the ring of symmetric functions define
the fusion product $Q_{\tilde{\mu}^{\prime }}^{\prime }\ast Q_{\tilde{\nu}%
^{\prime }}^{\prime }:=\sum_{\lambda \in \mathcal{A}_{k,n}^{+}}N_{\mu \nu
}^{\lambda }(t)Q_{\tilde{\lambda}^{\prime }}^{\prime }$. Then the Frobenius
coproduct yields
\begin{eqnarray*}
\Delta _{n,k}P_{\tilde{\lambda}^{\prime }}^{\prime } &=&\sum_{\mu ,\nu \in
\mathcal{A}_{k,n}^{+}}N_{\mu \nu }^{\lambda }(t)P_{\nu ^{\prime }}^{\prime
}\otimes P_{\mu ^{\prime }}^{\prime } \\
&=&\sum_{\mu \in \mathcal{A}_{k,n}^{+},d\geq 0}P_{\lambda ^{\prime }/d/\mu
^{\prime }}^{\prime }\otimes P_{\mu ^{\prime }}^{\prime }\;.
\end{eqnarray*}%
This links the Frobenius algebra $\mathfrak{F}_{n,k}$ to the partition
function of the statistical mechanics model with transfer matrix (\ref{ncG'}%
).
\end{remark}

In principle, we can compute the structure constants $N_{\mu \nu }^{\lambda
}(t)$ of the Frobenius algebra $\mathfrak{F}_{n,k}$ from the representation (%
\ref{qplacticrep}) employing (\ref{Fproduct}) and the definition (\ref%
{ncQ'S'}). An alternative is to use our description of the coordinate ring $%
\Bbbk \lbrack \boldsymbol{V}_{k,n}]$ in the previous section.

\begin{theorem}[restricted Hall algebra]
The map $|\lambda \rangle \mapsto \lbrack P_{\tilde{\lambda}}]$ defines for $%
z=1$ an algebra isomorphism $\mathfrak{F}_{n,k}\cong \Bbbk \lbrack
\boldsymbol{V}_{k,n}]$. That is, the coefficients in the product expansions%
\begin{equation}
\lbrack P_{\tilde{\mu}}][P_{\tilde{\nu}}]:=[P_{\tilde{\mu}}P_{\tilde{\nu}%
}]=\sum_{\lambda \in \mathcal{A}_{k,n}^{+}}\tilde{N}_{\mu \nu }^{\lambda
}(t)[P_{\tilde{\lambda}}],\qquad  \label{PPquotient}
\end{equation}%
and%
\begin{equation}
\lbrack P_{\tilde{\mu}}][s_{\tilde{\nu}}]:=[P_{\tilde{\mu}}s_{\tilde{\nu}%
}]=\sum_{\lambda \in \mathcal{A}_{k,n}^{+}}K_{\nu ,\lambda /d/\mu }(t)[P_{%
\tilde{\lambda}}]  \label{Psquotient}
\end{equation}%
coincide with the expansion coefficients of the cylindric skew Macdonald
functions (\ref{cylP'}),%
\begin{equation}
\tilde{N}_{\mu \nu }^{\lambda }(t)=N_{\mu \nu }^{\lambda }(t)=\langle
\lambda |\boldsymbol{Q}_{\tilde{\nu}^{\prime }}^{\prime }|\mu \rangle \qquad
\text{and\qquad }K_{\nu ,\lambda /d/\mu }(t)=\langle \lambda |\boldsymbol{S}%
_{\tilde{\nu}^{\prime }}^{\prime }|\mu \rangle \;.  \label{NQ'KS'}
\end{equation}
\end{theorem}

\begin{proof}
The proof rests once more on the existence of an eigenbasis, Theorem \ref%
{complete}, and the expression (\ref{Verlinde1}). Namely, using (\ref%
{PQreduction}) we have
\begin{eqnarray}
N_{\mu \nu }^{\lambda }(t) &=&\sum_{\sigma \in \mathcal{A}_{k,n}^{+}}\frac{%
P_{\tilde{\mu}}(y_{\sigma };t)P_{\tilde{\nu}}(y_{\sigma };t)\mathcal{S}%
_{\sigma \lambda }^{-1}(t)}{||\mathfrak{b}_{\sigma }||}  \notag \\
&=&\sum_{\rho ,\sigma \in \mathcal{A}_{k,n}^{+}}\frac{\tilde{N}_{\mu \nu
}^{\rho }(t)P_{\tilde{\rho}}(y_{\sigma };t)\mathcal{S}_{\sigma \lambda
}^{-1}(t)}{||\mathfrak{b}_{\sigma }||} \\
&=&\sum_{\rho \in \mathcal{A}_{k,n}^{+}}\tilde{N}_{\mu \nu }^{\rho
}(t)\sum_{\sigma \in \mathcal{A}_{k,n}^{+}}\mathcal{S}_{\rho \sigma }(t)%
\mathcal{S}_{\sigma \lambda }^{-1}(t)=\tilde{N}_{\mu \nu }^{\lambda }(t)\;.
\notag
\end{eqnarray}%
The second line employs Lemma \ref{nullstellen} which ensures that the
expansion of the product $P_{\tilde{\mu}}(y_{\sigma };t)P_{\tilde{\nu}%
}(y_{\sigma };t)$ equals the expansion of $[P_{\tilde{\mu}}P_{\tilde{\nu}}]$
in $\Bbbk \lbrack \boldsymbol{V}_{k,n}]$. The second assertion follows from
an analogous computation.
\end{proof}

Recall that $Q_{(r)}=g_{r}$ and $P_{(1^{r})}=e_{r}$, then a direct
consequence of our earlier computations and the last theorem is the
following obvious corollary which links the transfer matrices (\ref{ncE})
and (\ref{ncG'}) to the coordinate ring.

\begin{corollary}[Pieri rules in the quotient]
\label{Pieri}Let $\mu \in \mathcal{A}_{k,n}^{+}$ and $0\leq r<n$, $0\leq
r^{\prime }\leq k$. Then we have the following modified Pieri rules in the
coordinate ring $\mathfrak{R}_{n,k}[z]:=\Bbbk \lbrack \boldsymbol{V}%
_{k,n}]\otimes _{\Bbbk }\Bbbk \lbrack z],$
\begin{eqnarray}
\lbrack g_{r}P_{\tilde{\mu}}] &=&\sum_{\lambda \in \mathcal{A}%
_{k,n}^{+}}\langle \lambda |\boldsymbol{e}_{r}|\mu \rangle \lbrack P_{\tilde{%
\lambda}}]=\sum_{\substack{ \lambda /d/\mu =(r),  \\ \lambda \in \mathcal{A}%
_{k,n}^{+}}}z^{d}\Phi _{\lambda /d/\mu }(t)[P_{\tilde{\lambda}}],
\label{quotientpieri1} \\
\lbrack e_{r^{\prime }}P_{\tilde{\mu}}] &=&\sum_{\lambda \in \mathcal{A}%
_{k,n}^{+}}\langle \lambda |\boldsymbol{g}_{r^{\prime }}^{\prime }|\mu
\rangle \lbrack P_{\tilde{\lambda}}]=\sum_{\substack{ \lambda /d/\mu
=(1^{r^{\prime }}),  \\ \lambda \in \mathcal{A}_{k,n}^{+}}}z^{d}\Psi
_{\lambda ^{\prime }/d/\mu ^{\prime }}^{\prime }(t)[P_{\tilde{\lambda}}]\;.
\label{quotientpieri2}
\end{eqnarray}
\end{corollary}

\begin{remark}[Open boundary conditions]
Note that either of the Pieri rules (\ref{quotientpieri1}) and (\ref%
{quotientpieri2}) fixes the product in $\mathfrak{R}_{n,k}[z]$. Setting $z=0$
and $t=-1$ the quotient $\mathfrak{R}_{n,n}[z]/\langle z,t+1\rangle $ is
isomorphic to the cohomology ring of the orthogonal Grassmannian OG$(n,2n)$
in the basis of $P$-polynomials. The Pieri rule for $Q$-polynomials
coincides with the cohomology ring of the Lagrangian Grassmannian LG$%
(n-1,2n-2)$. For general $z$ we obtain a deformation of these cohomology
rings which is different from the usual quantum cohomology \cite{BKT}.
\end{remark}

\subsubsection{Algorithm to compute deformed fusion coefficients}

We demonstrate on an explicit example how the expansion coefficients (\ref%
{cylP'2s}) and (\ref{cylP'2P'}) can be computed in the coordinate ring $%
\Bbbk \lbrack \boldsymbol{V}_{k,n}]$. The general procedure can be described
as follows: first compute the normal product expansion $P_{\tilde{\mu}}P_{%
\tilde{\nu}}=\sum_{\lambda }f_{\mu \nu }^{\lambda }(t)P_{\lambda }$ in $%
\mathbb{Z}[t][e_{1},\ldots ,e_{k}]$, which is possible since explicit
formulae for computing the coefficients $f_{\mu \nu }^{\lambda }(t)$ are
known; see \cite{Macdonald}. In the second step rewrite those $P_{\lambda }$
with $\lambda $ outside the fundamental region $\mathcal{\tilde{A}}%
_{k,n}^{+} $ in terms of $P_{\tilde{\lambda}}$'s with $\tilde{\lambda}\in
\mathcal{\tilde{A}}_{k,n}^{+}$ using first the affine (\ref{straight4}) and
then the non-affine straightening rules (\ref{straight1}), (\ref{straight2}%
). Collecting coefficients of the individual terms one obtains the expansion
in $\Bbbk \lbrack \boldsymbol{V}_{k,n}]$ and, hence, $N_{\mu \nu }^{\lambda
}(t) $.

\begin{example}
Set $n=4,$ $k=3$ and $z=1$. Consider the partitions $\lambda =(3,2,1)$ and $%
\mu =(4,3,1)$. We exploit the fact that $f_{\lambda \mu }^{\nu }(t)\equiv 0$
unless $f_{\lambda \mu }^{\nu }(0)=c_{\lambda \mu }^{\nu }\neq 0$ where $%
c_{\lambda \mu }^{\nu }$ is the Littlewood-Richardson coefficient. Perform
the Littlewood-Richardson algorithm, one finds that for all partitions $\nu $
of length $\leq k$ the nonzero coefficients are%
\begin{equation*}
c_{\lambda \mu }^{(7,5,2)}=c_{\lambda \mu }^{(7,4,3)}=c_{\lambda \mu
}^{(6,6,2)}=c_{\lambda \mu }^{(6,4,4)}=c_{\lambda \mu }^{(5,5,4)}=1,\qquad
c_{\lambda \mu }^{(6,5,3)}=2\;.
\end{equation*}%
With the help of a computer one then calculates the following expansion
coefficients in the corresponding product of Hall-Littlewood polynomials,%
\begin{equation*}
f_{\lambda \mu }^{(7,5,2)}=f_{\lambda \mu }^{(7,4,3)}=1,\quad f_{\lambda \mu
}^{(6,6,2)}=f_{\lambda \mu }^{(6,4,4)}=f_{\lambda \mu }^{(5,5,4)}=1+t,\quad
f_{\lambda \mu }^{(6,5,3)}=2-t^{2}\;.
\end{equation*}%
Applying the straightening rules (\ref{straight3}) and (\ref{straight1}) one
finds%
\begin{gather*}
P_{(7,5,2)}=tP_{(3,2,1)},\quad P_{(7,4,3)}=[2]P_{(4,3,3)},\quad
P_{(6,6,2)}=[3]P_{(2,2,2)}, \\
P_{(6,5,3)}=P_{(3,2,1)},\quad P_{(6,4,4)}=P_{(4,4,2)},\quad
P_{(5,5,4)}=P_{(4,1,1)}\;.
\end{gather*}%
Thus, after removing all $n$-rows we arrive at the expansion%
\begin{multline*}
P_{(3,2,1)}P_{(4,3,1)}=(2+t-t^{2})P_{(3,2,1)}+(1+t)(1+t+t^{2})P_{(2,2,2)} \\
+(1+t)(P_{(3,3,0)}+P_{(1,1,0)}+P_{(2,0,0)})\;.
\end{multline*}%
From this computation we thus obtain the following $t$-deformed fusion
coefficients%
\begin{gather*}
N_{\lambda \mu }^{(3,2,1)}(t)=2+t-t^{2},\qquad N_{\lambda \mu
}^{(2,2,2)}(t)=(1+t)(1+t+t^{2}), \\
N_{\lambda \mu }^{(4,3,3)}(t)=N_{\lambda \mu }^{(4,1,1)}(t)=N_{\lambda \mu
}^{(4,4,2)}(t)=1+t\;.
\end{gather*}%
Setting $t=0$ one can verify, using the Verlinde formula or the Kac-Walton
algorithm, that one obtains the correct fusion coefficients of the $\widehat{%
\mathfrak{sl}}(n)_{k}$-Verlinde algebra.
\end{example}

So far we have focussed on the matrix elements $N_{\mu \nu }^{\lambda
}(t)=\langle \lambda |\boldsymbol{Q}_{\nu ^{\prime }}^{\prime }|\mu \rangle $%
. However, the definition (\ref{ncQ'S'}) expresses the latter in terms of
the matrix elements $K_{\nu ,\lambda /d/\mu }(t)=\langle \lambda |%
\boldsymbol{S}_{\nu ^{\prime }}^{\prime }|\mu \rangle $. Sample computations
of the latter lead to the following observation.

\begin{conjecture}
Let $\lambda ,\mu \in \mathcal{A}_{k,n}^{+}$ and $\nu \in \mathcal{\tilde{A}}%
_{k,n}^{+}$. The expansion coefficients (matrix elements) $\langle \lambda |%
\boldsymbol{S}_{\nu ^{\prime }}^{\prime }|\mu \rangle $ are always
polynomials with non-negative coefficients.
\end{conjecture}

This conjecture has been numerically verified for several examples.

\begin{example}
Choose $n=k=5$ and set $\mu =(3,2,2,1,1)$, $\nu =(4,3,1)$. Then the table
below lists all $\lambda \in \mathcal{A}_{k,n}^{+}$ for which the matrix
element $\langle \lambda |\boldsymbol{S}_{\nu ^{\prime }}^{\prime }|\mu
\rangle $ is nonvanishing.%
\begin{equation*}
\begin{tabular}{|c|c|}
\hline\hline
$\lambda $ & $\langle \lambda |\boldsymbol{S}_{\nu ^{\prime }}^{\prime }|\mu
\rangle $ \\ \hline\hline
\multicolumn{1}{|l|}{$4,2,2,2,2$} & \multicolumn{1}{|l|}{$%
1+3t+6t^{2}+8t^{3}+8t^{4}+6t^{5}+3t^{6}+t^{7}$} \\ \hline
\multicolumn{1}{|l|}{$4,3,2,2,1$} & \multicolumn{1}{|l|}{$%
2+9t+16t^{2}+14t^{3}+6t^{4}+t^{5}$} \\ \hline
\multicolumn{1}{|l|}{$4,3,3,1,1$} & \multicolumn{1}{|l|}{$%
1+6t+13t^{2}+14t^{3}+8t^{4}+2t^{5}$} \\ \hline
\multicolumn{1}{|l|}{$4,4,2,1,1$} & \multicolumn{1}{|l|}{$%
1+6t+11t^{2}+10t^{3}+3t^{4}$} \\ \hline
\multicolumn{1}{|l|}{$2,2,1,1,1$} & \multicolumn{1}{|l|}{$%
1+3t+4t^{2}+3t^{3}+t^{4}$} \\ \hline
\multicolumn{1}{|l|}{$3,1,1,1,1$} & \multicolumn{1}{|l|}{$%
1+2t+3t^{2}+3t^{3}+2t^{4}+t^{5}$} \\ \hline
\multicolumn{1}{|l|}{$5,4,3,3,2$} & \multicolumn{1}{|l|}{$%
1+7t+20t^{2}+31t^{3}+29t^{4}+17t^{5}+6t^{6}+t^{7}$} \\ \hline
\multicolumn{1}{|l|}{$5,4,4,2,2$} & \multicolumn{1}{|l|}{$%
1+6t+17t^{2}+24t^{3}+20t^{4}+9t^{5}+2t^{6}$} \\ \hline
\multicolumn{1}{|l|}{$5,2,2,2,1$} & \multicolumn{1}{|l|}{$%
2+6t+9t^{4}+7t^{3}+3t^{4}$} \\ \hline
\multicolumn{1}{|l|}{$5,3,2,1,1$} & \multicolumn{1}{|l|}{$3+8t+9t^{4}+3t^{3}$%
} \\ \hline
\multicolumn{1}{|l|}{$5,4,1,1,1$} & \multicolumn{1}{|l|}{$%
1+3t+4t^{2}+3t^{3}+t^{4}$} \\ \hline
\multicolumn{1}{|l|}{$5,5,3,2,2$} & \multicolumn{1}{|l|}{$%
2+8t+16t^{2}+17t^{3}+10t^{4}+3t^{5}$} \\ \hline
\multicolumn{1}{|l|}{$5,5,3,3,1$} & \multicolumn{1}{|l|}{$%
1+5t+10t^{2}+10t^{3}+5t^{4}+t^{5}$} \\ \hline
\multicolumn{1}{|l|}{$5,5,4,2,1$} & \multicolumn{1}{|l|}{$%
1+5t+8t^{2}+6t^{3}+2t^{4}$} \\ \hline
\end{tabular}%
\end{equation*}%
Note that the constant term for each listed polynomial coincides with the
fusion coefficient.
\end{example}

\begin{remark}
Let $\lambda ,\mu \in \mathcal{A}_{k,n}^{+}$ and $\nu \in \mathcal{\tilde{A}}%
_{k,n}^{+}$. Then for $dn=|\mu |+|\nu |-|\lambda |=0$ the matrix elements $%
N_{\mu \nu }^{\lambda }(t)=\langle \lambda |\boldsymbol{Q}_{\nu ^{\prime
}}^{\prime }|\mu \rangle $ and $K_{\nu ,\lambda /d/\mu }(t)=\langle \lambda |%
\boldsymbol{S}_{\nu ^{\prime }}^{\prime }|\mu \rangle $ specialise to the
known polynomials%
\begin{equation*}
f_{\mu \nu }^{\lambda }(t)=\langle Q_{\lambda /\mu },P_{\nu }\rangle
_{t}=\langle Q_{\lambda },P_{\mu }P_{\nu }\rangle _{t}
\end{equation*}%
and%
\begin{equation*}
K_{\nu ,\lambda /\mu }(t)=\sum_{\rho \in \mathcal{\tilde{A}}%
_{k,n}^{+}}K_{\nu \rho }(t)f_{\rho \mu }^{\lambda }(t)=\langle S_{\nu
},Q_{\lambda /\mu }\rangle _{t}=\langle Q_{\lambda },P_{\mu }s_{\nu }\rangle
_{t}\;,
\end{equation*}%
respectively. Here $S_{\lambda }$ is the dual Schur function with respect to
the Hall inner product; see (\ref{SS'}).
\end{remark}

\begin{example}
Let $\nu =(4,3,1,0,0)$, $\mu =(3,2,2,1,1)$, $\lambda =(5,5,3,2,2)$ then%
\begin{equation*}
f_{\mu \nu }^{\lambda }(t)=2+3t+t^{2}-t^{3}-t^{4}
\end{equation*}%
and%
\begin{equation*}
K_{\nu ,\lambda /\mu }(t)=2+8t+16t^{2}+17t^{3}+10t^{4}+3t^{5}
\end{equation*}

With the help of a computer one finds the following values for the Hall and
Kostka-Foulkes polynomials from the known formulae in \cite[Chapter III.6]%
{Macdonald},%
\begin{equation*}
\begin{tabular}{|c|c|c|}
\hline\hline
$\rho $ & $f_{\rho \mu }^{\lambda }(t)$ & $K_{\nu \rho }(t)$ \\ \hline\hline
\multicolumn{1}{|l|}{$4,3,1$} & \multicolumn{1}{|l|}{$2+3t+t^{2}-t^{3}-t^{4}$%
} & \multicolumn{1}{|l|}{$1$} \\ \hline
\multicolumn{1}{|l|}{$4,2,2$} & \multicolumn{1}{|l|}{$1+t$} &
\multicolumn{1}{|l|}{$t$} \\ \hline
\multicolumn{1}{|l|}{$4,2,1,1$} & \multicolumn{1}{|l|}{$1+2t+t^{2}$} &
\multicolumn{1}{|l|}{$t+t^{2}$} \\ \hline
\multicolumn{1}{|l|}{$4,1,1,1,1$} & \multicolumn{1}{|l|}{$0$} &
\multicolumn{1}{|l|}{$t^{3}+t^{4}+t^{5}$} \\ \hline
\multicolumn{1}{|l|}{$3,3,2$} & \multicolumn{1}{|l|}{$1+t-t^{2}-t^{3}$} &
\multicolumn{1}{|l|}{$t+t^{2}$} \\ \hline
\multicolumn{1}{|l|}{$3,3,1,1$} & \multicolumn{1}{|l|}{$2+2t-2t^{3}-t^{4}$}
& \multicolumn{1}{|l|}{$t+2t^{2}+t^{3}$} \\ \hline
\multicolumn{1}{|l|}{$3,2,2,1$} & \multicolumn{1}{|l|}{$1+2t+t^{2}$} &
\multicolumn{1}{|l|}{$2t^{2}+2t^{3}+t^{4}$} \\ \hline
\multicolumn{1}{|l|}{$3,2,1,1,1$} & \multicolumn{1}{|l|}{$1+t$} &
\multicolumn{1}{|l|}{$t^{2}+2t^{3}+3t^{4}+2t^{5}+t^{6}$} \\ \hline
\multicolumn{1}{|l|}{$2,2,2,2$} & \multicolumn{1}{|l|}{$0$} &
\multicolumn{1}{|l|}{$t^{3}+t^{4}+2t^{5}+2t^{6}+t^{7}$} \\ \hline
\multicolumn{1}{|l|}{$2,2,2,1,1$} & \multicolumn{1}{|l|}{$0$} &
\multicolumn{1}{|l|}{$t^{3}+3t^{4}+3t^{5}+3t^{6}+2t^{7}+t^{8}$} \\ \hline
\end{tabular}%
\end{equation*}
\end{example}

\section{Conclusions}

There has been recently a lot of attention on integrable quantum many-body
systems, such as the nonlinear quantum Schr\"{o}dinger model, in connection
with the infrared limit of four-dimensional supersymmetric $N=2$ gauge
theories \cite{NS}. It has been observed that the quantum Hamiltonians can
be identified with operators in the chiral ring, that is operators that are
annihilated by one chiral half of the supercharges. Moreover, it has been
asserted that the Bethe ansatz equations describe the vacua of the gauge
theory and that their solutions, the Bethe roots, have the interpretation of
coordinates on the moduli space of these vacua. It has been further argued
that the vacua should correspond to the states of a 2D topological quantum
field theory. The latter are known to be in correspondence with commutative
Frobenius algebras \cite{Kock}.

In light that the quantum integrable model investigated in this article is a
discrete version of the quantum nonlinear Schr\"{o}dinger model \cite%
{vanDiejen} our findings summarised in the following table confirm the
connection between quantum integrable systems and two-dimensional
topological quantum field theories.%
\begin{equation*}
\begin{tabular}{||l||c|c|c|}
\hline
\begin{tabular}{l}
{\small quantum} \\
{\small integrable model}%
\end{tabular}
&
\begin{tabular}{l}
{\small quantum} \\
{\small Hamiltonians}%
\end{tabular}
& {\small Bethe vectors} &
\begin{tabular}{l}
{\small Bethe ansatz} \\
{\small equations}%
\end{tabular}
\\ \hline
\begin{tabular}{l}
{\small Frobenius} \\
{\small algebra}%
\end{tabular}
& {\small generators} & {\small idempotents} &
\begin{tabular}{l}
{\small coordinate ring} \\
{\small presentation}%
\end{tabular}
\\ \hline
\end{tabular}%
\end{equation*}

\begin{center}
{\small Table 8.1. Dictionary between commutative Frobenius algebras and
quantum integrable models.}\bigskip
\end{center}

Moreover, our discussion highlights the central role of the eigenbasis of the
quantum Hamiltonians, the so-called Bethe vectors: they provide the algebra isomorphism
between the subalgebra in $\limfunc{End}S^{k}(V)$ generated by the deformed fusion
matrices $\boldsymbol{Q}_{\nu }^{\prime }$'s and the quotient of the spherical
Hall algebra. The latter is the coordinate ring we discussed in Section 7 and according
to the above correspondence it should have the interpretation of the moduli space of
vacua of a quantum field theory.

In this context we note that the Verlinde algebra can be
understood in terms of a purely topological construction using so-called
equivariant $K$-theory \cite{FTH}. Frobenius algebra deformations of the
Verlinde algebra have been suggested in \cite{Teleman} and it would be interesting to see if these constructions are related.

The discussion here is very much motivated along similar lines as
investigations of the so-called Bethe algebra for the Gaudin
and related models. There it has been shown that the Bethe algebra related
to Yangians describes the equivariant cohomology of flag manifolds \cite{RSTV}. In
contrast the discussion in \cite{KS, Korff, Korffproceedings} shows that the representations
of the Bethe algebra related to the quantum group $U_{q}\widehat{\mathfrak{gl}}(n)$ in the crystal limit,
 $q=0$,  are identical to the WZNW fusion ring and the small \emph{quantum}
cohomology ring of the Grassmannian. The present work extends this
discussion to the case when $q\neq 0$.\bigskip

\textbf{Acknowledgment.} The author is financially supported by a University
Research Fellowship of the Royal Society. In the course of this work the
author has benefitted from numerous discussions with colleagues and would
like to thank Kenneth Brown, Iain Gordon, Ulrich Kr\"{a}hmer, Junichi
Shiraishi and Robert Weston. Special thanks are due to Alastair Craw, for
his helpful explanations regarding affine varieties and Hilbert polynomials, Alain
Lascoux for pointing out reference \cite{Sanderson}, Anne Schilling and Nicolas
Thiery for sharing their insight and hospitality
at the Universit\'{e} Paris Sud in November 2008, Catharina Stroppel for sharing knowledge, ideas and hospitality at the University Bonn in December 2009.

\end{document}